\documentclass[11pt]{article}
\usepackage[margin=1in]{geometry}

\usepackage{bm}
\usepackage{todonotes}
\usepackage{lipsum}
\usepackage{setspace}
\usepackage{scrextend}
\usepackage{ragged2e}
\usepackage{microtype}
\usepackage{relsize}
\usepackage[noadjust]{cite} 
\usepackage{graphicx}
\usepackage{textcomp}
\usepackage{xcolor}
\usepackage{tgpagella}
\usepackage{mhsetup}
\usepackage{enumerate}
\usepackage{cuted}
\usepackage[linesnumbered,algoruled,boxed,lined]{algorithm2e}
\makeatletter
\newcommand{\removelatexerror}{\let\@latex@error\@gobble}
\makeatother

\usepackage{breakcites}
\usepackage{amsfonts,amssymb,amsmath,amsthm}
\usepackage{mathrsfs}
\usepackage{tabularx}
\usepackage{array}
\usepackage{complexity} 
\usepackage{datetime}
\usepackage{hhline}
\usepackage{multirow}
\usepackage{stmaryrd}
\usepackage{booktabs}
\usepackage{enumitem}
\usepackage{wrapfig}
\usepackage{subcaption}
\usepackage[T1]{fontenc}
\usepackage[utf8]{inputenc}
\usepackage{color}
\usepackage{cases}
\usepackage{xparse}
\usepackage{xargs}
\usepackage{algpseudocode}
\usepackage{verbatim, xspace}
\usepackage{tikz}
\usetikzlibrary{arrows.meta}
\usetikzlibrary{decorations.pathreplacing, calc}
\usepackage{mathtools}
\usepackage{thm-restate}
\usepackage[framemethod=TikZ]{mdframed}
\usepackage{pgfplots}
\pgfplotsset{compat=1.18}
\usepackage{xcolor}
\usepackage[hidelinks]{hyperref}
\usepackage[capitalise, noabbrev]{cleveref}

\mathchardef\mhyphen="2D

\newcommand{\PHES}{Pigeonhole Equal Subset Sum\xspace}
\newcommand{\PESS}{\text{PESS}\xspace}
\newcommand{\Phes}{\PESS}
\newcommand{\SIVP}{\text{SIVP}\xspace}

\newcommand{\cclass}[1]{\mathrm{#1}\xspace}
\newcommand{\eps}{\varepsilon}
\renewcommand{\Pr}{\mathbb{P}}

\newcommand{\ignore}[1]{}
\newcommand{\Oh}{\mathcal{O}}
\newcommand{\Os}{\mathcal{O}^\star}
\newcommand{\Ot}{\widetilde{\mathcal{O}}}

\newcommand\restr[2]{{
  \left.\kern-\nulldelimiterspace 
  #1 
  \littletaller 
  \right|_{#2} 
  }}

\addtocontents{toc}{\protect\setcounter{tocdepth}{2}}

\newtheorem{theorem}{Theorem}[section]

\newtheorem{definition}{Definition}
\newtheorem{lemma}{Lemma}[theorem]
\newtheorem{corollary}{Corollary}
\newtheorem{claim}{Claim}[theorem]
\newtheorem{proposition}{Proposition}[theorem]

\newtheorem{hypothesis}{Hypothesis}

\newtheorem{remark}{Remark}
\newtheorem{problem}{Open Question}

\newtheorem*{maingoal*}{Main Question}
\newenvironment{claimproof}[1]{\par\noindent\emph{Proof:}\space#1}{\hfill $\blacksquare$}

\crefname{hypothesis}{Hypothesis}{Hypotheses}
\Crefname{hypothesis}{Hypothesis}{Hypotheses}

\newcommand{\N}{\mathbb{N}}
\newcommand{\Z}{\mathbb{Z}}

\renewcommand{\R}{\mathbb{R}}

 \newcommand{\ceil}[1]{ \lceil #1 \rceil }
\newcommand{\floor}[1]{ \lfloor #1 \rfloor}

\newcommand{\Ball}{\mathrm{Ball}}
\newcommand{\mat}[1]{\mathbf{#1}}

\newcommand{\svp}{\mathsf{SVP}}

\newcommand{\norm}[1]{||#1||}

\newcommand{\xor}{\oplus}

\newcommand{\etal}{et al.\xspace}

\renewcommand{\L}{\mathcal{L}}

\def\floor#1{\left\lfloor #1 \right\rfloor}
\def\ceil#1{\left\lceil #1 \right\rceil}

\def\set#1{\left\{ #1 \right\}}

\def\norm#1{\left\| #1 \right\|}

\newcommand{\Aa}{\mathcal{A}}
\newcommand{\calB}{\mathcal{B}}
\newcommand{\Bb}{\calB}

\newcommand{\Cc}{\mathcal{C}}

\newcommand{\Ll}{\mathcal{L}}

\renewcommand{\vec}[1]{\ensuremath{\boldsymbol{\mathrm{#1}}}}

\newcommand{\w}{\Sigma}
\newcommand{\BigO}{\mathcal{O}}
\renewcommand{\poly}{\mathrm{poly}}

\newcommand{\pdc}[1]{\textcolor{red!50!black}{PD: #1}}

\newcommand{\rank}{\mathsf{rank}}

\SetCommentSty{commentfont}
\DontPrintSemicolon
\usepackage{titlesec}
\titlespacing*{\paragraph}{0pt}{0.5\baselineskip}{0.8em}
\usepackage[normalem]{ulem}

\makeatletter
\renewenvironment{strip}{\par\noindent\begin{minipage}{\linewidth}}{\end{minipage}\par\medskip}
\makeatother

\def\BibTeX{{\rm B\kern-.05em{\sc i\kern-.025em b}\kern-.08em
    T\kern-.1667em\lower.7ex\hbox{E}\kern-.125emX}}
\begin{document}

\title{Pigeonhole Equal Subset Sum: Subexponential Algorithm, 
Tight Lower Bounds, and
Average-Case Analysis}

\author{
Deepak Bhati\\
\small CISPA Helmholtz, Saarbr\"ucken, Germany\\
\small deepak2602161@gmail.com
\and
Pranjal Dutta\thanks{Nanyang Technological University, Singapore. Supported by the Nanyang Assistant Professorship (NAP) 
and an SUG grant (\#025774-00001), both funded by Nanyang Technological University.}\\
\small Nanyang Technological University, Singapore\\
\small pranjal.dutta@ntu.edu.sg
\and
Antoine Joux\\
\small CISPA Helmholtz, Saarbr\"ucken, Germany\\
\small joux@cispa.de
\and
Mahesh Sreekumar Rajasree\thanks{Work done while at CISPA Helmholtz, Germany. Supported by the European Union (ERC, LACONIC, 101041207). 
Views and opinions expressed are however those of the author(s) only and do not necessarily reflect 
those of the European Union or the European Research Council. Neither the European Union nor the 
granting authority can be held responsible for them.}\\
\small University of Warsaw, Warsaw\\
\small srmahesh1994@gmail.com
\and
Karol W\k{e}grzycki\thanks{Max Planck Institute for Informatics, Germany. 
Supported by the Deutsche Forschungsgemeinschaft (DFG, German Research Foundation) grant number 559177164.}\\
\small Max Planck Institute for Informatics, Saarbr\"ucken, Germany\\
\small kwegrzyc@mpi-inf.mpg.de
}

\date{}
\maketitle

\begin{abstract}
% ===== begin abstract.tex =====
The \textsc{Equal Subset Sum} problem asks, given $n$ positive integers bounded by $M$, to find two distinct non-empty subsets with the same sum. We study this problem in the regime $M = 2^{o(n)}$, where the pigeonhole principle guarantees that a solution exists, and the main challenge is to find one efficiently. A straightforward dynamic-programming algorithm runs in time $\poly(n)\cdot M$, and after a sequence of improvements, the current best algorithm, due to Jin, Williams, and Zhang (ESA 2025), achieves running time $\poly(n)\cdot M^{1/3}$.

In this work, we give the first $\poly(n)\cdot M^{o(1)}$-time algorithm for \textsc{Equal Subset Sum} in the regime $M = 2^{o(n)}$. We complement this algorithmic result with a tight conditional lower bound showing that any substantial improvement would require a major breakthrough in approximation algorithms for lattice problems, building on recent techniques of Vafa and Vaikuntanathan (STOC 2025).

Finally, we show that our techniques yield an asymptotically faster algorithm for the regime when
$M = \Theta(2^{n})$: a $\poly(n) \cdot 2^{n/4}$-time algorithm 
in the average case.
% ===== end abstract.tex =====
\end{abstract}

% ===== begin intro.tex =====
\section{Introduction}

Equal Subset Sum is a fundamental problem in combinatorial optimization. The
problem can be stated as follows: given a set $S$ of $n$ positive integers bounded by $M$, find two distinct subsets $A, B
\subseteq S$ such that $\Sigma(A) = \Sigma(B)$, where $\Sigma(\cdot)$ denotes the sum of the integers in the set. A standard
meet‑in‑the‑middle algorithm solves the problem in $\Os(3^{n/2})$, which can be
improved to $\Os(1.7067^n)$~\cite{MNPW19}. 
\footnote{$\Os(\cdot)$ and $\Ot(\cdot)$ hide polynomial and polylogarithmic factors in the input respectively. }

Papadimitriou~\cite{Pap94} introduced the total variant called the {\em Pigeonhole
Equal Subset Sum} (\Phes). In this variant, the sum of the numbers is less
than $2^{n}-1$, and the goal remains to find distinct subsets $A, B \subseteq
S$ with equal sums. Since there are $2^n - 1$ non-empty subsets but only $2^n-2$ possible
nonzero sums, a collision must exist by the pigeonhole principle. \Phes belongs
to the class $\cclass{TFNP}$ and, more specifically, to $\cclass{PPP}$, and it
is conjectured to be $\cclass{PPP}$-complete~\cite{Pap94}.  The setting, when
the number of collisions is even larger, e.g., the sum of integers is bounded by $2^{n-1}$, is
called \emph{weak-pigeonhole regime} and is also an important setting~\cite{SZZ18}.
In that weaker setting, \PESS belongs to the class
$\cclass{PWPP}$ (problems whose totality is proven through the weak pigeonhole
principle)~\cite{SZZ18}. For convenience, we will consider the variants where each of the numbers is bounded by either $(2^{n}-1)/n$ or $2^{n-1}/n$, so the sum remains bounded by $2^n-1$ or $2^{n-1}$. To this effect, we introduce a parameter $M$ that denotes the bound on the integers we allow in our instances.

To the best of our knowledge, the first explicit algorithm for \Phes was given
by Vasilakis~\cite{vasilakis2022improved}, who obtained an $\Ot(nM)$
time dynamic programming algorithm\footnote{More precisely, algorithm of Vasilakis~\cite{vasilakis2022improved} runs in $\Oh(n+\Sigma(S) \log(\Sigma(S)))$ time, where $\Sigma(S)$ is the sum of the input integers.}.
Allcock, Hamoudi, Joux, Klingelhöfer and
Santha~\cite{AHJ+21} gave an $\Os(2^{n/2})$-time algorithm when
the sum of the numbers is less than $2^{n}-1$.
Notice that by discarding the largest
items until we reach a modified instance with $n'$ items whose sum is in $[2^{n'-1}-1,2^{n'}-1)$ and then running the algorithm of Allcock et al.~\cite{AHJ+21}, the overall running-time  
can be presented as $\Os(M^{1/2})$. Later, Jin and Wu~\cite{JW24} refined the meet‑in‑the‑middle approach of~\cite{AHJ+21} and achieved
$\Os(M^{0.4})$ running time. The current best algorithm by Jin,
Williams, and Zhang (ESA 2025) runs in $\Os(M^{1/3})$~\cite{JWZ25}
(see~\cref{tab:prev-our-results} for an overview). At the end of their paper,
Jin, Williams, and Zhang raised the question of whether the $M^{1/3}$ barrier
can be improved~\cite{JWZ25}.

\begin{problem}[\cite{JWZ25}]\label{prob:faster}
    Can we get a $\poly(n)\cdot M^{1/3-\eps}$-time algorithm for \PHES for some $\eps > 0$?
\end{problem}

Note that \Phes is a total problem and, as such, is not expected to be NP-hard. As noted above, it lies in $\cclass{PPP}$. Inspired by these questions, Jin~\cite{jin-slides} asked whether subexponential-time algorithms for \Phes are possible:

\begin{problem}[\cite{jin-slides,JW24}]\label{prob:subexponential}
    Can we get a subexponential-time algorithm for \PHES?
\end{problem}

More specifically, Jin~\cite{jin-slides} asked for $2^{o(n)}$-time algorithm,
which in parameterization by maximum integer can be expressed as $\poly(n) \cdot M^{o(1)}$. 

\subsection{Our Results}
\label{sec:our-results}
In this work, we make partial progress on both of these questions.
First, we present a $\text{poly}(n) \cdot M^{o(1)}$~-time algorithm for \Phes in the regime $M = 2^{o(n)}$.

\begin{restatable}{theorem}{thmsubexponential}\label{thm:subexponential}
    There exists a deterministic algorithm that, given an integer $n$ and a set
    $S\subseteq\{1,\dots,M\}$ of size $n$ such that $nM < 2^{n}-1$, runs in time
    $\poly(n)\cdot M^{\Oh(1/\log{\delta})}$ and returns two distinct subsets of
    $S$ with the same sum, where $\delta = \delta(n,M)\coloneqq n/\log{M}$ is the density of
    the instance.

    As a corollary, when $M=2^{o(n)}$, the algorithm runs in time $\poly(n) \cdot M^{o(1)}$.
\end{restatable}

Although the algorithm in~\cref{thm:subexponential} does not fully 
resolve Problem~\ref{prob:subexponential}, we believe
it provides significant progress on that question. Notice, that~\cref{thm:subexponential} also does not resolve Problem~\ref{prob:faster}, however, more precise running time
analysis of~\cref{thm:subexponential} shows that it improves the $\Os(M^{1/3})$
time algorithm~\cite{JWZ25} when $\delta > 19$
(see~\cref{remark:exact-runtime} for an exact runtime). Later
in~\cref{thm:average-case-integer} we show that techniques
behind~\cref{thm:subexponential} can be used to improve state-of-the-art for
$\delta \ge 1$ in the average-case regime.

One could ask whether the $\Oh(1/\log{\delta})$ term in the exponent can be
significantly improved; for example, can we get $2^{\Oh(n^{0.99})}$ time, or at
least $2^{\Oh(n/\log^{100}{n})}$ time? We show that this is highly unlikely.

\begin{restatable}{theorem}{thmlowerbound} \label{thm:lowerbound}
    Under~\cref{hyp:lambda-n},
    there is no $2^{o(n/\log{n})}$ time algorithm that solves Pigeonhole Equal Subset Sum. Moreover, there is no $\poly(n) \cdot M^{o(1/\log{n})}$ time algorithm for \PESS (even for $\delta = n^{1/2+o(1)}$).
\end{restatable}

This lower bound implies that the running time
in~\cref{thm:subexponential} cannot be improved to $M^{o(1/\log{\delta})}$, and
that our algorithm is optimal for polynomial $\delta$ up to constants in the exponent. We also show a lower bound that
excludes a $2^{o(n/\log{n})}$ time algorithm for \PESS, which conditionally
completes our answer to Problem~\ref{prob:subexponential} (up to a single $\log{n}$ factor in
the exponent). See \cref{fig:overview} for an overview of our results.

\cref{thm:lowerbound} is conditioned on the hardness of lattice problems. 
Shortest Vector Problem is a fundamental hardness assumption.  For a tunable
parameter $\kappa \ge 1$ and dimension $n$, one can obtain an
$n^{\kappa}$-approximation in $2^{\Oh(n/\kappa)}$
time~\cite{schnorr1987hierarchy} and after a substantial effort over the past 40
years (see,
e.g.,~\cite{schnorr1987hierarchy,schnorr1994lattice,GN08,ALNS20}),
no $2^{o(n)}$ algorithm is known that would yield polynomial approximation for
SVP. Significant improvements over these trade-offs would break modern
cryptosystems and have been already assumed by~Vafa--Vaikuntanathan~\cite{VV25} to establish hardness of \emph{symmetric binary
perceptron, number partitioning problem} and more. 
Here, we consider a stronger, fine-grained version of their assumption~\footnote{Formally,~\cref{hyp:lambda-n} concerns an important
variant of the problem called $n$-th Successive Minima Problem.}:
\begin{restatable}{hypothesis}{hypmain} \label{hyp:lambda-n}
    For any constant $c \ge 1$, there exists $\eps > 0$ such that there is no
    $\Oh(2^{\eps n})$ time algorithm that given an $n$-dimensional lattice $\Ll$
    with high probability returns a non-zero vector $\mathbf{v} \in \Ll$ such that
    $\norm{\mathbf{v}}_{2} \le \Oh(n^c) \cdot \lambda_n(\L)$.
\end{restatable}
We show that lack of $2^{o(n)}$-time polynomial approximation for this problem is a major obstacle 
to improve our $M^{\Oh(1/\log{\delta})}$-time algorithm for \PESS.
We stress that~\cref{hyp:lambda-n} is not a standard complexity assumption, 
and to the best of
our knowledge, we are the first ones to use it in that form. Regardless of whether one believes in~\cref{hyp:lambda-n} or not, our message is as follows: 

\begin{center}
\emph{To improve the running time of~\cref{thm:subexponential}, one would need first to make\\ a substantial breakthrough for lattice problems.}
\end{center}

Finally, we focus on the case $M \le 2^{n}$. \footnote{The algorithm in the average case works even for density less than $1$, 
when $M = \Theta(n^{\kappa} \cdot 2^n)$ for some $\kappa > 0$; the only caveat is that the success 
probability decreases with increasing $\kappa$.}
The average‑case setting is an important relaxation in which the
algorithm is required to work correctly on most inputs rather than on every
input. This perspective is relevant for both practical and theoretical reasons.
For example, average‑case algorithms for Subset Sum have been widely studied
and play an important role in cryptography, since such instances were long
considered as hardness assumptions~\cite{crypto-merkle78,shamir84}. Understanding
the average‑case complexity of Subset Sum, therefore, matters for both practice
and theory~\cite{LO85, HJ10, BCJ11,low-density-subset-sum,low-density-subset-sum2,dinur}.
Motivated by this, we present an $\Os(M^{1/4})$-time algorithm for 
equal Subset Sum in the average‑case setting when $M \le 2^n$.

\begin{theorem}\label{thm:average-case-integer}
    Let $M = \Theta\left(2^n\right)$. There exists a deterministic algorithm that given a set $A$ of $n$ random integers sampled uniformly from the set $\{1,\ldots,M\}$, finds distinct $X,Y \subseteq A$ such that $\Sigma(X) = \Sigma(Y)$ with probability 
    $\Omega\left(n^{-{2}}\right)$ in $\Tilde{\Oh}(2^{n/4})$ time. 
\end{theorem}

Note that this result holds in the \emph{average case}, whereas previous work
guaranteed solutions in the worst case. The algorithm is a modification of the algorithm behind~\cref{thm:subexponential}; however, the analysis is much more involved. 
Curiously, we can show that our algorithm also works for the \emph{modular version} of Equal
Subset Sum. In that
setting, the best algorithm is due to Allcock et al.~\cite{AHJ+21}, who gave an
$\Os(M^{1/2})$-time algorithm. In fact, obtaining a faster algorithm in the
modular setting is posed as an open question in~\cite{jin-slides}.

\begin{theorem}\label{thm:average-case-modular}
    Let the modulus $M = \Theta\left(2^n\right)$, where $M = pqr$,
    $p,q = \Theta\left(M^{1/4}\right)$ and $p$, $q$, and $r$
    are pairwise co-prime integers. 
    There exists a deterministic algorithm that given a set $A$ of $n$ random modular integers sampled uniformly from $\mathbb{Z}_M$, finds distinct $X,Y \subseteq A$ such that $\Sigma(X) \equiv_M \Sigma(Y)$ with probability $\Omega\left(1\right)$ in $\Tilde{\Oh}(2^{n/4})$ time.
\end{theorem}

Notice that in this setting, we require $M$ to be a product of three pairwise co-prime integers, 
but the probability of success is higher. If the modulus is not a product of such integers, 
then using \cref{thm:average-case-integer}, trivially we can
solve the modular problem with success probability $\Omega(n^{-2})$. 
A faster average-case algorithm often leads to improved worst-case algorithms. 
This happened for Subset Balancing, a generalization of Equal Subset Sum, 
where recent average-case results~\cite{Chen22} motivated improved worst-case 
algorithms~\cite{DBLP:journals/corr/abs-2511-10823}. Because our analysis is already quite technical in the average-case setting, 
we hope it will eventually yield worst-case algorithms with comparable running time. 
Note that the algorithm behind~\cref{thm:subexponential} is deterministic, albeit less efficient for small densities. In the following
figure, we summarize all our results with the related previous work. 

\begin{table}[h!]
    \centering
    \begin{tabular}{lll}
        \toprule
        Runtime & Reference & Comment\\
        \bottomrule
        $\Os(M)$ & \cite{vasilakis2022improved} & \\ \hline
        $\Os(M^{1/2})$ & \cite{AHJ+21}& \\ \hline
        $\Os(M^{0.4})$ & \cite{JW24} & \\ \hline
        $\Os(M^{1/3})$ & \cite{JWZ25} &\\ \hline
        $\Os(M^{o(1)})$  & \cref{thm:subexponential} & $M=2^{o(n)}$\\ \hline
        No $M^{o(1/\log{n})}$ & \cref{thm:lowerbound} & Under~\cref{hyp:lambda-n}\\ \hline
        $\Os(M^{1/4})$ &  \cref{thm:average-case-integer} & Average-case for $M \leq 2^{n}$ \\ \hline
        $\Os(M^{1/4})$ & \cref{thm:average-case-modular} & Average-case $\mathbb{Z}_M$ for 
        $M=pqr \leq 2^n$\\ \hline
    \end{tabular}
    \caption{Summary of prior and our results for Equal Subset Sum in the pigeonhole regime. Notice that \cite{vasilakis2022improved, AHJ+21, JW24, JWZ25} also work in the regime when the sum is less than $2^{n}-1$.}
    \label{tab:prev-our-results}
\end{table}

\section{Our Techniques}
\subsection{Technique behind~\cref{thm:subexponential}}

\paragraph{Ingredient 1: Dynamic Programming.} First, there is a straightforward $\poly(n)\cdot M$ algorithm for Equal Subset Sum: enumerate subsets of $[n]$ and stop as soon as a collision is found. Since a collision must occur within the first $M$ nonzero sums, this procedure runs in $\poly(n)\cdot M$ time. Vasilakis~\cite{vasilakis2022improved} refined this idea to obtain an $\Ot(nM)$-time algorithm. If the subsets are inspected in a suitable random order, the birthday paradox implies that a collision typically appears after about $\poly(n)\cdot M^{1/2}$ steps, recovering the approach of Allcock et al.~\cite{AHJ+21}. Jin and Wu~\cite{JW24}, and later Jin, Williams, and Zhang~\cite{JWZ25}, further improve this line of attack by combining careful preprocessing with structural observations about subset-sum collisions to reach $\poly(n)\cdot M^{1/3}$ running time.

\paragraph*{Ingredient 2: Differencing Algorithm.} Next, we reformulate the
classical differencing technique of Karmarkar-Karp~\cite{KK82}. Although
originally designed for partitioning in the average-case setting, the differencing strategy can be adapted
to the Equal Subset Sum setting even in the worst-case setting.  It gives a 
polynomial-time
algorithm when integers are bounded by $M = 2^{c\log^2 n}$ for a small constant $c>0$. We reinterpret this procedure as follows:

Let $0 < a_1 < \ldots < a_n$ be the sorted input numbers and, for simplicity, assume
that $n$ is a power of $2$. For every $i \in [n/2]$, define $\Delta_i \coloneqq
a_{2i} - a_{2i-1}$, the difference between the two consecutive numbers in the
$i$-th pair (note that $\Delta_i \ge 0$). We call $i \in [n/2]$ \emph{good} if
$\Delta_i \le 4M/n$, and \emph{bad} otherwise.

Observe that there are at least $n/4$ good pairs. Indeed, if there were more than
$n/4$ bad indices, then $\sum_{i \in \text{bad}} \Delta_i > |\text{bad}| \cdot
\frac{4M}{n} > M$, which contradicts the fact that $\sum_{i \in [n/2]} \Delta_i
< a_n \le M$.

Now observe the following: if we discard bad indices, we obtain at least $n/4$ numbers whose
values decreased by a factor $n/4$. If we repeat this for $\kappa =
\Theta(\log n)$ rounds, we get a number whose value is bounded by $\le
M \cdot (\frac{4}{n})^\kappa < 1 $ for sufficiently small constant $c \in (0,1/4]$.
~\footnote{For simplicity, we ignored the fact that $n$ decreases in each iteration! 
As we will see in the proof of~\cref{thm:subexponential} 
this turns out to be a low-order factor.} Therefore, eventually, we must observe integer zero, because throughout the
procedure, the numbers are nonnegative integers.

The key combinatorial observation is that the sign chosen in the last $\Delta_i$
determines to which subset the corresponding element $a_i$ belongs in
the final solution. Hence, an integer zero naturally corresponds to two different
subsets of the same sum.

\paragraph*{Combining the ingredients.}

Karmarkar--Karp differencing approach already has several applications in
cryptology.  Importantly for us, Wagner~\cite{wagner2002generalized} extended
it to solve \emph{the $k$-list problem} in the average-case setting (see also
an independent work of~\cite{blum2003noise}). Similarly to these works, we also
combine dynamic programming with the Karmarkar--Karp differencing approach.
Somehow surprisingly, in comparison to these approaches, our algorithm is deterministic and
works in the worst case.

Let $k$ be a largest integer such that $M \le k^{n/(4k)}$. For simplicity, let us assume
that $k$ is a power of $2$ and $k$ divides $n$ and let $W \coloneqq 2^{n/(2k)}$
(in \cref{sec:subexponential} we give tighter bounds). We start by partitioning
the input into $k$ equal blocks and precompute the list of all subset sums
inside each block to obtain sets $\mathcal{A}_i^{(0)}$ for each $i \in [k]$.  Rather than matching high‑order bits (as in the differencing
approach), we repeatedly merge blocks and keep large sumsets that share the
same low‑order bits, i.e., we work modulo $W^j$ in level $j$. 

More precisely, at each level $j = 1,2,\ldots,\log{k}$ we maintain sets
$\mathcal{A}_i^{(j)}$ for every $i \in [k/2^j]$. The sets $\mathcal{A}_{i}^{(j)}$ 
arise by \emph{merging} two sets $\mathcal{A}_{2i-1}^{(j-1)}$ and $\mathcal{A}_{2i}^{(j-1)}$ from the previous levels. Merging
is performed by taking the sumset $\mathcal{A}_{2i-1}^{(j-1)} +
\mathcal{A}_{2i}^{(j-1)}$.\footnote{We use $X+Y \coloneqq \{x + y \mid x \in X,\,y \in Y\}$ to denote a sumset.}
Notice that the size of that set increases; therefore, we need to \emph{prune} it.
We achieve it by partitioning the sumset into $W$ sets of integers based on their value modulo $W^j$ and keeping a representative subset of size exactly $W$ from the largest set. This pruning keeps the lists to have size exactly $W$ while maintaining a
large ensemble of sums whose lowest bits agree (see~\cref{alg:1}).

\begin{algorithm}[ht!]
\caption{Algorithm behind~\cref{thm:subexponential}.}
\label{alg:1}
Let $M \le k^{n/(4k)}$ and $W \coloneqq 2^{n/(2k)}$

Partition $S$ arbitrarily into $S_1 \uplus \dots \uplus S_k$ with $|S_i|=n/k$

$\Aa_i^{(0)} \gets \{ \Sigma(X) \mid X \subseteq S_i \}$ for every  $i \in [k]$ and 
return collision if found

$p_i^{(0)} \gets 0$ for every $i \in [k]$.

\For{$j = 1,\ldots, \log k$ and $i = 1,\ldots,k/2^j$}{
    $\mathcal{B} \gets \Aa_{2i-1}^{(j-1)} + \Aa_{2i}^{(j-1)}$ and $b \gets p_{2i-1}^{(j-1)} + p_{2i}^{(j-1)}$
    and return a collision if found
    
    $\Ll_{i,p}^{(j)} \gets \{ a \in \mathcal{B} \mid a \equiv p \cdot W^{j-1} + b \pmod{W^j} \}$ 
    for all $0 \le p < W$\\
    $p^\star \gets \arg\max_{0 \le p < W} |\Ll_{i,p}^{(j)}|$

    $p_i^{(j)} \gets p^\star \cdot W^{j-1} + b$

    $\Aa_i^{(j)}$ is arbitrary subset of $\Ll_{i,p^\star}^{(j)}$ of size $W$
}
\end{algorithm}

In each level $j$, we have $k/2^j$ pairwise disjoint sets of size at most $W$.
Transition in each level can be computed in quadratic time; hence the total
time is $\poly(n) \cdot W^{\Oh(1)} = 2^{\Oh(n/k)} = M^{\Oh(1/\log{k})} = M^{o(1)}$ if $k$ is allowed to be large enough. 

For correctness, we claim that we maintain the following invariants: if no
collision was detected until level $j$, then each $\mathcal{A}_i^{(j)}$ has size
$W$ and consists of subset sums that are all congruent to some value $p_i^{(j)}$
modulo $W^j$ . Hence, if no collision was detected, the last set
$\Aa_i^{(\log{k})}$ contains two different numbers that are congruent modulo 
$W^{\log{k}}$. By the choice of $W$ and $k$ this means that $W^{\log{k}} > nM$,
which implies equality as integers, contradicting the invariant that
$|A^{(j)}_1| = W \ge 2$. Therefore, at some level $j
\le \log{k}$, the algorithm detects a collision! We include the full proof of~\cref{thm:subexponential} in~\cref{sec:subexponential}. 

\subsection{Techniques behind~\cref{thm:lowerbound}}
Very recently, the connection between the Karmarkar-Karp differencing method and lattice problems was highlighted by Vafa and Vaikuntanathan~\cite{VV25}. Their work studies an average-case version of the \emph{Number Partitioning Problem} (NPP), where one is given a random vector $\mathbf{a} = (a_1,\dots,a_n) \sim \mathcal{N}(\mathbf{0},\mathbf{I}_n)$ and the goal is to find a vector $\mathbf{x} \in \{-1,1\}^n$ such that
$\left| \mathbf{a}^{\top}\mathbf{x} \right|
    \leq
    \kappa(n)\sqrt{n}.$
    
They prove a worst-case-to-average-case reduction from the lattice problem \emph{Incremental Guaranteed Distance Decoding} (IncGDD) to NPP. In particular, they show that an average-case polynomial-time algorithm for NPP achieving $\kappa(n) = 2^{-\Omega(\log^{3+\varepsilon} n)}$ would contradict their assumed subexponential hardness of approximate lattice problems~\cite{VV25} including the Shortest Independent Vector Problem (SIVP) (see~\cref{hyp:weak-sivp}). Thus, under these assumptions, Karmarkar-Karp is nearly optimal in polynomial time.

They also note that their reduction extends to related variants, including the \emph{Number Balancing Problem} (NBP) where the Gaussian distribution is replaced with the uniform distribution on $[0,1]$ and solution vectors are $\{-1,0,1\}^n \setminus \{\mathbf{0}\}$~\cite{VV25}. At a high level, one may view NBP as closely related to an equal-sum type problem after an appropriate discretization of the input numbers. In particular, if the target error is $\kappa(n)$, then the relevant integer range after scaling is roughly of magnitude $1/\kappa(n)$. This suggests that the hardness result of~\cite{VV25} is consistent with the absence of polynomial-time algorithms for average-case \PESS when $M \approx \frac{1}{\kappa(n)} = 2^{\Omega(\log^{3+\varepsilon} n)}$.

Consequently, this also provides evidence against polynomial-time algorithms for the corresponding worst-case regime. However, these results still leave open the possibility of quasipolynomial-time, or even subexponential-time 
(see Problem~\ref{prob:subexponential}), algorithms for worst-case \PESS when $M \geq 2^{\log^{2+\varepsilon}(n)}$.

\paragraph*{Technique behind our lower bound.} Inspired by the above, we show a
connection between lattice problems and \PESS. We achieve a tight lower bound,
but in comparison to~\cite{VV25}: (i) we work only in the
worst-case setting, (ii) we only show lower bounds against subexponential time algorithms, and
(iii) \cref{hyp:lambda-n} is much stronger than the assumptions in~\cite{VV25}
(see discussion in~\cref{sec:prelim}).

Our reduction is from \emph{$n^{th}$-Successive Minima Problem} (nSMP)
(see~\cref{def:nSMP}). The task in nSMP\footnote{There is a well-studied lattice
problem called the \emph{Successive Minima Problem} (SMP)~\cite{Mic08}, which
asks for $n$ linearly independent lattice vectors
$\mathbf{v}_1,\dots,\mathbf{v}_n$ such that $\norm{\mathbf{v}_i}=\lambda_i(\L)$
for every $i \in [n]$.} is to find a nonzero lattice vector $\mathbf{v}\in \L$
such that $\norm{\mathbf{v}}_2 \leq \lambda_n(\L)$. By
\cref{thm:SIVP-SVP-variant}, 
this formulation is equivalent (up to an irrelevant to us loss
in approximation), to the approximate \SIVP.

At a high level, the reduction is motivated by the classical sieve of Ajtai,
Kumar, and Sivakumar~\cite{AKS01} for obtaining a short vector. However, our
proof does not literally simulate a sieve. Instead, it samples many lattice
vectors, encodes them as positive integers, invokes the assumed black-box \PESS
algorithm on the resulting instance, and decodes the returned equal-sum
collision into a short lattice vector. Thus, \PESS plays the role of a
collision-finding primitive analogous in spirit to the sieve step.

More concretely, on input we are given an $n$-dimensional lattice $\L$ and we
will use hypothetical $M^{o(1/\log{N})}$-time algorithm that solves \PESS on
instances with $N = \Theta(n^2)$ integers bounded by $M = 2^{\Oh(n\log{n})}$.  We
start by preprocessing an $n$-dimensional lattice $\L$. First, we compute a short
basis $\mathbf{D}$.  If we use the state-of-the-art algorithms (e.g.,
\cite{schnorr1987hierarchy,ALNS20}) in time $\Oh(2^{\eps n})$ we can
compute a basis of length $B \cdot \lambda_n(\Ll)$, where $B =
n^{\Theta(1/\eps)}$ for any tunable $\eps > 0$.
Next, we
sample independent lattice points $\mathbf{x}^{(1)},\dots,\mathbf{x}^{(N)} \sim D_{\L,s}$ using the (good) basis $\mathbf{D}$ generated in the first step, for a suitable Gaussian parameter $s = \Theta(\sqrt{n}\,N \norm{\mathbf{D}}_2)$. Standard tail bounds imply that, except with negligible probability, all sampled vectors satisfy
\[
    \norm{\mathbf{x}^{(i)}}_{\infty} \leq BN^2\lambda_n(\L)
    \qquad\text{for every } i\in[N].
\]

We then encode each sampled vector as a positive integer using a base-$W$ representation, where $W = \Theta(BN^3)$. Schematically, the encoding has the form
\[
    a_i = W^n + \sum_{j=1}^n W^{j-1}
    \left\lfloor
        \frac{(\mathbf{x}^{(i)})_j/\lambda_n(\L)}{r_j}
    \right\rfloor,
\]
where each $r_j$ is a random prime in $[N^{3}]$; in the actual proof, we also add independent random shifts before rounding to ensure all integers are different with high probability. The additive term $W^n$ ensures positivity, while the choice of $W$ is large enough to prevent carries between coordinates. Strictly speaking, the reduction does not know $\lambda_n(\L)$ in advance. Instead, it guesses this quantity up to a factor of $2$ by enumerating powers of two, and for the correct guess obtains $\widetilde{\lambda}_n$ satisfying $\lambda_n(\L) \leq \widetilde{\lambda}_n \leq 2\lambda_n(\L)$ (similarly to Ajtai et al.~\cite{AKS01}). Accordingly, the encoding uses $\widetilde{\lambda}_n$ in place of $\lambda_n(\L)$.

Since each coordinate contributes only $\Oh(BN^2)$ to a digit, every encoded integer is at most $M = (BN)^{\Oh(n)}$. For our choice of parameters, this satisfies $M \leq 2^{o(N)}$, so the resulting instance lies in the subexponential input range covered by the assumed \PESS algorithm.

Now suppose the \PESS oracle returns two distinct subsets $X,Y \subseteq [N]$ such that $\sum_{i\in X} a_i = \sum_{j\in Y} a_j$. From this collision, the reduction outputs the lattice vector $\mathbf{v} \;=\; \sum_{i\in X} \mathbf{x}^{(i)} - \sum_{j\in Y} \mathbf{x}^{(j)}$. Because the base $W$ is much larger than the magnitude of every summed digit (i.e., there can be at most $2N$ sums of $\Oh(BN^2)$ digits), equality of the encoded sums implies that there are no carries between digit positions. Hence the corresponding rounded quantities cancel coordinate-wise, and the only loss comes from the rounding error. If we plug in the parameters, this yields $\norm{\mathbf{v}}_2 \leq N^5 \cdot \lambda_n(\L)$. Thus, any successful \PESS call gives a sufficiently short lattice vector.

The main technical issue is to guarantee that the decoded vector is \emph{nonzero}. To address this, we analyze a modified experiment (similarly to the tossing argument~\cite{AKS01, Reg04}). Let $\mathbf{v}^{\ast}\in\L$ be a shortest nonzero vector. Instead of sampling only from $D_{\L,s}$, we first sample $\mathbf{x}^{\prime(1)},\dots,\mathbf{x}^{\prime(N)} \sim D_{\L,s}$ and then independently choose bits $b_1,\dots,b_N \in \{0,1\}$, replacing each sample by $\mathbf{x}^{(i)} = \mathbf{x}^{\prime(i)} + b_i \mathbf{v}^{\ast}$.

Since $s$ is large, the distributions $D_{\L,s}$ and $D_{\L,s}+\mathbf{v}^{\ast}$ are statistically close, so the original and modified sampling procedures are nearly indistinguishable. On the other hand, the random primes (and the random shifts) in the encoding make it unlikely that adding $\mathbf{v}^{\ast}$ changes any encoded integer. Therefore, with high probability, the \PESS oracle sees the same instance in both experiments. However, if we flip one bit $b_t$, then the decoded lattice output changes by $\pm \mathbf{v}^{\ast}$. Hence, among the two paired executions, at most one can output the zero vector. It follows that the modified reduction outputs a nonzero short vector with constant probability; therefore, so does the original reduction, up to negligible loss.

Finally, the preprocessing, sampling, and encoding steps all run in time $\Oh(2^{\eps n})$ for some very small $\eps > 0$. The dominant term is the call to the assumed \PESS algorithm on $N$ integers of magnitude at most $M$, which presumably runs in time:
\begin{align*}
    \poly(n) \cdot M^{o(1/\log N)} &\ll \poly(n) \cdot M^{\eps^2/\log{N}} \\
    &=
    (BN)^{\Oh(\eps^2 n/\log{N})} = 2^{\Oh(\eps n)},
\end{align*}
as $N = \poly(n)$ and $B = n^{\Oh(1/\eps)}$.  Therefore, for any constant $\eps
> 0$ the assumed fast \PESS algorithm yields a $2^{\Oh(\eps n)}$-time algorithm that
with high probability outputs a nonzero vector of length $N^5 \cdot \lambda_n(\Ll) \le n^{10}\cdot
\lambda_n(\Ll)$, contradicting~\cref {hyp:lambda-n}. See~\cref{sec:lowerbound} for a complete proof.

\subsection{Techniques behind~\cref{thm:average-case-integer} and \cref{thm:average-case-modular}}
\cref{thm:average-case-integer} and \cref{thm:average-case-modular} are based
on the same two-level filtering framework via the {\em filtering algorithm} (see
below). We partition the $n$ input integers into four equal blocks, enumerate all subset sums
inside each block, and then perform two rounds of merging: first modulo $p$,
and then modulo $q$. Thus, in both settings, the output of the filtering stage
is a final list of subset sums that agree modulo $pq$. This agreement modulo
$pq$ is only an intermediate condition. In the modular setting, to obtain a
collision modulo $M=pqr$, one must additionally ensure collision in the
residual $r$-part, after which the Chinese Remainder Theorem yields a
collision modulo $M$. In the integer setting, agreement modulo $pq$ again
controls only the low-order bits, and one must additionally understand carry
propagation and the remaining high-order contribution. Thus, the main
difference lies not in the algorithmic template, but in its analysis $-$
in particular, the analysis of the collision variable that denotes the number of 
collisions in the final list. The modular 
setting is relatively simpler, due to uniform distributions, CRT, and modular arithmetic. 
Whereas in the integer setting, even the first-moment calculation of the
collision variable becomes substantially more intricate. Nevertheless, the modular analysis 
remains crucial. As we will see later, 
controlling the success probability requires a second-moment bound, and in the integer case, 
this is obtained by comparison with the more tractable modular collision variable.
 
\paragraph*{Setup.} For simplicity, here, we assume that the input integer bound/modulus $M = pqr$, where 
$2^{n-1} \leq M \leq 2^n$, $p, q \in [2^{n/4 - 1}, 2^{n/4}]$, and $p$, $q$, and $r$ are pairwise co-prime integers. 
Refer to Problem~\ref{prob: average_case_modular} and \ref{prob: average_case_integer} 
for the formal definition of the problems. We will investigate 
the case of unstructured $M \neq pqr$ in detail in the main section. 

\begin{table}[h!]
    \centering
    \begin{tabular}{  m{5em}  m{7cm} m{6cm}  } 
        \toprule
        & \textbf{Modular} & \textbf{Integer} \\ 
        \bottomrule
        Input & $\Vec{a} \in \mathbb{Z}_M^n$ & $\Vec{b} \in [M]^n$\\
        \hline
        Problem & Find $\Vec{u}_i \neq \Vec{u}_j$ in $\{0, 1\}^n \setminus \{0^n\}$,
        & Find $\Vec{u}_i \neq \Vec{u}_j$ in $\{0, 1\}^n \setminus \{0^n\}$ \\
        & $\Vec{a}\cdot\Vec{u}_i \equiv_M \Vec{a}\cdot\Vec{u}_j$ 
        &$\Vec{b}\cdot\Vec{u}_i = \Vec{b}\cdot\Vec{u}_j$\\ 
        \hline
        Input Split & $a_i^{(pq)} \coloneqq a_i \mod pq; a_i^{(r)} \coloneqq a_i \mod r$ 
        & $b_i = b_i^{(pq)} + pq\cdot b_i^{(r)}$ such that\\ 
        & $a_i^{(pq)} \in \mathbb{Z}_{pq}$, $a_i^{(r)} \in \mathbb{Z}_{r}$ 
        & $b_i^{(pq)} \in [0, pq)$, $b_i^{(r)} \in [0, r]$\\
        \hline
        Distribution & $a_i^{(pq)}, a_i^{(r)}$ independent & $b_i^{(pq)}, b_i^{(r)}$ independent \\
        & and uniform over their support & and uniform over their support\\
        \hline
    \end{tabular}
    \caption{Setup for Modular and Integer problem}
\end{table}
From now on, we will only use the variable $\vec{a}$ for both cases. From the context, the reader can understand the support of the variables. 

A natural question is why we do not work directly with a modulus of the form $M = pq$. The reason lies 
in the complexity of the different levels of our algorithm (see below), which is proportional to 
the sizes of the integers $p$ and $q$ 
used in the filtering steps. To balance this complexity, the modulus must 
contain two distinct prime factors of size approximately $M^{1/4}$. The remaining factor, denoted by $r$, 
accounts for the rest of the modulus. Moreover, requiring the factors to be pairwise coprime is 
essential to apply the Chinese Remainder Theorem.

\paragraph*{Filtering Algorithm.}
The algorithm behind~\cref{thm:average-case-integer} and \cref{thm:average-case-modular} has two levels.
First, we start by partitioning the input elements into four equal blocks of size $n/4$, and then enumerating all subset sums
inside each block. Let $L_1, L_2, L_3$, and $L_4$ denote these four lists consisting of $2^{n/4}$ elements. 
We merge these four lists to store subsets that collide over low-order bits. 

Formally, at the first level, we merge lists $L_1$ and $L_2$ to get a list $L_{12}$ such that this list contains all
the subsets of the first $n/2$ elements that agree on mod $p$ bits. We merge $L_3$ and $L_4$ similarly to get another list $L_{34}$.  
This merging can be achieved by a dynamic programming approach that finds solutions modulo $p$ (see ~\cite{JW24, JWZ25}).
Then, at the second level, we merge these two lists 
such that all the subset sums agree on mod $q$ bits. 
The algorithm searches the final list for the desired collision. If one is found, it outputs
the collision; otherwise, it outputs $\bot$. We represent this algorithm schematically in Figure $\ref{fig:filtering_algo}$. 

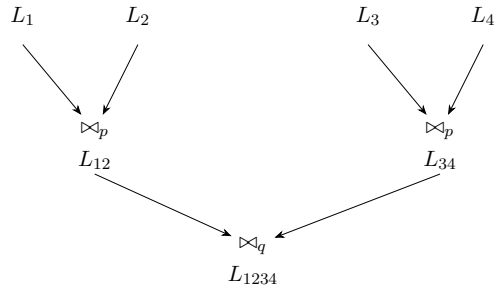
\begin{figure}[h!]
\centering
    \resizebox{.40\linewidth}{!}{% (0.90\columnwidth in the two-column version; same physical size)
    \begin{tikzpicture}[>=Stealth, node distance=2.5cm]
        % Bottom language labels
        \node (L1) at (-4,4) {$L_1$};
        \node (L2) at (-2,4) {$L_2$};
        \node (L3) at (2,4) {$L_3$};
        \node (L4) at (4,4) {$L_4$};
        
        % Concatenation operator node
        \node (concat_12) at (-2.75,2) {$\bowtie_{p}$};
        \node (concat_34) at (3.25, 2) {$\bowtie_{p}$};
        
        % Arrows from L1 and L2 to operator
        \draw[->] (-4, 3.5) -- (concat_12);
        \draw[->] (-2, 3.5) -- (concat_12);
        
        % Arrows from L3 and L4 to operator
        \draw[->] (2, 3.5) -- (concat_34);
        \draw[->] (4, 3.5) -- (concat_34);
        
        \node (L12) at (-2.75, 1.5) {$L_{12}$};
        \node (L34) at (3.25, 1.5) {$L_{34}$};
        \node (concat_1234) at (0, 0) {$\bowtie_{q}$};
        
        % Arrow downward to result
        \draw[->] (-2.75, 1.25) -- (concat_1234);
        \draw[->] (3.25, 1.25) -- (concat_1234);
        \node (L1234) at (0, -0.5) {$L_{1234}$};
    \end{tikzpicture}}
\caption{Schematic diagram of the filtering algorithm.}\label{fig:filtering_algo}
\end{figure}

To analyze this common filtering procedure, we introduce collision-counting random variables and study their first and second moments.

\paragraph*{Distribution of collisions.}
Now, we aim to analyze the distribution of integer and modular collisions in the final list $L_{1234}$.
For this, we first define two random variables that denote the number of non-trivial integer
and modular collisions in this final list. Let $\mathcal{C}^{(\mathsf{int})}$ and $\mathcal{C}^{(\mathsf{pqr})}$ denote these two
random variables, respectively. Our filtering algorithm succeeds if the value of these variables is at least $1$. Hence, we look for the probability of this event. 

For any input instance, we can compute the value of $\mathcal{C}^{\mathsf{(int)}}$, and similarly $\mathcal{C}^{\mathsf{(pqr)}}$, 
by iterating through all the possible pairs of binary vectors, 
and considering whether the vectors (or the corresponding subset sum) pass through the 
filters of the algorithm and collide, over integers or modulo $pqr$, based on the problem.  
For each pair of binary vectors $\vec{u}_i, \vec{u}_j$, we define a binary random variable $\mathcal{C}^{(\mathsf{int})}_{ij}$
and $\mathcal{C}^{(\mathsf{pqr})}_{ij}$,
such that their value is $1$ if and only if $\vec{u}_i$ and $\vec{u}_j$ pass through all the filters of the algorithm, 
and they collide over integers or modulo $r$ based on the problem. 

\paragraph*{First Moment.} 
    In the modular case, we have that the probability of vectors $\vec{u}_i$ and $\vec{u}_j$ surviving the algorithm is $1/p^4q^2$
    \footnote{The probability
    can be different if the left part (first $n/2$ elements) or the right part (last $n/2$ elements) 
    of these two binary vectors is the zero vector. We'll carefully analyze those
    cases in the main section.},
    and the probability of them colliding modulo $r$ is $1/r$, given $\vec{u}_i \neq \vec{u}_j$. 
    Hence, we have 
    $\underset{\vec{a}}{\mathbb{E}} \left(\mathcal{C}^{(\mathsf{pqr})}_{ij}\right) = 1/p^4q^2r$.
    Since there are $\Theta(4^n)$ pairs of binary vectors, using the bounds
    for $M, p$, and $q$, we have
    $\underset{\vec{a}}{\mathbb{E}} \left(\mathcal{C}^{(\mathsf{pqr})}\right) = \Omega(1)$ (see~\cref{sec: average-case-modular} for details). 

    Next, we focus on the integer case.  The vectors surviving the algorithm and their collision over the remaining bits are
    dependent on each other. Hence, the analysis of the first moment in this case
    is a bit involved.
    We split the integer subset sum as follows: 
    \vspace{-1em}
    
    \begin{equation*}
        \begin{aligned}
            \vec{a} \cdot \vec{u}_i &= \left(\vec{a}^{(pq)} + \vec{a}^{(r)} \cdot pq\right)
            \cdot \vec{u}_i,\\
            &= \underbrace{\vec{a}^{(pq)} \cdot \vec{u}_i \mod pq}_{\mod pq 
            \text{ bits}} + 
            \underbrace{\frac{\vec{a}^{(pq)}\cdot \vec{u}_i - (\vec{a}^{(pq)}\cdot\vec{u}_i 
            \mod pq)}{pq}}_{\text{carry bits}} \cdot pq + 
            \underbrace{\vec{a}^{(r)} \cdot \vec{u}_i}_{\text{remaining bits}} \cdot \,pq.
        \end{aligned}
    \end{equation*}
    
    To make the analysis simpler, also since we are only interested in the lower bound of the first moment, we redefine the random variable $\mathcal{C}^{(\mathsf{int})}_{ij}$ as follows: its value is $1$ if
    and only if $\vec{u}_i, \vec{u}_j$ survive filtering algorithm,
    they collide over the carry bits, and the remaining bits separately.
    Since the vectors $\vec{a}^{(pq)}$ and $\vec{a}^{(r)}$ are independent, we can analyze the corresponding sums separately. 
    
    To characterize the distribution of the sum of integer-valued random variables, we employ the 
    Local Central Limit Theorem (LCLT)~\cite{ibragimov1965independent, mitzenmacher2017probability}. Informally, the LCLT asserts that the distribution of the sum of I.I.D. random variables converges to a discrete Gaussian distribution, and that in a neighborhood of the mean, the point probabilities are approximately uniform up to constant factors. 
    To apply this result in our setting, we restrict attention to binary vectors whose Hamming weight is 
    $\Theta(n)$. A formal statement of the theorem is provided in~\cref{sec: average-case-integer}.

    Leveraging the smoothness of the sum distributions, we first show that the carry bits converge in distribution to a discrete normal distribution, 
    and henceforth we prove that the total sum modulo $pq$
    and the left sum modulo $p$ condition on the carry bits converges to the uniform distribution asymptotically.
    Finally, we show that the following three events $E_1, E_2$ and $E_3$ are asymptotically independent, where $E_1 \coloneqq \sum_{j = 1}^{n/2} a_j^{(pq)} \cdot u_{ij} \equiv_{p} \alpha, \; E_2\coloneqq \vec{a}^{(pq)} \cdot \vec{u}_i \equiv_{pq} \beta$,
    and $E_3 \coloneqq \floor{\vec{a}^{(pq)} \cdot \vec{u}_i/pq} = \gamma$.

    After analyzing the dependency of the events relying on the vector $\vec{a}^{(pq)}$, and using the probability of two Gaussian distributions colliding, 
    we get that the probability that both $\vec{u}_i$, $\vec{u}_j$ survive the filter algorithm and they collide over the carry bits is $\Theta(1/(p^4q^2\sqrt{n}))$. 
    Combining this result and using the Local Central Limit Theorem again for the
    remaining bits, we have $\underset{\vec{a}}{\mathbb{E}} \left(\mathcal{C}^{(\mathsf{int})}_{ij}\right) = \Omega(1/(np^4q^2r))$. Since there are $\Theta(4^n)$ pairs of binary vectors such that 
    their Hamming weights satisfy \textit{certain} specific conditions (mentioned in~\cref{sec: average-case-integer}),
    we have $\underset{\vec{a}}{\mathbb{E}} \left(\mathcal{C}^{(\mathsf{int})}\right) = \Omega(1/n)$.

\paragraph{Second moment.}
At this point, we know that the expected number of collisions is positive (and in the modular case, even bounded away from $0$).  However, a lower bound on $\mathbb{E}[\mathcal{C}]$ alone does not imply that the algorithm succeeds on a typical input. Indeed, a random variable can have expectation $1$ while
being $0$ with overwhelming probability and very large on a tiny exceptional set. 
Thus, the first moment identifies only the correct \emph{scale} of the number of collisions; by itself, it {\em does not} yield a lower bound on $\mathbb{P}[\mathcal{C} > 0]$.

Therefore, to prove that the filtering algorithm succeeds with substantial probability, it is essential to control the second moment of $\mathcal{C}$.  This is precisely why our analysis focuses on bounding $\mathbb{E}[\mathcal{C}^2]$, and, in the integer
case, why we first pass through the more tractable modular variable
$\mathcal{C}^{(pqr)}$ in~\cref{sec: average-case-modular}.

As discussed above, analyzing $\mathcal{C}^{(\mathsf{int})}_{ij}$ directly is challenging because of the interaction between pseudo-uniform distributions and carry propagation. In particular, obtaining a useful upper bound on its second moment appears substantially {\em more delicate} than in the modular setting.
However, since
$\mathcal{C}^{(\mathsf{int})}_{ij} \leq \mathcal{C}^{(\mathsf{pqr})}_{ij}$ and
our goal is to upper-bound the second moment, it is enough to study the second
moment of $\mathcal{C}^{(\mathsf{pqr})}_{ij}$ instead!\footnote{In principle, the relevant moments in the modular case can also be approached using exponential sums. However, this becomes convoluted 
at the level of the second moment, so we instead adopt an algebraic approach. We expect that 
a more refined use of exponential sums could potentially improve the success
probability by exploiting additional randomness in the filtering algorithm.}
This modular quantity is considerably more tractable, since it can be analyzed using the structural advantages of the Chinese Remainder Theorem and modular arithmetic. Thus, the modular case is not merely a warm-up: it provides the key input for controlling the second moment, and hence the success probability, in the integer setting.

Next, we explain how to compute the second moment $\mathcal{C}^{(\mathsf{pqr})}$.
The second moment of $\mathcal{C}^{(\mathsf{pqr})}$ can be written as follows (we replace $\mathcal{C}^{(\mathsf{pqr})}$ with $\mathcal{C}$ for the ease of presentation): 
\begin{equation*}
    \begin{aligned}
        \underset{\vec{a}}{\mathbb{E}}\left(\mathcal{C}^2\right) &= 
       \underset{\vec{a}}{\mathbb{E}}\left(\sum_{}\mathcal{C}_{ij}.\sum_{}\mathcal{C}_{\ell k}\right)\\
        &= \sum_{i < j, \ell < k} \underset{\vec{a}}{\mathbb{E}}\left(\mathcal{C}_{ij}\cdot\mathcal{C}_{\ell k}\right) \\
        &= \sum_{\substack{(i, j) \neq (\ell, k)}} \Pr_{\vec{a}} \left(\mathcal{C}_{ij} = 1 \land \mathcal{C}_{\ell k} = 1\right) + 
        \underset{\vec{a}}{\mathbb{E}}\left(\mathcal{C}\right)
    \end{aligned}
\end{equation*}
The main difficulty in analyzing the second moment stems from dependencies between events induced by linear relationships among the binary vectors.

To address this issue, we partition the set 
$S$ of all possible $4-$tuples of binary vectors into cells such that, for every tuple within a given cell, the joint probability $\Pr_{\vec{a}} (\mathcal{C}_{ij} = 1 \land \mathcal{C}_{\ell k} = 1)$,
admits a uniform upper bound, while the size of each cell can also be tightly bounded. This structured partitioning allows us to control the contribution of dependent events to the second moment. The construction and analysis of these cells rely on a combination of linear-algebraic and combinatorial techniques.
After analyzing all these dependencies, we get the desired result that $\underset{\vec{a}}{\mathbb{E}}\left(\mathcal{C}^2\right) =  \mathcal{O}(1)$ (see~\cref{sec: average-case-modular} for details).

\paragraph{Complexity Analysis.}
The complexity analysis of the algorithm boils down to computing the expected sizes of the list 
generated by the algorithm. A crucial observation is that the algorithm only depends on the random variable $\Vec{a}^{(pq)}$. Now, if $\Vec{a}^{(pq)}$ is uniformly distributed over $\mathbb{Z}_{pq}^n$ then
$\Vec{a}^{(pq)} \cdot \Vec{u}$ modulo $pq$ is uniformly distributed over $\mathbb{Z}_{pq}$, given $\vec{u} \in \{0, 1\}^n \setminus \{0^n\}$. 
Using similar arguments, we can say that in both the modular and the integer case,
subset sums mod $p$ and $q$ are uniformly distributed over $\mathbb{Z}_p$ and $\mathbb{Z}_q$ respectively. Therefore, the probability of any subset
surviving the first level is $1/p^2 = 1/2^{n/2}$ ($1/p$ for each constraint), 
and it reaching the second level is $1/p^2q = 1/2^{3n/4}$. 
And, since there are $2^n$ subsets in total, the expected size of the final list is $2^{n/4}$. 
Using similar arguments, we conclude that the expected sizes of the lists $L_{12}$ and $L_{34}$ are also $2^{n/4}$.

To the best of our knowledge, this is the first fine-grained average-case
analysis for Equal Subset Sum (with density $1$) that combines two-round
Wagner-style pruning with an explicit second-moment argument. We further
conjecture that the same viewpoint may eventually extend to the worst-case setting, once combined with the structural ideas developed by Jin and Wu~\cite{JW24}. We also conjecture that by using the randomness available in the filtering algorithm, i.e., choosing the bins modulo $p$ and $q$
uniformly randomly rather fixed values, we can improve the success probability to $1 - o(1)$.

\subsection{Related Work}
Ajtai~\cite{ajtai96} showed that an $\Oh(\sqrt{n})$-approximation to Minkowski's vector problem would yield an $\Oh(n)$-approximation to the Shortest Vector problem.
Hoberg et al.~\cite{hoberg2017number} showed connection between Number Balancing
Problem and Minkowski. In the setting of Pigeonhole Equal Subset Sum, their
result implies that a polynomial-time algorithm when the sum of the integers is $2^{n-\sqrt{n}}$ would
yield a polynomial approximation to the Shortest Vector Problem in polynomial time.

There has been a significant interest in analyzing the Subset Sum in restricted regimes or under average-case assumptions. The classical result by Lagarias and Odlyzko~\cite{LO85} shows that Subset Sum can be solved in polynomial time on average when the input lies in the low-density regime. This line of work has had cryptographic implications, particularly in attacking knapsack-based cryptosystems, and has led to further developments~\cite{low-density-subset-sum2,HJ10,BCJ11,DBLP:conf/stacs/JouxW26}. More recently, Agrawal\etal~\cite{ASS+24} presented a new algorithm for average-case $k$-XOR problems at low densities.

With regard to reductions from variants of the Subset Sum problem to lattice problems, van Emde Boas~\cite{van81} gave a reduction from Equal Sum to the Closest Vector Problem (CVP). A standard reduction from ordinary Subset Sum to CVP is explicitly presented by Micciancio~\cite{Mic02}. Krishnamoorthy, Webb, and Moyer~\cite{KWM12} reduced a decisional version of a variant of the Number Partitioning Problem to CVP. Yu and Xiao~\cite{YX18} reduced multi-dimensional modular subset sum to the Bounded Distance Decoding problem. Dutta and Rajasree~\cite{DR23} showed that Unbounded Subset Sum reduces to CVP under the $\ell_1$ and $\ell_\infty$ norms. More recently, Hittmeir~\cite{hit25} gave a reduction from a constrained variant of Subset Sum to the gap variant of Shortest Vector Problem.

\section{Preliminaries}\label{sec:prelim}

Throughout the paper, all logarithms are base $2$ unless stated otherwise. We denote the set of all positive integers up to $n$ as $[n] \coloneqq \set{1, \ldots, n}$. Let $\BigO^*(\cdot),\Omega^*(\cdot)$ hide $\poly(n)$ factors, where $n$ is the number of input integers in the \PHES problem.  We use $\w(A) = \sum_{i \in A} a_i$ for $A \subseteq [n]$. For integers $a,b,c$ we use $a \equiv_c b$ or $a \equiv b \pmod c$ to denote 
the fact that $c$ divides $a-b$. 

We denote vectors in boldface, e.g., $\vec{v}$, and matrices in uppercase, e.g., $\mat{A}$. For a vector $\vec{v}$, we denote its $i^{th}$ coordinate by $v_i$ and its $\ell_2$-norm by $\norm{\vec{v}}_2$. For any matrix $\mat{B} \in \R^{n \times m}$, we denote $\norm{\mat B}_2 = \max_{i \in [m]} \{\norm{\mat b_i}_2\}$.

\begin{theorem}[Prime Number Theorem]\label{thm:prime-counting}
    For any $x \geq 2$ the number of primes less than or equal to $x$ is given by $\pi(x) = \Theta\left(\frac{x}{\log(x)}\right)$.
\end{theorem}

\subsection{Lattices}
A lattice $\mathcal{L}$ is the set of all integer combinations of linearly independent vectors $\mathbf{b}_1, \mathbf{b}_2, \ldots, \mathbf{b}_n \in \mathbb{R}^d$:
\[
\mathcal{L} = \left\{ \sum_{i=1}^n z_i \mathbf{b}_i \mid z_i \in \mathbb{Z} \right\}
\]
where $n$ is called the rank of the lattice, and $d$ is the dimension. The vectors $\mathbf{b}_1, \mathbf{b}_2, \ldots, \mathbf{b}_n$ form a basis of the lattice, and the matrix $\mathbf{B}$ with columns $\mathbf{b}_1, \mathbf{b}_2, \ldots, \mathbf{b}_n$ is called a basis matrix. 

\begin{definition}[Shortest Vector Problem (SVP)]
    Given a lattice $\mathcal{L}$, the Shortest Vector Problem (SVP) asks for a shortest non-zero vector in the lattice, i.e.,
    \[
    \svp(\mathcal{L}) = \mathsf{argmin} \{ \| \mathbf{v} \|_2 \mid \mathbf{v} \in \mathcal{L}, \mathbf{v} \neq \mathbf{0} \}
    \]
\end{definition}

A standard notation for the length of the shortest vector in a lattice $\L$ is denoted by $\lambda_1(\L)$. In fact, a generalization of $\lambda_1$ known as successive minima is defined as follows.

\begin{definition}[Successive minima]
    Given a lattice $\mathcal{L}$ and $i \in [n]$, the $i^{th}$ successive minima $\lambda_i(\L)$ is defined as
    \[\lambda_i(\L) = \min \{ r \in \R \;|\; \mathsf{dim}(\mathsf{span}(\L \cap \Ball_{\ell_2}(0,r))) \geq i \}\]
\end{definition}

Throughout the remainder of this section, unless stated otherwise, all norms are with respect to the $\ell_2$-norm.

\begin{definition}[Approximate Shortest Independent Vector Problem (SIVP)]
    Given an $n$-dimensional lattice $\mathcal{L}$ and an approximation function $\gamma: \N \to \R$, the $\gamma$-approximate SIVP asks for $n$ linearly independent vectors $\mathbf{v}_1,\ldots,\mathbf{v}_n \in \mathcal{L}$ such that $\norm{\mathbf{v}_i}_2 \le \gamma(n)\lambda_n(\mathcal{L})$ for all $i \in [n]$.
\end{definition}

We consider the following variant of the above problem, which asks for a single vector.

\begin{definition}[Approximate $n^{th}$-Successive Minima Problem (nSMP)]\footnote{There is a well-studied lattice problem, called the \emph{Successive Minima Problem} (SMP)~\cite{Mic08}, which asks for $n$ linearly independent lattice vectors $\mathbf{v}_1, \dots, \mathbf{v}_n$ such that $\norm{\mathbf{v}_i} = \lambda_i(\L)$ for every $i \in [n]$.}
\label{def:nSMP}
    Given a lattice $\mathcal{L}$ and an approximation function $\gamma: \N \to \R$, the $\gamma$-approximate $n^{th}$ Successive Minima Problem (nSMP) asks for a non-zero vector $\mathbf{v}$ in the lattice $\L$ such that
    $\norm{\mathbf{v}} \leq \gamma(n) \lambda_n(\L)$.
\end{definition}

Using~\cref{thm:SIVP-SVP-variant}, one can show that any algorithm for $\gamma(n)$-approximate nSMP can be used to solve $\sqrt{n} \gamma(n)$-approximate SIVP.

\subsection{Hardness Hypothesis}

For any parameter $\kappa \ge 1$, the currently best algorithms for
$n^\kappa$-approximate SVP run in $2^{\Oh(n/\kappa)}$
time~\cite{schnorr1987hierarchy,schnorr1994lattice,GN08,ALNS20}.  More
precisely, when $\kappa \ge 1$, $n^{\kappa}$-approximation to SVP can be
computed in $\Os(2^{0.802n/(\kappa+1)})$ time~\cite{ALNS20}.  Motivated by
the apparent hardness, Vafa and Vaikuntanathan~\cite{VV25} implicitly stated the
following hypothesis~\footnote{To be precise,~\cite{VV25} needed only to assume
$\ell = \Theta(\sqrt{n})$ and assumed hardness of SVP rather than SIVP. In their setting, these problems are equivalent~\cite{MR07}.}:

\begin{hypothesis}\label{hyp:weak-sivp}
    For any $\eps > 0$ and parameter $\ell \in [1,n]$, there is no $2^{\Oh(n/\ell)^{1-\eps}}$-time algorithm that given a lattice $\mathcal{L}$ of dimension $n$ returns $n$ non-zero independent vectors $\mathbf{v}_i \in \mathcal{L}$ such that $\norm{\mathbf{v}_i}_{2} \le 2^{\Oh(\ell)} \cdot \lambda_n(\L)$. 
\end{hypothesis}

However, notice that since the work of Schnorr~\cite{schnorr1987hierarchy}, there is no $2^{o(n)}$ time algorithm that solves $n^c$-approximate SVP for any constant $c \ge 0$. Motivated by this, we state an even stronger hypothesis than~\cref{hyp:weak-sivp}.

\hypmain*
To the best of our knowledge, this hypothesis was not explicitly stated. However, if the hypothesis turns out to be false, it would mean a significant breakthrough on the approximate Shortest Vector Problem. 
In section~\cref{sec:nsmp-variant}, we show that these problems are tightly connected.

\subsection{Declaration on the Use of AI Tools.}
\textbf{No} AI tools were used for the version submitted to the FOCS 2026 conference.  
A generative AI tool, Claude (Opus), was used solely for proofreading, 
language editing, and identifying potential inconsistencies for this arXiv version. 
All technical contributions are original to the authors, who independently verified relevant AI-generated suggestions and take full responsibility for the manuscript.

% ===== end preliminaries.tex =====
% ===== begin subexponential.tex =====
\section{Subexponential algorithm: Proof of~\cref{thm:subexponential}} \label{sec:subexponential}

In this section, we present our subexponential algorithm for weak \PESS.

\thmsubexponential*

Let $k$ be the largest power of $2$ such that $n \cdot M \le 2^{\floor{n/k}
\log{k}}$ (if such a $k$ does not exist, then $\delta = \Omega(1)$ and
the brute-force algorithm solves the problem in the desired time).
Let $W \coloneqq 2^{\floor{n/k}}$. 

If $k = \Omega(n)$, then $M = \poly(n)$ and we can run 
dynamic programming of ~\cite{vasilakis2022improved} in $\Ot(nM) = \poly(n)$ time. Hence, we may assume $k = o(n)$.  We will present an algorithm running in $\poly(n) \cdot W^2$ time.
This will conclude the proof of~\cref{thm:subexponential} by setting $\delta = \Theta(k/\log{k})$.

First, arbitrarily partition the input set $S = S_1 \uplus \ldots \uplus S_k \uplus S_{k+1}$ such that $|S_i| = \floor{n/k}$ for every $i \in [k]$. We will ignore the set $S_{k+1}$. In the base case, for every $i \in [k]$, compute the sets:
\begin{align*}
    \Aa_i^{(0)} \coloneqq \left\{ \Sigma(X) \mid X \subseteq S_i \right\}.
\end{align*}
Alongside every $a \in \Aa_{i}^{(0)}$ (and later sets $\Aa_{i}^{(j)}$'s), we will store a set $X \subseteq S$ such that $\Sigma(X) = a$. This data structure incurs only an additional $\poly(n)$ overhead in the running time and will allow us to return a collision. Observe that $|\Aa_i^{(0)}| = 2^{\floor{n/k}}$ for every $i\in [k]$: if $|\Aa_i^{(0)}| < 2^{\floor{n/k}}$, then during construction of the set $\Aa_i^{(0)}$ we detected a collision, i.e., two different sets of the same sum, which we return immediately (e.g., by using a data structure described in the previous sentence).
Moreover, all sets $\Aa_i^{(0)}$ can be computed in $\poly(n) \cdot W$ time. Let
$p_{i}^{(0)} \coloneqq 0$ for every $i \in [k]$. 

This concludes the construction of the sets $\Aa_i^{(j)}$ and $p_i^{(j)}$ for $j=0$ and $i \in [k]$. We now iterate over $j = 1,2,\ldots,\log{k}$. In the $j$-th step, we compute the set $\Bb_{i,p}^{(j)}$.
This set $\Bb_{i,p}^{(j)} \subseteq \Aa_{2i-1}^{(j-1)} + \Aa_{2i}^{(j-1)}$ such that 
every $a \in \Bb_{i,p}^{(j)}$ satisfies the following modular equation
\[
    a  \equiv p \cdot W^{j-1} + p_{2i-1}^{(j-1)} + p_{2i}^{(j-1)} \pmod{W^j}.
\]

for every $i \in [k/2^j]$ and $p \in \{0,\ldots,W-1\}$.  
If at any point we detect a collision (i.e., two equal elements), we immediately return it.
Next, for every $i \in [k/2^j]$, let $p^\star \in \{0,\ldots,W-1\}$ be the index maximizing
$|\Bb_{i,p^\star}^{(j)}|$, and define:
\begin{align*}
    p_{i}^{(j)} \coloneqq p^\star \cdot W^{j-1} + p_{2i-1}^{(j-1)} + p_{2i}^{(j-1)},
\end{align*}
and let $\Aa_{i}^{(j)}$ be an arbitrary subset of $\Bb_{i,p^\star}^{(j)}$ of size at most $W$
(i.e., we truncate $\Bb_{i,p^\star}^{(j)}$ if it exceeds size $W$). This completes the algorithm description.

\paragraph{Running Time.} Note that
$|\Aa_i^{(j)}| \le W$ for all $j \in [\log{k}]$ and $i \in
[k/2^j]$. The computation of each of these sets takes $\poly(n) \cdot W^{2}$ time, which implies the desired running time.

Before addressing correctness, observe that throughout the computation, we maintain the following invariant:

\begin{claim}\label{clm:partition}
    For every $j \in [\log{k}]$ and $i \in [k/2^j]$, it holds that:
    \begin{align}\label{eq:partition}
        \Aa^{(j-1)}_{2i-1} + \Aa^{(j-1)}_{2i} = \Bb^{(j)}_{i,0} \uplus \Bb^{(j)}_{i,2} \uplus \ldots \uplus \Bb^{(j)}_{i,W-1}.
    \end{align}
\end{claim}
\begin{claimproof}
    Fix $j \in [\log{k}]$ and $i \in [k/2^j]$. For every $p \neq p'$ the sets
    $\Bb^{(j)}_{i,p}$ and $\Bb^{(j)}_{i,p'}$ are disjoint because
    they contain numbers in different congruence classes modulo $W^j$. 
    Next, observe that for every $j \in [\log{k}]$ and $i \in [k/2^j]$, it holds that for every $a \in
    \Aa_i^{(j)}$:
    \begin{align*}
        a \equiv p^{(j)}_i \pmod{W^j}.
    \end{align*}
    Let $\Bb \coloneqq \Aa^{(j-1)}_{2i-1} + \Aa^{(j-1)}_{2i}$. By definition $\Bb \supseteq
    \Bb^{(j)}_{i,p}$ for every $0 \le p < W$ , hence it remains to prove the $\subseteq$
    of~\eqref{eq:partition}. Observe that for every $b \in \Bb$ it holds that:
    \begin{align*}
        b \equiv p^{(j-1)}_{2i-1} + p^{(j-1)}_{2i} \pmod{W^{j-1}}.
    \end{align*}
    Therefore, for every $b \in \Bb$ there exists $p \in \{0,\ldots,W-1\}$ such that
    \begin{align*}
        b \equiv p \cdot W^{j-1} + p^{(j-1)}_{2i-1} + p^{(j-1)}_{2i} \pmod{W^j},
    \end{align*}
    which in particular means that $b \in \Bb^{(j)}_{i,p}$.
\end{claimproof}

\paragraph{Correctness.} First, we prove the following:

\begin{claim}\label{clm:correctness}
    If the algorithm does not detect a collision, then 
    \begin{align*}
        |\Aa_{i}^{(j)}| \ge W,
    \end{align*}
    for all $j \in \{0,\ldots,\log{k}\}$ and $i \in [k/2^j]$.
\end{claim}
\begin{claimproof}
    Suppose, for contradiction, that the claim is false. Let $j \ge 0$ be the smallest index
    for which it fails. We must have $j > 0$, since in the base case we have $|\Aa^{(0)}_i| = 2^{\floor{n/k}}$ for all $i \in [k]$.

    So assume that $|\Aa^{(j)}_i| < W$ for some $i \in [k/2^j]$, and let $\Bb \coloneqq
    \Aa^{(j-1)}_{2i-1} + \Aa^{(j-1)}_{2i}$. Since the algorithm did not detect
    a collision at levels $< j$, it must be that $|\Aa^{(j-1)}_{2i-1}|$ and
    $|\Aa^{(j-1)}_{2i}|$ are at least $W$. Hence $|\Bb| \ge W^2$,
    as the algorithm did not detect a collision at level $j$.

    On the other hand, by Claim~\ref{clm:partition}:
    \begin{align*}
        \Bb = \Bb^{(j)}_{i,0} \uplus \Bb^{(j)}_{i,1} \uplus \ldots \uplus \Bb^{(j)}_{i,W-1},
    \end{align*}
    and by the assumption that $|\Bb^{(j)}_{i,p}| \le |\Aa_i^{(j)}| < W$ for all
    $p \in \{0,\ldots,W-1\}$, we get:
    \begin{align*}
        |\Bb| = \sum_{p=0}^{W-1} |\Bb^{(j)}_{i,p}| \le \sum_{p=0}^{W-1} |\Aa_i^{(j)}| < W^2,
    \end{align*}
    contradicting $|\Bb| \ge W^2$.
\end{claimproof}

With Claim~\ref{clm:correctness}, we can now conclude the proof of~\cref{thm:subexponential}. Let $\Cc = \Aa_{1}^{(\log{k})}$ and $p =
p^{(\log{k})}_1$, and assume that the algorithm did not detect a collision in any of the previous steps. By Claim~\ref{clm:correctness}, we have $|\Cc| \ge 2$,
so there exist two distinct elements $a, b \in \Cc$ such that:
\begin{align*}
    a \equiv b \equiv p \pmod{W^{\log{k}}}.
\end{align*}
However, both $a$ and $b$ are bounded above by $n \cdot M \le 2^{\floor{n/k} \log{k}} \le W^{\log{k}}$,
so the congruence implies $a = b$, contradicting the assumption that no collision was detected.
\qed

\begin{remark}\label{remark:exact-runtime}
    By plugging in the exact parameters, for $k=2^7$, algorithm behind \cref{thm:subexponential} solves Equal Subset Sum in $\Os(2^{n/64})$ time when $n M \le 2^{7n/128} < 2^{n/18}$.
\end{remark}

% ===== end subexponential.tex =====
% ===== begin lowerbound.tex =====
\section{Conditional Lowerbound: Proof of~\cref{thm:lowerbound}}\label{sec:lowerbound}

In this section, we present our lower bound for solving the Pigeonhole Equal Subset Sum
problem in the strong regime by reducing it to the problem of finding a short
vector in a lattice. We assume reader is familiar with~\cref{sec:prelim}.

\subsection{Construction}

In this section, we present a reduction from \SIVP to \PESS.  At a high level,
the reduction first computes a sufficiently short basis of the input lattice
$\L$. It then uses this basis to sample many lattice points from a suitable
discrete Gaussian distribution and encodes these samples as positive integers.
These integers are given to the \PESS algorithm, whose output yields two
collision sets. From these sets, the reduction extracts a short-ish (non-zero)
lattice vector in $\L$. We begin by describing the basis reduction and sampling
procedure used in the reduction, followed by the encoding step and the
extraction of the lattice vector. We then analyze the correctness and running
time of the reduction.

\paragraph*{Step 1: Preprocessing.}
We begin by computing a sufficiently short basis of the input lattice. The following lemma provides the required basis-reduction procedure. As this is a standard step and follows from the previous works~\cite{schnorr1987hierarchy,schnorr1994lattice, GN08, ALNS20} we deferred the proof to \cref{sec:proof-short-basis}.

\begin{restatable}{theorem}{lemshortbasis} \label{lem:short-basis}
    For any $\kappa > 0$, there exists a constant $C > 0$ and an algorithm such
    that, given an $n$-dimensional lattice $\Ll$, computes a basis $\mathbf{D} =
    [\mathbf{d}_1,\ldots,\mathbf{d}_n]$ of $\L$ in time
    $\Oh(2^{C\, n/\kappa})$ such that $\norm{\mathbf{D}}_2 =
    \max_{i \in [n]} \|\mathbf{d}_i\|_2 \leq B \cdot \lambda_n(\L)$, where $B = C \cdot n^\kappa$.
\end{restatable}

Let $N$ and $B$ be parameters to be chosen later.  We will set them so
$\Omega(n \log{n}) \le N \le B \ll 2^{N/n}$, but for now, let us not consider
an exact exponent in $n$.  This parameter will correspond to the size of the
instance First, it sets $s \coloneqq 4\pi\sqrt{n}\,N \norm{\mat D}_2$ and
samples $N$ lattice points using \cref{thm:wide_sampling_shift_invariant}. 

\begin{restatable}{theorem}{thmsampling}\label{thm:wide_sampling_shift_invariant}
Let $\L \subseteq \R^n$ be a full-rank lattice. Let $\mathbf{B}$ be a basis of $\L$. Then there is a randomized algorithm that on input $\mathbf{B}$ and $s \geq 4 \pi \sqrt{n}\, \norm{\mathbf{B}}_2$, outputs $N$ (independent) lattice points $\mat x^{(1)},\ldots, \mat x^{(N)} \sim \mathcal{D}_{\L, s}$ where $N \geq 2$ such that
\begin{enumerate}
    \item The running time of the algorithm is $N\cdot \poly(n,\mathrm{bitlength}(\mathbf{B}), \log(s))$. 

    \item For any $R \geq s \sqrt{n\log(N)}$, except with probability at most $2^{-n}$, we have that
    $x_j \in \L \cap \Ball_\infty(0,R)$ for all $j\in[N]$.
\end{enumerate}
\end{restatable}

An exact definition of the distribution $\mathcal{D}_{\L,s}$ is not important
for the rest of the proof, as we will only use the properties of
the~\cref{thm:wide_sampling_shift_invariant}.  Sampling is a standard tool and
as \cref{thm:wide_sampling_shift_invariant} follows from~\cite{MR07} we defer
the proof to~\cref{sec:sampling}

Notice that using the above theorem, the reduction can efficiently generate
samples whose joint distribution is negligibly close to that of $N$ independent
samples from $\mathcal{D}_{\L,s}$. Moreover, this choice of $s$ also yields the
$\ell_\infty$ bound on all sampled vectors required later in the analysis,
except with probability at most $2^{-n}$. The sampler may produce the same
lattice point more than once. One could simply abort in this case, since such an
event occurs with negligible probability. However, this is unnecessary for our
proof, as repeated samples are handled directly in the correctness analysis.

\paragraph{Step 2: From Lattice Samples to a Candidate Lattice Vector.}

Next, the reduction converts the sampled lattice points into positive integers.
To do this, it first samples $r_1,\ldots,r_n$ independently and uniformly from
the set of primes in $\{1,\dots,N^3\}$. This can be implemented by enumerating
all primes in $\{1,\dots, N^3\}$ in time $N^{\Oh(1)}$~\cite{AKS04} and then
sampling uniformly from this set\footnote{Alternatively, one could repeatedly
sample a random integer and test whether it is prime. However, this would yield
only an expected-time procedure for sampling a uniformly random prime from the
interval.}. It also samples $s_1,\ldots,s_n$ independently and uniformly from
$\{1,\dots,N^2\}$. 

Next, the reduction enumerates powers of two in the range $[2^{\lceil \log(\norm{\mathbf{D}}) \rceil}, \dots, 2^{\lfloor \log(\norm{\mathbf{D}} - \log(B) \rfloor}]$ \footnote{From \cref{lem:short-basis}, we have $\norm{\mat D}_2 \leq B \cdot \lambda_n(\L)$ which implies $\lambda_n(\L) \geq \norm{\mat D}_2/B$. On the other hand, $\lambda_n(\L) \leq \norm{\mat D}_2$.} to obtain a value $\tilde{\lambda}_n$ satisfying $\lambda_n(\L) \leq \tilde{\lambda}_n \leq 2\lambda_n(\L)$. For each guess, the reduction performs the following two steps.

\paragraph*{Step 2a: Integer Encoding.}
Using the random choices $\{r_j, s_j\}_{j \in [n]}$, define
\begin{align*}
    f(\mathbf{x})
    \coloneqq
    W^n + \sum_{j=1}^n W^{j-1}
    \left\lfloor
        \frac{(\mathbf{x})_j \cdot \tilde{\lambda}_n^{-1} + s_j}{r_j}
    \right\rfloor,
\end{align*}
where $W \coloneqq 8BN^3$. Finally, for each $i \in [N]$, the reduction sets $a_i \coloneqq f(\mathbf{x}^{(i)})$ and defines the \PESS instance as $\mathcal{A} \coloneqq (a_1,\ldots,a_N)$.

Notice that the parameter $W$ is chosen large enough so that, for the correct guess $\widetilde{\lambda}_n$, there are no carries between different base-$W$ digits, while the additive term $W^n$ ensures that $f(\mathbf{x})$ is positive for every sampled lattice point $\mathbf{x}^{(i)}$ except with negligible probability.

\paragraph*{Step 2b: Reconstruction of a Lattice Vector.} The reduction runs the \PESS algorithm on this instance, and converts the returned unequal sets $(X,Y)$ into a lattice vector by computing $\mathbf{v}^{(\tilde{\lambda}_n)} := \sum_{i \in X} \mathbf{x}^{(i)} - \sum_{j \in Y} \mathbf{x}^{(j)}$. If the \PESS algorithm fails to find a solution, the reduction treats this as returning $(\emptyset,\emptyset)$. In this case, the corresponding candidate vector is $\mathbf{0}$ and is discarded at the final output stage.

\paragraph*{Step 3: Final vector.} After constructing the vectors $\mat v^{(\tilde{\lambda}_n)}$ for $\tilde{\lambda}_n \in [2^{\lceil \log(\norm{\mathbf{D}}) \rceil}, 2^{\lceil \log(\norm{\mathbf{D}}) \rceil - 1},\allowbreak \dots,\allowbreak 2^{\lfloor \log(\norm{\mathbf{D}} - \log(B) \rfloor}]$, the reduction returns the shortest non-zero vector among them. 

\medskip
\noindent We present the full pseudocode of the reduction in \cref{alg:2}.

\begin{algorithm}[ht!]
\caption{Reduction behind~\cref{thm:lowerbound}.}
\label{alg:lowerbound}\label{alg:2}
\KwIn{A lattice basis $\mathbf{B} = [\mathbf{b}_1,\dots,\mathbf{b}_n]$, parameters $N \coloneqq N(n)$ and $B \coloneqq B(n)$}
\KwOut{A lattice vector $\mathbf{v} \in \mathcal{L}(\mathbf{B})$}

Perform the lattice reduction algorithm in \cref{lem:short-basis} to obtain $\mathbf{D}$.
Set $W \coloneqq 8BN^3$.

Sample $\mathbf{x}^{(1)}, \dots, \mathbf{x}^{(N)}$ using the algorithm in \cref{thm:wide_sampling_shift_invariant} on $\mathbf{D}$ with $s = 4 \pi \sqrt{n} N \norm{\mathbf{D}}_2$.

Sample $n$ primes $r_1,\dots,r_n$ independently and uniformly from $[N^3]$.

Sample $n$ integers $s_1,\dots,s_n$ independently and uniformly from $[N^2]$.

\For{$\tilde{\lambda}_n = 2^{\lceil \log(\norm{\mathbf{D}}_2) \rceil}, \dots, 2^{\lfloor \log(\norm{\mathbf{D}}_2) - \log(B) \rfloor} $}{
    
    Set $a_i = f(\mathbf{x}^{(i)})$ for all $i \in [N]$    
    Obtain $X, Y$ by running the \PESS algorithm on the input $(a_1,\dots,a_N)$.
    
    Set $
    \mathbf{v}^{(\tilde{\lambda}_n)}
    \coloneqq
    \sum_{i \in X} \mathbf{x}^{(i)} - \sum_{j \in Y} \mathbf{x}^{(j)}$.
}

Output the shortest non-zero vector among the candidates $\left\{ \mathbf{v}^{(\tilde{\lambda}_n)} \right\}_{\tilde{\lambda}_n}$.
\end{algorithm}

\subsection{Correctness}

We prove correctness by analyzing a modified version of the reduction.  Recall that the reduction sets $s \coloneqq 4\pi\sqrt{n}\, N \norm{\mathbf{D}}_2$. 
By \cref{thm:wide_sampling_shift_invariant}, the actual sampling procedure used by the reduction produces vectors whose joint distribution is negligibly close to that of $N$ independent samples from $\mathcal{D}_{\L,s}$. Thus, it suffices to analyze an idealized version of the reduction in which $(\mathbf{x}^{(1)},\dots,\mathbf{x}^{(N)})$ are sampled independently from $\mathcal{D}_{\L,s}$. This discrepancy between the actual and idealized samplers contributes only a negligible additive error to the overall failure probability.

We now introduce a modified idealized reduction that differs only in its sampling step. First, sample $\mathbf{x}'^{(1)},\dots,\mathbf{x}'^{(N)}$ independently from $\mathcal{D}_{\L,s}$. Then, for each $i \in [N]$, sample an independent uniform bit $\tau_i \in \{0,1\}$ and define
\[
\mathbf{x}^{(i)} \coloneqq
\begin{cases}
\mathbf{x}'^{(i)} + \mathbf{v}^* & \text{if } \tau_i = 1,\\[1mm]
\mathbf{x}'^{(i)} & \text{if } \tau_i = 0,
\end{cases}
\]
where $\mathbf{v}^*$ is a fixed shortest nonzero vector of $\L$. All subsequent steps of the reduction remain unchanged.

This modified sampling procedure is not efficient, since it requires knowledge of $\mathbf{v}^*$. However, efficiency is irrelevant for the present argument, as we use the modified reduction only for the correctness analysis. Since $\|\mathbf{v}^*\|_2 = \lambda_1(\L) \leq \lambda_n(\L) \leq \tilde{\lambda}_n$, \cref{lem:stat-distance-DGS} implies that
\[
\Delta\!\left(\mathcal{D}_{\L,s},\, \mathcal{D}_{\L,s} + \mathbf{v}^*\right)
\leq \Oh((NB)^{-1}),
\]
where $\Delta(\cdot,\cdot)$ denotes statistical distance. Applying \cref{lem:stat-dist-prod-binary}, we obtain

\begin{align*}
    &\Delta\!\left(
    (\mathbf{x}^{(1)}, \dots, \mathbf{x}^{(N)}),
    (\mathbf{x}'^{(1)} + \tau_1\mathbf{v}^*, \dots, \mathbf{x}'^{(N)} + \tau_N\mathbf{v}^*)
    \right) \\
    &\leq \Oh(B^{-1}).
\end{align*}
We will set $B$ to be large polynomial, and hence $\Oh(B^{-1})$ will be $1/\poly(n)$.
In particular, the failure probability of the original reduction is at most the failure probability of the modified reduction plus an additive $\Oh(B^{-1})$ term. Combining this with the negligible discrepancy between the actual sampler and the idealized sampler, it remains to bound the failure probability of the modified reduction.

\noindent We do so in three steps with respect to the correct guess $\tilde{\lambda}_n$, i.e., $\lambda_n(\L) \leq \tilde{\lambda}_n \leq 2\lambda_n(\L)$.
\begin{enumerate}[label=(\arabic*)]
    \item\label{prop:linfty} We bound the $\ell_\infty$-norms of the sampled vectors $\mathbf{x}^{(i)}$ with sufficiently high probability. This will allow us to control the encoding range and determine a suitable choice of $W$.

    \item\label{prop:instance} We show that the resulting multiset of integers $(a_1,\dots,a_N)$ forms a valid \PESS instance, in particular that its total sum is at most $2^{o(N)}$. This ensures that the \PESS algorithm returns a valid solution, from which the reduction constructs a lattice vector of the desired length.

    \item\label{prop:nonzero} Finally, we show that the output lattice vector is nonzero with constant probability.
\end{enumerate}

\paragraph{Proof of~Property~\ref{prop:linfty}:}
We first bound the $\ell_\infty$-norms of the sampled lattice vectors. Recall that the reduction sets $s \coloneqq 4\pi\sqrt{n}\, N \norm{\mat D}_2$. By \cref{thm:wide_sampling_shift_invariant}, for any $R \geq s\sqrt{n\log N}$, the sampled vectors satisfy
\[
\|\mathbf{x}'^{(i)}\|_\infty \leq R
\qquad \text{for all } i \in [N],
\]
except with probability at most $2^{-n}$. Note that for all sufficiently large $n$, $B N^2 \tilde{\lambda}_n \geq s\sqrt{n\log N}$. Indeed, by substituting the value of $s$, this condition becomes
\[
B N^2 \tilde{\lambda}_n
\;\geq\;
4\pi\sqrt{n}\, N \norm{\mat D}_2 \cdot \sqrt{n\log N},
\]
which is equivalent to $N \geq 2\pi n \sqrt{\log N}$. This holds for all sufficiently large $n$ since we set $N = \Omega(n\log{n})$. Therefore, we conclude that, except with probability at most $2^{-n}$,
\[
\|\mathbf{x}'^{(i)}\|_\infty \leq B N^2 \tilde{\lambda}_n
\qquad \text{for all } i \in [N].
\]

In the remainder of the correctness analysis, we condition on this event. In particular, this bound will allow us to choose a suitable value of $W$ and to control the magnitudes of the encoded integers $a_i$.

\paragraph{Proof of Property~\ref{prop:instance}:}
We now show that, conditioned on the event
\[
\|\mathbf{x}'^{(i)}\|_\infty \leq B N^2 \tilde{\lambda}_n
\qquad\text{for all } i \in [N],
\]
the constructed multiset $\mathcal{A} \coloneqq \{a_1,\dots,a_N\}$ is a valid input to the \PESS algorithm, and that any equal-sum solution returned by \PESS yields a short lattice vector. Recall that the reduction sets $W \coloneqq 8BN^3$ and
\[
a_i = f(\mathbf{x}^{(i)})
\coloneqq
W^n + \sum_{j=1}^n W^{j-1}
\left\lfloor
\frac{(\mathbf{x}'^{(i)})_j/\tilde{\lambda}_n + s_j}{r_j}
\right\rfloor,
\]
where each $r_j$ is a prime in $[N^3]$ and each $s_j \in [N^2]$. We first bound the individual ``digits'' appearing in this encoding. Since $\| \mathbf{x}'^{(i)} \|_\infty \leq N^2B\tilde{\lambda}_n$ and $1 \leq s_j \leq N^2$, we have, for every $i \in [N]$ and $j \in [n]$,
\[
-N^2B + 1
\;\leq\;
\frac{(\mathbf{x}'^{(i)})_j}{\tilde{\lambda}_n} + s_j
\;\leq\;
N^2B + N^2.
\]
Since $r_j \geq 1$ and $\norm{\mat v^*}_{\infty} \leq \tilde{\lambda}_n$, it follows that
\begin{align*}
    -(N^2B+N^2+2)
    &\;\leq\;
    \left\lfloor
    \frac{(\mathbf{x}'^{(i)} + \mathbf{v^*})_j/\tilde{\lambda}_n + s_j}{r_j}
    \right\rfloor \\
    &\;\leq\;
    N^2B+N^2+1.
\end{align*}
In particular, every encoded digit has absolute value at most $\ell \coloneqq N^2B+N^2+2$. It follows that each $a_i$ is positive and satisfies
\begin{align*}
a_i
&\leq
W^n + \sum_{j=1}^n W^{j-1} \ell \\
&\leq
W^n + \ell \sum_{j=0}^{n-1} W^j \\
&\leq
W^n + \ell \cdot n W^{n-1}.
\end{align*}
Since $W = 8BN^3$ and $\ell = N^2B+N^2+2 \leq 2BN^2$ for all sufficiently large $N$, we have
\[
\ell \cdot n W^{n-1}
\leq
\frac{n}{4N} \cdot W^n
\leq
W^n
\]
for all sufficiently large $n$, because $N = \Omega(n\log{n})$. Hence $0 < a_i \leq 2W^n \eqqcolon M$ for all $i \in [N]$. 
Since $W = 8BN^3$, we have $M = 2(8BN^3)^n$. Using the fact that $N  \le B \ll 2^{N/n}$, it follows that $M \le 2^{o(N)}$. In other words, we have verified that the input instance lies in the pigeonhole regime required by the \PESS algorithm.

Since each $a_i$ is positive, the multiset $\mathcal{A}$ is a valid strong-regime instance for \PESS. Consequently, the \PESS algorithm returns unequal subsets $X,Y \subseteq [N]$ such that $\sum_{i\in X} a_i = \sum_{j\in Y} a_j$.

Define $\mathbf v \coloneqq \sum_{i\in X} \mathbf{x}^{(i)} - \sum_{j\in Y} \mathbf{x}^{(j)} \in \L$. We next show that the length of $\mathbf v$ is well-bounded. From the equality $\sum_{i\in X} a_i = \sum_{j\in Y} a_j$ and the definition of the encoding, we obtain
\begin{equation*}
    \begin{aligned}
        (|X|-|Y|)W^n + 
        \sum_{k=1}^n W^{k-1}\cdot \Bigg(
        \sum_{i\in X}
        \left\lfloor
        \frac{(\mathbf{x}^{(i)})_k/\tilde{\lambda}_n + s_k}{r_k}
        \right\rfloor - \sum_{j\in Y}
        \left\lfloor
        \frac{(\mathbf{x}^{(j)})_k/\tilde{\lambda}_n + s_k}{r_k}
        \right\rfloor \Bigg)
        =0.
    \end{aligned}
\end{equation*}
For each $k\in[n]$, define
\[
C_k \coloneqq
\sum_{i\in X}
\left\lfloor
\frac{(\mathbf{x}^{(i)})_k/\tilde{\lambda}_n + s_k}{r_k}
\right\rfloor
-
\sum_{j\in Y}
\left\lfloor
\frac{(\mathbf{x}^{(j)})_k/\tilde{\lambda}_n + s_k}{r_k}
\right\rfloor.
\]
Since each digit lies in $[-M,M]$, we have $|C_k| \leq (|X|+|Y|)M \leq 2NM$. Using $M \leq 2BN^2$ for large $N$, this gives$ |C_k| \leq 4BN^3 < \frac{W}{2}$. Thus, every coefficient in the base-$W$ expansion above has magnitude strictly less than $W/2$. It follows that no carries or overflows are possible, and hence each coefficient must vanish individually. In particular,
\[
|X|-|Y| = 0
\qquad\text{and}\qquad
C_k = 0 \quad \text{for every } k \in [n].
\]
Therefore, for every $k\in[n]$,
\[
\sum_{i\in X}
\left\lfloor
\frac{(\mathbf{x}^{(i)})_k/\tilde{\lambda}_n + s_k}{r_k}
\right\rfloor
=
\sum_{j\in Y}
\left\lfloor
\frac{(\mathbf{x}^{(j)})_k/\tilde{\lambda}_n + s_k}{r_k}
\right\rfloor.
\]

Using the elementary bound $t-1 \le \lfloor t \rfloor \le t$, we conclude that
\begin{align*}
    &\left|
    \sum_{i\in X}
    \frac{(\mathbf{x}^{(i)})_k/\tilde{\lambda}_n + s_k}{r_k}
    -
    \sum_{j\in Y}
    \frac{(\mathbf{x}^{(j)})_k/\tilde{\lambda}_n + s_k}{r_k}
    \right| \le |X|+|Y| \le 2N.
\end{align*}
Since $|X|=|Y|$, the terms involving $s_k$ cancel, and we obtain
\[
\left|
\frac{1}{r_k\tilde{\lambda}_n}
\left(
\sum_{i\in X} (\mathbf{x}^{(i)})_k
-
\sum_{j\in Y} (\mathbf{x}^{(j)})_k
\right)
\right|
\le 2N.
\]
By the definition of $\mathbf v$, this is
\[
\frac{|\mat v_k|}{r_k\tilde{\lambda}_n} \le 2N.
\]
Since $r_k \le N^3$, it follows that $|\mat v_k| \le 2N^4 \tilde{\lambda}_n$, for every $k \in [n]$. Hence
\[
\|\mathbf v\|_\infty \le 2N^4\tilde{\lambda}_n
\qquad\text{and}\qquad
\|\mathbf v\|_2 \le 2\sqrt{n}\,N^4\tilde{\lambda}_n.
\]

Thus, conditioned on the event that all sampled vectors satisfy $\|\mathbf{x}^{(i)}\|_\infty \le N^2B\tilde{\lambda}_n$, the constructed multiset $\mathcal A$ is a valid strong-regime \PESS instance, and the solution returned by \PESS yields a lattice vector $\mathbf v \in \L$ of norm at most $2\sqrt{n}\, N^4\tilde{\lambda}_n \leq 4\sqrt{n}\, N^4\lambda_n (\L) \leq N^5 \cdot \lambda_n(\Ll)$ as $N = \Omega(n\log{n})$.

\paragraph{Proof of Property~\ref{prop:nonzero}}
It remains to show that the lattice vector returned by the modified reduction is nonzero with constant probability. Recall that, in the modified reduction, each sampled vector is of the form
\[
\mathbf{x}^{(i)} = \mathbf{x}'^{(i)} + \tau_i \mathbf{v}^*,
\]
where $\mathbf{x}'^{(i)} \sim \mathcal{D}_{\L,s}$, each $\tau_i \in \{0,1\}$ is an independent uniform bit, and $\mathbf{v}^*$ is a fixed shortest nonzero vector of $\L$. We first bound the probability that the encoding changes when $\mathbf{v}^*$ is added.

\begin{claim}
    \label{cl:bad-prime-bound}
    Let $\alpha$ be uniformly random in $[N^2]$ and let $\beta$ be uniformly random among all primes in $[N^3]$. 
    For fixed $b,\eps \in \mathbb{R}$ satisfying $b \in [-BN^2, BN^2]$ and $\eps \in [-1,1]$, we have 
    \begin{align*}
        \Pr_{\alpha,\beta}\!\left[
        \left\lfloor \frac{b+\alpha}{\beta} \right\rfloor
        \neq
        \left\lfloor \frac{b+\eps+\alpha}{\beta} \right\rfloor
        \right] 
        \le
        \Oh\!\left(\frac{1}{N^2}+\frac{\log(BN^2+N^2)\log N}{N^3}\right).
    \end{align*}
    In particular, if $B\ge N$, then this is
    \[
    \Oh\!\left(\frac{1}{N^2}+\frac{\log(BN)\log N}{N^3}\right).
    \]
\end{claim}
\begin{claimproof}
For a fixed value of $\alpha$, define $x:= b+\alpha$. Consider the bad event
\[
\mathcal{E}_{\alpha,\beta}
:=
\left\{
\left\lfloor \frac{x}{\beta} \right\rfloor
\neq
\left\lfloor \frac{x+\eps}{\beta} \right\rfloor
\right\}.
\]
Since $\beta>0$, if
\[
\left\lfloor \frac{x}{\beta} \right\rfloor
\neq
\left\lfloor \frac{x+\eps}{\beta} \right\rfloor,
\]
then there exists an integer $k \in \mathbb{Z}$ such that
\[
k \in \left[\min\!\left\{\frac{x}{\beta},\frac{x+\eps}{\beta}\right\},
           \max\!\left\{\frac{x}{\beta},\frac{x+\eps}{\beta}\right\}\right].
\]
Multiplying by $\beta$, we obtain
$ k\beta \in J_\alpha$, where
$J_\alpha :=
\left[\min\{x,x+\eps\},\max\{x,x+\eps\}\right]$.
Thus, if $\mathcal{E}_{\alpha,\beta}$ occurs, then $J_\alpha$ contains an integer divisible by $\beta$. Since $|J_\alpha| = |\eps| \le 1$, the interval $J_\alpha$ contains at most two integers. We split into two cases.

\noindent \emph{Case 1 ($J_\alpha$ contains $0$).}
Then
$\alpha \in [-b-\max\{0,\eps\},\, -b-\min\{0,\eps\}]$,
and this interval has length at most $1$. Hence, it contains at most two integers, so
\[
\Pr_\alpha[\,0\in J_\alpha\,] \le \frac{2}{N^2} = \Oh\!\left(\frac{1}{N^2}\right).
\]

\medskip
\noindent \emph{Case 2 ($J_\alpha$ does not contain $0$).}
Fix such an $\alpha$. Since $J_\alpha$ contains at most two integers, there are at most two nonzero integers $z \in J_\alpha \cap \mathbb{Z}$ that could be equal to $k\beta$. Moreover, every such $z$ satisfies
\[
|z|
\le \max\{|x|,|x+\eps|\}
\le |b|+|\alpha|+|\eps|
\le BN + N^2 + 1.
\]
Each nonzero integer $z$ has at most $\log_2 |z| = \Oh(\log(BN^2+N^2))$ distinct prime divisors. 
Since there are at most two such integers, the total number of primes $\beta$ for which 
$\mathcal{E}_{\alpha,\beta}$ can occur is $\Oh(\log(BN^2+N^2))$. Now $\beta$ is chosen uniformly at random 
from the set of primes in $[N^3]$. By the prime number theorem,
\[
\pi(N^3)=\Theta\!\left(\frac{N^3}{\log(N^3)}\right)
=
\Theta\!\left(\frac{N^3}{\log N}\right).
\]
Therefore, for every fixed $\alpha$ with $0\notin J_\alpha$,
\begin{align*}
    \Pr_\beta[\mathcal{E}_{\alpha,\beta}]
    &\le
    \frac{\Oh(\log(BN^2+N^2))}{\Theta(N^3/\log N)} \\
    &= \Oh\!\left(\frac{\log(BN^2+N^2)\log N}{N^3}\right).
\end{align*}

Combining the two cases gives
\[
\Pr_{\alpha,\beta}[\mathcal{E}_{\alpha,\beta}]
\le
\Oh\!\left(\frac{1}{N^2}\right)
+
\Oh\!\left(\frac{\log(BN^2+N^2)\log N}{N^3}\right),
\]
which proves the claim. If $B\ge N$, then $BN^2+N^2 \le 2BN^2$, so $\log(BN^2+N^2)=\Oh(\log(BN))$, yielding the final bound.
\end{claimproof}

Fix a sample $\mathbf{x}'^{(i)}$ satisfying $\|\mathbf{x}'^{(i)}\|_\infty \le B N^2 \tilde{\lambda}_n$.
Then, for every coordinate $j \in [n]$,
\[
\frac{(\mathbf{x}'^{(i)})_j}{\tilde{\lambda}_n} \in [-N^2B,\, N^2B].
\]
Also, since $\mathbf{v}^*$ is a shortest nonzero vector, $|(\mathbf{v}^*)_j| \le \|\mathbf{v}^*\|_2 = \lambda_1(\L) \le \tilde{\lambda}_n$, and hence $\frac{(\mathbf{v}^*)_j}{\tilde{\lambda}_n} \in [-1,1]$. Applying Claim~\ref{cl:bad-prime-bound} with
\[
b = \frac{(\mathbf{x}'^{(i)})_j}{\tilde{\lambda}_n},
\qquad
\alpha = s_j,
\qquad
\beta = r_j,
\qquad
\eps = \frac{(\mathbf{v}^*)_j}{\tilde{\lambda}_n},
\]
we obtain, for every fixed $i$ and $j$,
\begin{align*}
    &\Pr_{r_j,s_j}\!\left[
    \left\lfloor \frac{(\mathbf{x}'^{(i)})_j/\tilde{\lambda}_n + s_j}{r_j} \right\rfloor
    \neq
    \left\lfloor \frac{((\mathbf{x}'^{(i)})_j + (\mathbf{v}^*)_j)/\tilde{\lambda}_n + s_j}{r_j} \right\rfloor
    \right] \\
    &\le
    \Oh\!\left(\frac{1}{N^2}+\frac{\log(BN)\log N}{N^3}\right).
\end{align*}
Taking a union bound over all coordinates $j \in [n]$, we get
\begin{align*}
    &\Pr_{\{r_j,s_j\}_{j\in[n]}}\!\left[
    f(\mathbf{x}'^{(i)}) \neq f(\mathbf{x}'^{(i)}+\mathbf{v}^*)
    \right] 
    \le
    \Oh\!\left(\frac{n}{N^2}+\frac{n\log(BN)\log N}{N^3}\right).
\end{align*}

Finally, taking a union bound over all $i \in [N]$, we conclude that
\begin{align*}
    &\Pr\!\left[
    \exists\, i \in [N] : f(\mathbf{x}'^{(i)}) \neq f(\mathbf{x}'^{(i)}+\mathbf{v}^*)
    \right] 
    \le
    \Oh\!\left(\frac{n}{N}+\frac{n\log(BN)\log N}{N^2}\right).
\end{align*}

\noindent \textbf{Remark:} Note that for every good set of sampled vectors $\{\mathbf{x}^{(i)}\}$ (\,i.e., $\|\mathbf{x}^{(i)}\|_\infty \le N^2B\tilde{\lambda}_n$ for all $i \in [N]$\,), there exists a choice of $\{r_j,s_j\}_{j\in[n]}$ satisfying the required properties. Since our reduction runs in subexponential time, it could, in principle, find such a choice deterministically. We nevertheless sample these values at random, as this leads to a simpler description of the reduction.

 \medskip
From this point onward, we condition on the following good event:
\begin{enumerate}
    \item $\|\mathbf{x}^{(i)}\|_\infty \le B N^2 \tilde{\lambda}_n$ for all $i \in [N]$, and
    \item $f(\mathbf{x}'^{(i)}) = f(\mathbf{x}'^{(i)}+\mathbf{v}^*)$ for all $i \in [N]$.
\end{enumerate}
By the previous parts of the proof, this conditioning fails with probability at most
\[
2^{-n}
+
\Oh\!\left(\frac{n}{N}+\frac{n\log(BN)\log N}{N^2}\right).
\]

We now show that, under this good event, the modified reduction outputs a nonzero vector with probability at least $1/2$.

Fix all randomness of the modified reduction except the bits $\tau_1,\dots,\tau_N$. Under the good event, changing a bit $\tau_i$ does not change the encoded value $a_i = f(\mathbf{x}'^{(i)} + \tau_i \mathbf{v}^*)$,
because $f(\mathbf{x}'^{(i)}) = f(\mathbf{x}'^{(i)}+\mathbf{v}^*)$. Hence the multiset $\mathcal{A} = \{a_1,\dots,a_N\}$ is independent of the choice of the bits $(\tau_1,\dots,\tau_N)$. Therefore, once the remaining randomness is fixed, the \PESS algorithm returns the same unequal subsets $X,Y \subseteq [N]$ for every choice of the bits.

Since $X \neq Y$, there exists an index $t \in (X\setminus Y)\cup(Y\setminus X)$. Consider two executions that differ only in the value of $b_t$. Because the input multiset $\mathcal{A}$ is unchanged, both executions produce the same subsets $X, Y$. Let $\mathbf v$ and $\mathbf v'$ denote the corresponding output vectors. Then flipping $b_t$ changes exactly one sampled vector by $\pm \mathbf{v}^*$, and therefore changes the output vector by $\pm \mathbf{v}^*$. More precisely,
\[
\mathbf v' - \mathbf v =
\begin{cases}
\mathbf{v}^* & \text{if } t \in X\setminus Y,\\
-\mathbf{v}^* & \text{if } t \in Y\setminus X.
\end{cases}
\]
In particular, $\mathbf v' \neq \mathbf v$. Since $\mathbf{v}^* \neq \mathbf 0$, at most one of $\mathbf v$ and $\mathbf v'$ can be equal to $\mathbf 0$. Thus, among each such paired choice of the bit $b_t$, at least one yields a nonzero output vector. Because $b_t$ is uniform in $\{0,1\}$, it follows that, conditioned on the good event, the modified reduction outputs a nonzero vector with probability at least $1/2$. Combining this with the earlier comparison between the original and modified reductions, we conclude that the original reduction fails with probability at most
\[
\frac{1}{2}
+
2^{-n}
+
\Oh\!\left(\frac{n}{N}+\frac{n\log(BN)\log N}{N^2}\right)
+ \PESS_{\mathsf{err}}.
\]
where $\PESS_{\mathsf{err}}$ is the probability that the \PESS algorithms fails to return a valid solution. Since
$\Omega(n\log{n}) \le N \le B \ll 2^{N/n}$,
and assuming the $\PESS_{\mathsf{err}}$ is negligible (or even smaller than than $1/2 - 1/\poly(n)$), the additive term above is $1 - 1/\poly(n)$. Therefore, by repeating the original reduction independently polynomially many times and returning the shortest nonzero vector obtained, the overall failure probability becomes negligible. This concludes the proof of Property~\cref{prop:nonzero}.

\subsection{Proof of~\cref{thm:lowerbound}}

In this section, we conclude the proof of~\cref{thm:lowerbound}. First, we prove a bit more precise statement for the $M^{\Omega(1/\log{N})}$ lower bound:

\begin{lemma}\label{lem:lowerbound-p1}
    Under~\cref{hyp:lambda-n}, there exists $\eps > 0$ such that there does not exist an algorithm that is given
    a set $S \subseteq \{1,\ldots,M\}$ of size $N$, with probability $>3/4$ finds two different $A, B \subseteq S$, and runs in time $\poly(N) \cdot M^{\eps/\log{N}}$ (even when $M = 2^{\Theta(\sqrt{N}\log{N})}$).
\end{lemma}

\begin{proof}
We set $N = n^2$. Let $c \coloneqq 10$. By~\cref{hyp:lambda-n}, there exists $\delta_H >
0$ such that there is no algorithm running in time $2^{\delta_H n}$ that given
an $n$-dimensional lattice $\Ll$ find a non-zero vector $v$ with $\norm{v} \le
n^{10} \cdot \lambda_n(\Ll)$.  Based on this, we select $\kappa \coloneqq C/(2\delta_H)$,
where $C$ is a constant in~\cref{lem:short-basis}, and
use~\cref{lem:short-basis} to compute a basis $\mathbf{D}$ in time
$\Oh(2^{C\,n/\kappa}) \le \Oh(2^{\delta_H n/2})$. This fixes the choice of $B =
C \cdot n^\kappa$.

Notice that indeed $\Omega(n\log{n}) \le N \le B \ll 2^{N/n}$, hence the
parameters are set up correctly. Recall that $W = 8BN^3$ and $M = 2W^n$, hence
$M \le 2^{2 \kappa n \log{n}}$, because we can assume that $\kappa,C \ll n$ as
otherwise the instance can be solved in constant time (as $\kappa$ and $C$ are
constant). Notice that $M = 2^{\Theta(\sqrt{N} \log{N})}$

We now analyze the running time of \cref{alg:2}
and derive the final contradiction to \cref{hyp:lambda-n}.

\paragraph{Running time of the reduction.}
We bound the running time of each stage of the algorithm, except that the oracle call to \PESS is:
\begin{itemize}[nosep]
    \item By \cref{lem:short-basis}, the initial basis-reduction step runs in time $\Oh(2^{\delta_H n/2})$ time.

    \item The sampling step uses \cref{thm:wide_sampling_shift_invariant} with $s = 4\pi\sqrt{n}\,N \norm{\mat D}_2$. Since $\|\mathbf{D}\|_2 \leq B \lambda_n(\L)$, we have $\log s = \tilde{\Oh}(n)$, and hence this step runs in time $N \cdot \poly(n\cdot \mathsf{bitlength}(\mathbf{D})\cdot \log s)$ which is polynomial in $n$.

    \item Sampling the primes $r_1,\dots,r_n$ and the shifts $s_1,\dots,s_n$ takes time $N^{\Oh(1)}$, since primality testing is polynomial-time in the bitlength of the input~\cite{AKS04}.

    \item Constructing the integers $a_1,\dots,a_N$ and the corresponding \PESS instance takes time polynomial in $N$, namely $N \cdot \poly(n\cdot\log W)$.

    \item The remaining vector-reconstruction and final minimization steps take time at most $\poly(n)$.
\end{itemize}
The loop over guesses of $\tilde{\lambda}_n$ runs for only $\poly(n\cdot\log\|\mathbf{D}\|_2)$ iterations, and this factor is absorbed into the $\Oh(\cdot)$ notation. All these steps, took time $\poly(n) \cdot 2^{\delta_H n/2}$. It remains to analyze the cost of the \PESS oracle. Assume for the sake of 
contradiction that for every $\eps > 0$, there exists an algorithm that solves \PESS in time $\poly(N) \cdot M^{\eps/\log{N}}$. Set $\eps = \delta_H/(4\kappa)$ and observe that 
we have $\poly(n)$ calls to the oracle, so in total this takes time:
\begin{align*}
    \poly(n) \cdot M^{\eps/\log{n}} &\le \poly(n) \cdot 2^{(2\kappa n\log{n}) \cdot \eps/\log{n}} \\
    &\le \poly(n) \cdot 2^{\delta_H \cdot n/2} \\
    &\le \Oh(2^{\delta_H \cdot n}).
\end{align*}
So in total, the reduction works in $\Oh(2^{\delta_H\cdot n})$ time. Finally, recall that
~\cref{alg:lowerbound} returns a non-zero vector $\mathbf{v} \in \Ll$ such that $\norm{\bold{v}} \le
N^5 \cdot \lambda_n(\Ll) = n^{10} \lambda_n(\Ll)$. This contradicts the initial assumption that no algorithm finds such a vector $\bold{v}$ in time $\Oh(2^{\delta_H \cdot n})$.
\end{proof}

We conclude the proof of~\cref{thm:lowerbound} by showing that a different tradeoff yields 
a $2^{\Omega(n/\log{n})}$ lower bound.

\begin{lemma}\label{lem:lowerbound-p2}
    Under~\cref{hyp:lambda-n}, there exists $\eps > 0$ such that there does not exist an algorithm that is given
    a set $S \subseteq \{1,\ldots,2^{N-1}\}$ of size $N$, with probability $>3/4$ finds two different $A, B \subseteq S$ and runs in time $\Oh(2^{\eps N/\log{N}})$.
\end{lemma}
\begin{proof}
    Set $c = 6$ and similarly, to the proof of~\cref{lem:lowerbound-p1}, this gives us parameters $\delta_H,\kappa$. Based on that, we select $N = 4 \kappa n \log{n}$ and similarly we select $B, W$ and $M$. As before, we know that $M \le 2^{2\kappa n\log{n}} < 2^{N-1}$ so the pigeonhole condition holds.
    
    Analogously to the proof of~\cref{lem:lowerbound-p1} the running time of~\cref{alg:lowerbound} is $\Oh(2^{\delta_H n/2})$, except for all the calls to the \PESS oracle, which if we set $\eps \coloneqq \delta_H/(4\kappa)$ take in total time:
    \begin{align*}
        \poly(n) \cdot 2^{\eps N/\log{N}} &\le 
        \poly(n) \cdot 2^{\eps \cdot 4 \kappa n \log{n}/(\log(4\kappa n\log{n}))} \\
        &\le \poly(n) \cdot 2^{\eps \cdot 4 \kappa n \log{n}/(2\log{n})} \\
        &\le \poly(n) \cdot 2^{2 \eps \kappa n } \\
        &\le \poly(n) \cdot 2^{\delta_H \cdot n/2} \\
        &\le 2^{\delta_H \cdot n}.
    \end{align*}
    So in total, the reduction works in $\Oh(2^{\delta_H\cdot n})$ time.
    Finally, recall that ~\cref{alg:lowerbound} returns a non-zero vector
    $\mathbf{v} \in \Ll$ such that $\norm{\bold{v}} \le N^5 \cdot \lambda_n(\Ll)
    = \Oh(n\log{n})^5 \lambda_n(\Ll) \le n^6 \cdot \lambda_n(\Ll)$ because $n$ is large enough. This contradicts
    the initial assumption that no algorithm can find such a vector $\bold{v}$ in time
    $\Oh(2^{\delta_H \cdot n})$.
\end{proof}

\section{Modular average-case: Proof of \cref{thm:average-case-modular}}\label{sec: average-case-modular}

\paragraph{Notations.} 
For an $n$ dimension vector $\vec{v}$ where $n$ is even, $\vec{v}^{(L)}$ and 
$\vec{v}^{(R)}$ denote the first $n/2$ and the last $n/2$
elements of the vector. $\mathbb{Z}_M$ denotes the ring of integers modulo $M$.

In this section, we conclude the proof of \cref{thm:average-case-modular}. Firstly, 
recall the definition of the \textbf{ACM-RESS} problem.

\paragraph{Average-Case Modular Restricted Equal Subset Sum (ACM-RESS).} 
\label{prob: average_case_modular}
Given a modulus $M$ and an integer $n$ divisible by $4$ as input, such that $M$ has the following structure
\begin{enumerate}
    \item $M = pqr$, where $p$, $q$, and $r$ are pairwise co-prime integers,
    \item $M = \Theta(2^n)$, and $p, q \in \Theta(M^{1/4})$.
    \footnote{As mentioned in the introduction, the average-case algorithms 
    work even for $M \leq \poly(n) \cdot 2^n$, the trade-off is the success probability of the algorithm. 
    We need to optimize the integers $p$ and $q$ for an optimal success probability. 
    But for simplicity of the presentation, we present an algorithm for $M = \Theta(2^n)$.}
\end{enumerate}
Choose $n$ elements $a_1, a_2, \ldots, a_n$ uniformly randomly from $\mathbb{Z}_{M}$
and the task is to find a non-trivial collision modulo $M$, i.e., find two \emph{distinct} sets 
$\vec{u}, \vec{v} \in \{0, 1\}^n$ and $(\vec{u}, \vec{v}) \neq (0^n, 0^n)$ 
such that $\vec{a} . \vec{u} \equiv \vec{a}. \vec{v} \mod M$,
where $.$ denotes the standard dot product.

\paragraph{Discussion on the modulus $M$.} 
Let $M < 2^n$, then we define $n' := \ceil{\log(M)}$.
And, we only consider any $n'$ elements from the list $\{a_1, a_2, \ldots, a_n\}$. Then, 
our problem reduces to a problem with input modulus $M$ and integer $n'$, where
we can observe that $2^{n' - 1} < M \leq 
2^{n'}$. 

Here, we first introduce the Filtering \cref{alg: averange-case-modular}, 
which, on average, solves \textbf{ACM-RESS} problem.

\begin{algorithm}[h!]
\caption{Algorithm for~\cref{thm:average-case-modular}.}
\label{alg: averange-case-modular}
\KwIn{$(a_1, a_2, \ldots, a_n) \in \mathbb{Z}_{M}^n$; 
\textbf{Parameters:} $s_0, s_1 \in \mathbb{Z}_p$ and $s_2 \in \mathbb{Z}_q$}
\KwOut{A non-trivial modular collision or $\bot$}
Build two mod-$p$ DP tables storing the sizes of the following sets $\forall$ $s \in \mathbb{Z}_p$ 
and $1 \leq i \leq n/2$
\begin{align*}
    B_{i, s}^{(L)} &:= \left\{ \vec{u} \in \{0, 1\}^i\; \middle| \; 
    \sum_{j = 1}^{i}a_ju_j \equiv s \bmod p \right\}, \\
    B_{i, s}^{(R)} &:= \left\{ \vec{v} \in \{0, 1\}^i \; \middle| \; 
    \sum_{j = 1}^{i}a_{n/2 + j}v_j \equiv s \bmod p \right\}.
\end{align*}

Compute the bins $B_{n/2, s_0}^{(L)}$ and $B_{n/2, s_1}^{(R)}$ by backtracing the DP table; look for a collision in these sets, else continue.

Construct the following set
\begin{equation*}
    \begin{aligned}
        B_{s_0, s_1, s_2} &:= 
        \left\{ \vec{u} || \vec{v} \in \{0, 1\}^n \; \middle| \; 
        \begin{aligned}
            \vec{u} \in B_{n/2, s_0}^{(L)}, \;  \vec{v} \in B_{n/2, s_1}^{(R)}, \\
            \vec{a}^{(L)} \cdot \vec{u} + \vec{a}^{(R)} \cdot \vec{v} \equiv s_2 \bmod q,
        \end{aligned}\right\}
    \end{aligned}
\end{equation*}

Look for a non-trivial collision in this list $B_{s_0, s_1, s_2}$.
If one is found, it outputs this collision; otherwise, it outputs $\bot$.
\end{algorithm}

\paragraph{Observation.} Using \textbf{CRT} (Chinese Remainder Theorem), we can say that every input $a_i$ can be split into three independent random variables $a_i^{(p)}, a_i^{(q)}$, and $a_i^{(r)}$, which are images of the CRT isomorphism. Then, we can first observe that the filtering \cref{alg: averange-case-modular} only depends on the random vectors $\vec{a}^{(p)} := (a_1^{(p)}, a_2^{(p)}, \ldots, a_n^{(p)})$ and 
$\vec{a}^{(q)} := (a_1^{(q)}, a_2^{(q)}, \ldots, a_n^{(q)})$. 

\subsection{Success Probability}
The filtering algorithm succeeds when it outputs a desired collision for the input instance. 
We are interested in computing a lower bound of the fraction of all the input 
instances for which the algorithm succeeds. The following theorem gives us the success
probability of \cref{alg: averange-case-modular}.
\begin{lemma}\label{lemma: Filtering_algo_succ_modular}
    Let $M, n, p,$ and $q$ be the same as defined in the above problem \textbf{ACM-RESS}. 
    Then the following holds for the \cref{alg: averange-case-modular}
    \begin{align*}
        \Pr_{\vec{a}}\left(\cref{alg: averange-case-modular} \text{ succeeds}\right) 
        = \Omega\left(1\right).
    \end{align*}
\end{lemma}

In the remainder of this section, we give a proof of \cref{lemma: Filtering_algo_succ_modular}.
We first define a collision-finding random variable, and then, by using
second-moment-based concentration bounds, we obtain the desired result. 

\paragraph{Distribution of collisions in $B_{s_0, s_1, s_2}$.} 
To analyze the distribution of collisions in the list $B_{s_0, s_1, s_2}$, 
we first define a random variable $\mathcal{C}$ that denotes the number of non-trivial 
modular collisions in this list. Our filtering algorithm succeeds if the value of these variables is at least $1$. Hence, to find the success probability of the algorithm, i.e., the fraction of instances for which the algorithm outputs a collision, we need to find the probability of the event $\mathcal{C} > 0$. Formally,
we define the random 
$\mathcal{C}$ over a random choice of $\vec{a}$
as follows
\begin{equation*}
    \begin{aligned}
        \mathcal{C} := \left| \left\{ (\vec{u}_i, \vec{u}_j) \; \middle| \;  
        \begin{aligned}
            i < j, \; \vec{u}_i, \vec{u}_j \in  B_{s_0, s_1, s_2}, \\
            \vec{a}^{(r)} \cdot \vec{u}_i \equiv \vec{a}^{(r)}\cdot 
            \vec{u}_j \mod r
        \end{aligned}\right\}\right|,
    \end{aligned}
\end{equation*}

where $\vec{a}^{(r)} := \{a_i \mod r\}_{i = 1}^{n}$.
And, we have
\begin{align*}
    \Pr_{\vec{a}}\left(\cref{alg: averange-case-modular} \text{ succeeds}\right) = 
    \Pr_{\vec{a}}\left( \mathcal{C} > 0 \right).
\end{align*}

Hence, now our problem reduces to analyzing this random variable $\mathcal{C}$. 
For any input instance, we can compute the value of $\mathcal{C}$ 
by iterating through all the possible pairs of binary vectors, 
and considering whether the vectors (or the corresponding subset sum) pass through the 
filters of the algorithm and collide modulo $r$. We define the following set that contains all
the possible binary vectors that can reach the final list $B_{s_0, s_1, s_2}$.
\begin{align*}
    \mathsf{A} := \{\vec{u} \in \{0, 1\}^n \; | \; 
    \vec{u}^{(L)} \neq 0^{n/2}, \vec{u}^{(R)} \neq 0^{n/2}\}.
\end{align*}
Let's assume that the elements of the set $\mathsf{A}$ 
are ordered in a lexicographical sense. Then let 
$\mathsf{A} = \{\vec{u}_1, \vec{u}_2, \ldots \vec{u}_N\}$, where $\vec{u}_i < \vec{u}_j$ for any $1 \leq i < j \leq N$ 
and $N = (2^{n/2} - 1)^2 = 2^n - 2^{n/2 + 1} + 1$.

A binary vector $\vec{u}_i$ survives the filtering algorithm if it satisfies
the following events
\begin{equation*}
    \begin{aligned}
        \mathsf{E}_{\mathsf{LS}, p}^{(i)} &: \vec{u}_i^{(L)}\cdot 
        \vec{a}^{(p)}{}^{(L)} \equiv s_0 \mod p && (\text{Left sum Mod $p$})\\
        \mathsf{E}_{\mathsf{RS}, p}^{(i)} &: \vec{u}_i^{(R)}\cdot \vec{a}^{(p)}{}^{(R)} \equiv s_1 \mod p && (\text{Right sum Mod $p$})\\
        \mathsf{E}_{\mathsf{TS}, q}^{(i)} &: \vec{u}_i\cdot\vec{a}^{(q)} \equiv s_2 \mod q
        && (\text{Total sum Mod $q$})
    \end{aligned} 
\end{equation*}

Now, for each pair of binary vectors $\vec{u}_i, \vec{u}_j$, we define the binary random variable
and $\mathcal{C}_{ij}$,
such that its value is $1$ if and only if $\vec{u}_i$ and $\vec{u}_j$ survive the
filtering algorithm, i.e., the above three events occur for both the vectors $\vec{u}_i$ and 
$\vec{u}_j$, and they collide modulo $r$, i.e., the following event occurs
\begin{equation*}
    \begin{aligned}
        \mathsf{E}_{\mathsf{col}, r}^{(ij)} &: \vec{u}_i\cdot \vec{a}^{(r)} \equiv \vec{u}_j\cdot\vec{a}^{(r)} \mod r.
    \end{aligned}
\end{equation*}

Now, we write the random variable $\mathcal{C}$ as a sum of the binary random variables $\mathcal{C}_{ij}$
as follows
\begin{align*}
    \mathcal{C} = \sum_{1 \leq i < j \leq N} \mathcal{C}_{ij}.
\end{align*}
Then, using the linearity of expectation, we have the following
\begin{align*}
    \underset{\vec{a}}{\mathbb{E}}\left(\mathcal{C}\right) = \sum_{i < j} \underset{\vec{a}}{\mathbb{E}}\left(\mathcal{C}_{ij}\right).
\end{align*}

Using this equation, we compute the first and second moments of the random variable $\mathcal{C}$ to prove~\cref{lemma: Filtering_algo_succ_modular}.

\subsubsection{First Moment} 
\begin{lemma}\label{lemma: first_moment_modular}
    Let $M, n, p,$ and $q$ be the same as defined in the above problem \textbf{ACM-RESS}. 
    Then the following holds for the \cref{alg: averange-case-modular}
    \begin{align*}
        \underset{\vec{a}}{\mathbb{E}}\left(\mathcal{C}\right) = \Theta\left(1\right).
    \end{align*}
\end{lemma}
\begin{proof}
    Let $v \in \{0, 1\}^{n} \setminus \{0^n\}$, then we observe the following
    \begin{equation*}
        \begin{aligned}
            \Pr_{\vec{a}^{(p)}}\left(\vec{v} \cdot \vec{a}^{(p)} \equiv \alpha \mod p\right) = \frac{1}{p},\\
            \Pr_{\vec{a}^{(q)}}\left(\vec{v} \cdot \vec{a}^{(q)} \equiv \beta \mod q\right) = \frac{1}{q},\\
            \Pr_{\vec{a}^{(r)}}\left(\vec{v} \cdot \vec{a}^{(r)} \equiv \gamma \mod r\right) = \frac{1}{r},
        \end{aligned}
    \end{equation*}
    $\forall$ $\alpha \in \mathbb{Z}_p$, $\beta \in \mathbb{Z}_q$, and $\gamma \in \mathbb{Z}_r$.

    Since the vectors $\vec{a}^{(p)}$, $\vec{a}^{(q)}$, and $\vec{a}^{(r)}$ are 
    mutually independent, we can observe that the following four events are mutually independent 
    \begin{equation}\label{eq: mutually_ind_events}
        \begin{aligned}
            \mathsf{E}_{\mathsf{LS}, p}^{(i)} \wedge \mathsf{E}_{\mathsf{LS}, p}^{(j)}, \;
            \mathsf{E}_{\mathsf{RS}, p}^{(i)} \wedge \mathsf{E}_{\mathsf{RS}, p}^{(j)}, \;
            \mathsf{E}_{\mathsf{TS}, q}^{(i)} \wedge \mathsf{E}_{\mathsf{TS}, q}^{(j)}, \;
            \mathsf{E}_{\mathsf{col}, r}^{(ij)}.
        \end{aligned}
    \end{equation}

    And, since $\vec{u}_i \neq \vec{u}_j$, the events dependent on the vector $\vec{u}_i$
    are independent of the events dependent on the vector $\vec{u}_j$. 
    Combining this observation with the above equation, we have the following
    \begin{align*}
        \underset{\vec{a}}{\mathbb{E}}\left(\mathcal{C}_{ij}\right) =
        \begin{cases}
             \frac{1}{p^3q^2r} & \quad \vec{u}_{i}^{(L)} = \vec{u}_{j}^{(L)}, 
             \vec{u}_{i}^{(R)} \neq \vec{u}_{j}^{(R)}\\
             \frac{1}{p^3q^2r} & \quad \vec{u}_{i}^{(L)} \neq \vec{u}_{j}^{(L)}, 
             \vec{u}_{i}^{(R)} = \vec{u}_{j}^{(R)}\\
             \frac{1}{p^4q^2r} & \quad \vec{u}_{i}^{(L)} \neq \vec{u}_{j}^{(L)}, 
             \vec{u}_{i}^{(R)} \neq \vec{u}_{j}^{(R)}
        \end{cases}.
    \end{align*}
    Then, we have the following 
    \begin{equation}\label{eq:Exp_C}
        \begin{aligned}
            \underset{\vec{a}}{\mathbb{E}}\left(\mathcal{C}\right) &= \sum_{i < j} \underset{\vec{a}}{\mathbb{E}}\left(\mathcal{C}_{ij}\right)\\
            &= \frac{1}{2} \sum_{i \neq j} 
            \underset{\vec{a}}{\mathbb{E}}\left(\mathcal{C}_{ij}\right)\\
            &= \frac{1}{2}\cdot \left(\frac{2N_1}{p^3q^2r} + \frac{N_2}{p^4q^2r}
            \right),
        \end{aligned}
    \end{equation}
    where 
    \begin{equation*}\label{eq:N_1}
        \begin{aligned}
            N_1 &= \left(2^{n/2} - 1\right)\cdot\left(2^{n/2} - 1\right)
            \cdot\left(2^{n/2} - 2\right) \\
            &= \Theta\left(2^{3n/2}\right),
        \end{aligned}
    \end{equation*}
    and 
    \begin{equation*}\label{eq:N_2}
        \begin{aligned}
            N_2 &= \left(\left(2^{n/2} - 1\right)\cdot\left(2^{n/2} - 2\right)\right)^2
            = \Theta\left(2^{2n}\right).
        \end{aligned}
    \end{equation*}
    In the computation of $N_1$ and $N_2$, we have used the fact that $\vec{u}_{i}^{(L)}$, $\vec{u}_{i}^{(R)}$,
    $\vec{u}_{j}^{(L)}$, and $\vec{u}_{j}^{(R)}$ cannot be equal to $0^{n/2}$.
    
    \paragraph{Lower Bound of 
    $\underset{\vec{a}}{\mathbb{E}}\left(\mathcal{C}\right)$.}
    
    Using equation $(\ref{eq:Exp_C})$, $M = pqr = \Theta\left(2^n\right)$, and
    $p,q = \Theta\left(M^{1/4}\right)$, we get
    \begin{equation*}\label{eq:first_moment_lower_bound}
        \begin{aligned}
            \underset{\vec{a}}{\mathbb{E}}\left(\mathcal{C}\right) 
            &= \frac{N_1}{p^3q^2r} + \frac{N_2}{2p^4q^2r} \\
            &= \Theta\left(\frac{2^{3n/2}}{p^2qM}\right) + \Theta\left(\frac{2^{2n}}{p^3qM}\right)\\
            &= \Theta\left(\frac{2^{3n/2}}{2^{7n/4}}\right) + 
            \Theta\left(\frac{2^{2n}}{2^{2n}}\right)\\
            &= \Theta(1). 
        \end{aligned}
    \end{equation*}
    This completes the proof of \cref{lemma: first_moment_modular}.
\end{proof}

\subsubsection{Second Moment} 
\begin{lemma}\label{lemma: second_moment_modular}
    Let $M, n, p,$ and $q$ be the same as defined in the above problem \textbf{ACM-RESS}. 
    Then the following holds for the \cref{alg: averange-case-modular}
    \begin{align*}
        \underset{\vec{a}}{\mathbb{E}}\left(\mathcal{C}^2\right) = 
        \mathcal{O}\left(1\right). 
    \end{align*}
\end{lemma}

\begin{proof}
    \begin{equation*}\label{eq:second_moment_1}
        \begin{aligned}
            \underset{\vec{a}}{\mathbb{E}}\left(\mathcal{C}^2\right) &= 
            \underset{\vec{a}}{\mathbb{E}}\left(\sum_{i < j}\mathcal{C}_{ij}.\sum_{l < k}\mathcal{C}_{lk}\right)\\
            &= \sum_{i < j, l < k} \underset{\vec{a}}{\mathbb{E}}\left(\mathcal{C}_{ij}\mathcal{C}_{lk}\right) \\
            &= \sum_{S} \Pr_{\vec{a}} \left(\mathcal{C}_{ij} = 1 \land \mathcal{C}_{lk} = 1\right) + 
            \underset{\vec{a}}{\mathbb{E}}\left(\mathcal{C}\right),
        \end{aligned}
    \end{equation*}
    where $S := \{ \left(\vec{u}_i, \vec{u}_j, \vec{u}_l, \vec{u}_k \right) \; | \; i < j, l < k, (i, j) \neq (l, k)\}$.
    Hence, we need to analyze the probability that the following events occur simultaneously
    \begin{enumerate}
        \item \textit{Left sum events:} 
        \begin{equation*}
            \begin{aligned}
                \mathsf{E}_{\mathsf{LS}, p}^{(i)} &: \vec{u}_i^{(L)}\cdot\vec{a}^{(p)}{}^{(L)} \equiv s_0 \mod p,\\
                \mathsf{E}_{\mathsf{LS}, p}^{(j)} &: \vec{u}_j^{(L)}\cdot\vec{a}^{(p)}{}^{(L)} \equiv s_0 \mod p,\\
                \mathsf{E}_{\mathsf{LS}, p}^{(l)} &: \vec{u}_l^{(L)}\cdot\vec{a}^{(p)}{}^{(L)} \equiv s_0 \mod p,\\
                \mathsf{E}_{\mathsf{LS}, p}^{(k)} &: \vec{u}_k^{(L)}\cdot\vec{a}^{(p)}{}^{(L)} \equiv s_0 \mod p,
            \end{aligned}
        \end{equation*}
        \item \textit{Right sum events:}
        \begin{equation*}
            \begin{aligned}
                \mathsf{E}_{\mathsf{RS}, p}^{(i)} &: \vec{u}_i^{(R)}\cdot\vec{a}^{(p)}{}^{(R)} \equiv s_1 \mod p,\\
                \mathsf{E}_{\mathsf{RS}, p}^{(j)} &: \vec{u}_j^{(R)}\cdot\vec{a}^{(p)}{}^{(R)} \equiv s_1 \mod p,\\
                \mathsf{E}_{\mathsf{RS}, p}^{(l)} &: \vec{u}_l^{(R)}\cdot\vec{a}^{(p)}{}^{(R)} \equiv s_1 \mod p,\\
                \mathsf{E}_{\mathsf{RS}, p}^{(k)} &: \vec{u}_k^{(R)}\cdot\vec{a}^{(p)}{}^{(R)} \equiv s_1 \mod p,
            \end{aligned}
        \end{equation*}
        \item \textit{Total sum events:}
        \begin{equation*}
        \begin{aligned}
            \mathsf{E}_{\mathsf{TS}, q}^{(i)} &: \vec{u}_i\cdot\vec{a}^{(q)} \equiv s_2 \mod q, \\
            \mathsf{E}_{\mathsf{TS}, q}^{(j)} &: \vec{u}_j\cdot\vec{a}^{(q)} \equiv s_2 \mod q, \\
            \mathsf{E}_{\mathsf{TS}, q}^{(l)} &: \vec{u}_l\cdot\vec{a}^{(q)} \equiv s_2 \mod q, \\
            \mathsf{E}_{\mathsf{TS}, q}^{(k)} &: \vec{u}_k\cdot\vec{a}^{(q)} \equiv s_2 \mod q, \\
            \mathsf{E}_{\mathsf{col}, r}^{(ij)} &: \vec{u}_i\cdot\vec{a}^{(r)} \equiv \vec{u}_j.\vec{a}^{(r)} \mod r\\
            \mathsf{E}_{\mathsf{col}, r}^{(lk)} &: \vec{u}_l\cdot\vec{a}^{(r)} \equiv \vec{u}_k\cdot\vec{a}^{(r)} \mod r,
        \end{aligned}
    \end{equation*}
    \end{enumerate}

    As we mentioned in the technical overview, the main difficulty in analyzing the second moment stems from dependencies between 
    the above-mentioned events, that is induced by linear relationships among the binary vectors. 
    To address this issue, we partition the set 
    $S$ of all possible $4-$tuples into cells such that, for every tuple within a given cell, the joint probability
    $$\Pr_{\vec{a}} \left(\mathcal{C}_{ij} = 1 \land \mathcal{C}_{lk} = 1\right)$$
    admits a uniform upper bound, while the size of each cell can also be tightly bounded. This structured partitioning allows us to control the contribution of dependent events to the second moment. 
    As we have discussed previously, the events mod $r$ are independent of the events mod $p$ and $q$. 
    Hence, first, we analyze them separately. Consider the following claim
    
    \begin{claim}\label{claim: Events_a_r_second_moment}
    Let $r \ge 2$ be an integer, let $\vec{w} := \vec{u}_i - \vec{u}_j$ and
    $\vec{w}' := \vec{u}_l - \vec{u}_k$, and let $\vec{a}^{(r)}$ be uniform on $\Z_r^n$.
    Write $P := \Pr_{\vec{a}^{(r)}}\big[\vec{w}\cdot\vec{a}^{(r)} \equiv_r 0 \wedge
    \vec{w}'\cdot\vec{a}^{(r)} \equiv_r 0\big]$.
    \begin{enumerate}
        \item If $\vec{w} \equiv \pm \vec{w}' \pmod r$, then $P = 1/r$.
        \item If $\vec{w} \not\equiv \pm \vec{w}' \pmod r$, then $P \in \{1/r^2,\, 2/r^2\}$;
        moreover $P = 1/r^2$ whenever $r$ is odd.
    \end{enumerate}
    \end{claim}
    
    \begin{proof}
    refer to \cref{sec: Events_a_r_second_moment}. 
    \end{proof}
    
    Using this claim, we can rewrite the second moment as follows
    \begin{equation}\label{eq:second_moment_2}
        \begin{aligned}
            \underset{\vec{a}}{\mathbb{E}}\left(\mathcal{C}^2\right) &=
            \sum_{\substack{i < j, l < k \\ (i, j) \neq (l, k)}} 
            \Pr_{\vec{a}} \left(\mathcal{C}_{ij} = 1 \land \mathcal{C}_{lk} = 1\right) + \underset{\vec{a}}{\mathbb{E}}\left(\mathcal{C}\right)\\
            &= S_1 + S_2 + \underset{\vec{a}}{\mathbb{E}}\left(\mathcal{C}\right),
        \end{aligned}
    \end{equation}
    where
    \begin{equation*}
        \begin{aligned}
            S_1 &= \frac{1}{r} \cdot 
            \sum_{\substack{i < j, l < k \\  (i, j) \neq (l, k)
            \\ \vec{u}_i - \vec{u}_j = \pm(\vec{u}_l - \vec{u}_k)
            }} \Pr_{\vec{a}^{(pq)}} \left(\mathcal{D}_{ij} = 1 \land \mathcal{D}_{lk} = 1\right), \\
            S_2 &\leq \frac{2}{r^2} \cdot \sum_{\substack{i < j, l < k
            \\  (i, j) \neq (l, k)
            \\ \vec{u}_i - \vec{u}_j \neq \pm(\vec{u}_l - \vec{u}_k)
            }} \Pr_{\vec{a}^{(pq)}} \left(\mathcal{D}_{ij} = 1 \land \mathcal{D}_{lk} = 1\right),
        \end{aligned}
    \end{equation*}
    and $\mathcal{D}_{ij}$ is a binary random variable equal to $1$ if and only if events $\mathsf{E}_{\mathsf{LS}, p}^{(i)}$,
    $\mathsf{E}_{\mathsf{RS}, p}^{(i)}$, $\mathsf{E}_{\mathsf{LS}, p}^{(j)}$, $\mathsf{E}_{\mathsf{RS}, p}^{(j)}$, 
    $\mathsf{E}_{\mathsf{TS}, q}^{(i)}$, and $\mathsf{E}_{\mathsf{TS}, q}^{(j)}$ occur simultaneously. 
    And, the random variable is $\mathcal{D}_{lk}$ is defined similarly. These two random variables are defined over the random choice
    of $\vec{a}^{(p)}$ and $\vec{a}^{(q)}$.
    
    Now, consider the following matrices 
    \begin{equation*}
        \begin{aligned}
            \mat{U}_{ijlk}^{(L)} &:= \left[\vec{u}_i^{(L)} \quad \vec{u}_j^{(L)} \quad
            \vec{u}_l^{(L)} \quad \vec{u}_k^{(L)}\right] , \\
            \mat{U}_{ijlk}^{(R)} &:= \left[\vec{u}_i^{(R)} \quad \vec{u}_j^{(R)} \quad
            \vec{u}_l^{(R)} \quad \vec{u}_k^{(R)}\right] , \\
            \mat{U}_{ijlk} &:= \left[\vec{u}_i \quad \vec{u}_j \quad
            \vec{u}_l \quad \vec{u}_k\right] = \left[\mat{U}_{ijlk}^{(L)} 
            \quad \mat{U}_{ijlk}^{(R)}\right],
        \end{aligned}
    \end{equation*}
    where the vectors $\vec{u}_{t}$ $\forall$ $t \in \{i, j, l, k\}$ are considered as column vectors. 
    For the following analysis, matrices $\mat{U}_{ijlk}^{(L)}$ and $\mat{U}_{ijlk}^{(R)}$ are consider over the ring $\mathbb{Z}_p$, and the matrix $\mat{U}_{ijlk}$ over the ring $\mathbb{Z}_q$. Then, we split the set $S$ based on the rank of the above matrices. Let $B_{t_1, t_2, t_3}$ be the set of $4$-tuple 
    $\left(\vec{u}_i, \vec{u}_j, \vec{u}_l, \vec{u}_k \right)$ such that 
    $i < j$, $l < k$, $(i, j) \neq (l, k)$, 
    $\mathsf{rank}\left(\mat{U}_{ijlk}\right) = t_3$,
    $\mathsf{rank}\left(\mat{U}_{ijlk}^{(L)}\right) = t_1$, and 
    $\mathsf{rank}\left(\mat{U}_{ijlk}^{(R)}\right) = t_2$.

    Now, using linear algebra techniques, we prove the following lemma that gives us 
    an upper bound on the size of the set $B_{t_1, t_2, t_3}$.
    \begin{proposition}\label{prop: B_t_1_t_2_t_3_size}
        For any $1 \leq t_1, t_2, t_3 \leq 4$, the following holds
        \[
        \left|B_{t_1, t_2, t_3}\right|
        \;=\;
        \mathcal{O}\!\left(2^{\frac{(t_1 + t_2)n}{2}}\right).
        \]
        And, since $t_3 \geq \max\{t_1, t_2\}$, we have 
        $\left|B_{t_1, t_2, t_3}\right| = 
        \mathcal{O}\!\left(2^{t_3n}\right)$.
    \end{proposition}
    \begin{proof}
        To upper bound $\left|B_{t_1, t_2, t_3}\right|$, we count the number
        of matrices $\mat{U}_{ijlk}^{(L)}$ and $\mat{U}_{ijlk}^{(R)}$ such that
        \[
        \rank\!\left(\mat{U}_{ijlk}^{(L)}\right) = t_1,
        \qquad
        \rank\!\left(\mat{U}_{ijlk}^{(R)}\right) = t_2.
        \]
    
        \begin{proposition}\label{prop: counting-rank-t-matrices}
            Let $\mathsf{M}_{t} := \{\mat{M} \in \mathbb{Z}_s^{n \times m} \; 
            | \; \mathsf{rank}(\mat{M}) = t,\ \mat{M}_{ij} \in \{0,1\}\}$ where $n \ge m$
            and $s$ is a prime. 
            Then
            \[
            \left|\mathsf{M}_{t}\right|
            = \BigO\!\left( \left(t^{t^2}\right)^{\,m - t} \cdot 2^{tn} \cdot m! \right).
            \]
            In particular, if $m$ is a constant, then $|\mathsf{M}_t| = \BigO(2^{tn})$.
        \end{proposition}
        \begin{proof}
            refer to \cref{sec: counting-rank-t-matrices}.
        \end{proof}
        Considering all concatenations of such matrices yields a set that contains all valid
        $\mat{U}_{ijlk}$, by \cref{prop: counting-rank-t-matrices}, the number of
        $\mat{U}_{ijlk}^{(L)}$ with rank $t_1$ is
        $\mathcal{O}(2^{t_1 n/2})$, and similarly the number of
        $\mat{U}_{ijlk}^{(R)}$ with rank $t_2$ is
        $\mathcal{O}(2^{t_2 n/2})$. Therefore, $
        \left|B_{t_1, t_2, t_3}\right|
        \;=\;
        \mathcal{O}\!\left(2^{\frac{(t_1 + t_2)n}{2}}\right).$
    \end{proof}

    Now that we have a good upper bound on the size of the set $B_{t_1, t_2, t_3}$, 
    in the following lemma, we give a relation between the desired probability
    % we give an estimate of the following 
    $$\Pr_{\vec{a}^{(pq)}}\!\left(\mathcal{D}_{ij} = 1 \wedge \mathcal{D}_{lk} = 1\right)$$
    and the rank of the matrices $\mat{U}_{ijlk}$, $\mat{U}_{ijlk}^{(L)}$ and $\mat{U}_{ijlk}^{(R)}$.
    This is a crucial lemma since it allows us to compute bounds on the probabilities using 
    simple linear algebra. 
    \begin{lemma}\label{lemma: prob-counting-bad-vector-set} 
        For any 
        $\big(\vec{u}_i, \vec{u}_j, \vec{u}_l, \vec{u}_k\big) \in B_{t_1, t_2, t_3}$, the following holds
        \[
        \Pr_{\vec{a}^{(pq)}}\!\left(\mathcal{D}_{ij} = 1 \wedge \mathcal{D}_{lk} = 1\right)
        \;\leq\;
        \frac{1}{p^{t_1}\cdot p^{t_2}\cdot q^{t_3}}.
        \]
    \end{lemma}
    \begin{proof}
        refer \cref{sec: prob-counting-bad-vector-set}.
    \end{proof}
    
    In ~\cref{prop: B_t_1_t_2_t_3_size}, we give an upper bound for the size of the set $B_{t_1, t_2, t_3}$. 
    In the next proposition, using a sophisticated combinatorial technique,
    we give another upper bound for the same when the vectors 
    satisfy a specific relation. 
    
    \begin{proposition}\label{prop:counting_collinear_tuples}
        For some $m \in \mathbb{Z}^{\geq 1}$, let 
        \begin{equation*}
            \begin{aligned}
                \mathsf{T_m} := \left\{(\vec{v}_1, \vec{v}_2, \vec{v}_3, \vec{v}_4) \; \middle | \; 
                \begin{aligned}
                    \vec{v}_1 \neq \vec{v}_2, \vec{v}_3 \neq \vec{v}_4, \\  
                    \vec{v}_1 - \vec{v}_2 = \vec{v}_3 - \vec{v}_4 
                \end{aligned}\right\},
            \end{aligned}
        \end{equation*}
        where $\vec{v}_i \in \{0, 1\}^m \setminus \{0^m\}$, $\forall$ $1 \leq i \leq 4$.
        Then, $|\mathsf{T_m}| < 6^m$.
    \end{proposition}
    \begin{proof}
        Let $\vec{w}$ be any vector of the set $\{0, \pm1\}^m \setminus \{0^m\}$. 
        We say that the vector $\vec{w}$ splits iff $\exists$ $\vec{u} \neq \vec{v} \in \{0, 1\}^m
        \setminus \{0^m\}$ such that $\vec{w} = \vec{u} - \vec{v}$.
        Let $\mathsf{S}_{m_0}$ be the set of all the $\vec{w} \in \{0, \pm1\}^m \setminus \{0^m\}$ such that they contain $m_0$ zeroes. Then, we can say that 
        \begin{align*}
            \left|\mathsf{S}_{m_0}\right| = \binom{m}{m_0}2^{m - m_0}. 
        \end{align*}
        Since the split vectors belong to $\{0, 1\}^m$, any non-zero bit in $\vec{w}$ fixes the corresponding bit in the split vectors. Hence, 
        for any $\vec{w}$ with $m_0$ zeroes, the number of possible splits is $2^{m_0}$. 
        Then, we have
        \begin{equation*}
            \begin{aligned}
                \left|\mathsf{T_m}\right| &< \sum_{m_0 = 0}^{m-1} 
                \left|\mathsf{S}_{m_0}\right| \binom{2^{m_0}}{2}, \\
                &< \sum_{m_0 = 0}^{m-1} \binom{m}{m_0}2^{m - m_0} 2^{2m_0},\\
                &= 2^m \sum_{m_0 = 0}^{m-1} \binom{m}{m_0}2^{m_0}, \\
                &< 6^m.
            \end{aligned}
        \end{equation*}
    \end{proof}

    Now, we have all the desired tools to upper bound the terms $S_1$ and $S_2$. 
    \paragraph{Upper bound of $S_1$.} Let us recall the expression of the sum $S_1$
    
    \begin{equation*}\label{eq:S_1}
        \begin{aligned}
            S_1 &= \frac{1}{r} \cdot \sum_{\substack{i < j, l < k \\ (i, j) \neq (l, k)
            \\ \vec{u}_i - \vec{u}_j = \pm(\vec{u}_l - \vec{u}_k)
            }} \Pr_{\vec{a}^{(pq)}}
            \left(\mathcal{D}_{ij} = 1 \land \mathcal{D}_{lk} = 1\right),
        \end{aligned}
    \end{equation*}
    where $\mathcal{D}_{ij}$ is a binary random variable equal to $1$ if and only if events $\mathsf{E}_{\mathsf{LS}, p}^{(i)}$,
    $\mathsf{E}_{\mathsf{RS}, p}^{(i)}$, $\mathsf{E}_{\mathsf{LS}, p}^{(j)}$, $\mathsf{E}_{\mathsf{RS}, p}^{(j)}$, 
    $\mathsf{E}_{\mathsf{TS}, q}^{(i)}$, and $\mathsf{E}_{\mathsf{TS}, q}^{(j)}$ occur simultaneously. 

    \begin{claim}\label{claim: rank_argument}
        Let $\vec{u}$, $\vec{v}$, $\vec{w}$, $\vec{x}$ $\in$ $\{0, 1\}^{n} \setminus \{0^n\}$
        such that all of them are distinct vectors. Consider the 
        following matrix over the field $\mathbb{Z}_s$, where $s$ is a prime
        \begin{align*}
            \mat{M} = \left[\vec{u} \quad \vec{v} \quad \vec{w} \quad \vec{x}\right].
        \end{align*}
        Here, all the vectors are considered as column vectors. Then
        $3 \leq \mathsf{rank}(\mat{M}) \leq 4$. 
    \end{claim}
    \begin{proof}
        refer to \cref{sec: rank_argument}.
    \end{proof}
    
    Let the $\mathsf{rank}(\mat{U}_{ijlk}^{(L)})$, $\mathsf{rank}(\mat{U}_{ijlk}^{(R)})$,
    and $\mathsf{rank}(\mat{U}_{ijlk})$ are denoted by $t_1$, $t_2$, and $t_3$ respectively. 
    Using Claim~\ref{claim: rank_argument}, we know that $3 \leq t_3 \leq 4$.
    Since, the vectors $\vec{u}_i, \vec{u}_j, \vec{u}_l$, and $\vec{u}_k$ satisfy a relation in the sum $S_1$, we have 
    $t_3 = 3$. And, since $t_3 \geq \max\{t_1, t_2\}$, we have 
    $1 \leq t_1, t_2 \leq 3$.
    Now, we split the sum $S_1$ based on the rank of the matrices $\mat{U}_{ijlk}^{(L)}$ and $\mat{U}_{ijlk}^{(R)}$ and using~\cref{lemma: prob-counting-bad-vector-set}

    \begin{equation*}
        \begin{aligned}
            S_1 \leq &\sum_{\substack{i < j, l < k \\ (i, j) \neq (l, k)
            \\ \vec{u}_i - \vec{u}_j = \pm(\vec{u}_l - \vec{u}_k) \\
            t_1 + t_2 \leq 5}} \frac{1}{p^{t_1} \cdot p^{t_2} \cdot q^3 \cdot r}
            + 
            &\sum_{\substack{i < j, l < k \\ (i, j) \neq (l, k)
            \\ \vec{u}_i - \vec{u}_j = \pm(\vec{u}_l - \vec{u}_k) \\
            t_1 = t_2 = 3}} \frac{1}{p^{3} \cdot p^3 \cdot q^3 \cdot r}.
        \end{aligned}
    \end{equation*}

    Then, using \cref{prop: B_t_1_t_2_t_3_size} and ~\cref{prop:counting_collinear_tuples}, we have
    \begin{equation}\label{eq:S_1_final}
        \begin{aligned}
            S_1 &< \frac{2^{(t_1 + t_2)n/2}}{p^{t_1} \cdot p^{t_2} \cdot q^3 \cdot r} + 
            \frac{6^n}{p^3 \cdot p^3 \cdot q^3 \cdot r} \\
            &= \mathcal{O}\left(1\right) + 
            \mathcal{O}\left(\frac{1}{2^{0.165n}}\right)\\
            &= \mathcal{O}\left(1\right).
        \end{aligned}
    \end{equation}
    
    \paragraph{Upper bound of $S_2$.} 
    Using similar techniques and arguments, we have
    \begin{equation}\label{eq:S_2_final}
        \begin{aligned}
            S_2 &\leq \frac{2}{r^2} \cdot \sum_{\substack{i < j, l < k \\ (i, j) \neq (l, k)
            \\ \vec{u}_i - \vec{u}_j \neq \pm(\vec{u}_l - \vec{u}_k)
            }} \Pr_{\vec{a}^{(pq)}} \left(\mathcal{D}_{ij} = 1 \land \mathcal{D}_{lk} = 1\right)\\
            &\leq \frac{2 \cdot 2^{(t_1 + t_2)n/2}}{p^{t_1} \cdot p^{t_2} \cdot q^{t_3} \cdot r^2} \\
            &= \mathcal{O}\left(1\right).
        \end{aligned}
    \end{equation}
    Now, combining equations $(\ref{eq:second_moment_2})$, $(\ref{eq:S_1_final})$ and $(\ref{eq:S_2_final})$, we get that
    \begin{equation*}
        \begin{aligned}
            \underset{\vec{a}}{\mathbb{E}}(\mathcal{C}^2) &= S_1 + S_2 + \underset{\vec{a}}{\mathbb{E}}\left(\mathcal{C}\right) \\
            &= \mathcal{O}\left(1\right) + \mathcal{O}(1) + 
            \mathcal{O}\left(1\right)\\
            &= \mathcal{O}\left(1\right).
        \end{aligned}
    \end{equation*}
    This complete the proof of \cref{lemma: second_moment_modular}.

    \medskip
    
    \subsubsection{Concentration}
    Now, we use the following second-moment concentration bound to prove that the random variable
    $\mathcal{C}$ is concentrated around its mean value.
    \begin{lemma}
        (Paley-Zygmund Inequality, see e.g., \cite{roch2024modern} Section 2.3). For any non-negative
        random variable $Z$ and any $\theta \in [0, 1]$,
        \begin{align*}
            \Pr\left(Z > \theta \mathbb{E}(Z)\right) \geq (1 - \theta)^2 \cdot 
            \frac{\mathbb{E}(Z)^2}{\mathbb{E}(Z^2)}.
        \end{align*}
    \end{lemma}
     
    Applying the above lemma on the random variable $\mathcal{C}$ with $\theta = 0$, we get
    \begin{equation*}
        \begin{aligned}
            \Pr_{\vec{a}}\left(\text{\cref{alg: averange-case-modular} succeeds}\right) =
            \Pr_{\vec{a}}(\mathcal{C} > 0) &\geq 
            \frac{\underset{\vec{a}}{\mathbb{E}}\left(\mathcal{C}\right)^2}
            {\underset{\vec{a}}{\mathbb{E}}\left(\mathcal{C}^2\right)}\\
            &= \Omega\left(1\right).
        \end{aligned}
    \end{equation*}
\end{proof}
This completes the proof of ~\cref{lemma: Filtering_algo_succ_modular}. 

\subsection{Complexity Analysis}
\begin{lemma}
    Let $M, n, p,$ and $q$ be the same as defined in the above problem \textbf{ACM-RESS}.
    Then \cref{alg: averange-case-modular} runs in expected time 
    $\Oh^{\star}(M^{1/4})$ over the random choice of the input vector
    $\vec{a}$. Let $\mathsf{T}$ be the running time of the algorithm. Then, we also have
    the following tail bound
    \begin{align*}
        \Pr_{\vec{a}} \left(\mathsf{T} \geq \poly(n) M^{1/4}\right) =  
        \mathcal{O}\left(\frac{1}{\poly(n)}\right).
    \end{align*}
\end{lemma}
\begin{proof}
    Constructing the lists $B_{n/2, s_0}^{(L)}$ and $B_{n/2, s_1}^{(R)}$
    can be achieved by constructing two mod$-p$ DP tables with $n/2$ elements 
    and then by backtracking the DP table ~\cite{JW24,JWZ25}, this step takes time $\mathcal{O}(np)$,
    i.e., $\mathcal{O}(nM^{1/4})$. 
    Constructing the final list requires sorting 
    the two lists generated and then performing a binary search. This step takes
    time $\Tilde{\mathcal{O}}(L)$, where 
    \begin{align*}
        L := \max\left\{\left|B_{n/2, s_0}^{(L)}\right|, \left|B_{n/2, s_1}^{(R)}\right|,
            \left|B_{s_0, s_1, s_2}\right|\right\}. 
    \end{align*}
    Hence, we need to compute the expected size of the lists generated in the algorithm.
    
    \paragraph{Expected Size of the lists.}
    Let $\vec{v} \in \{0, 1\}^{n/2} \setminus \{0^{n/2}\}$, then the probability of this vectors surviving
    the first level of the algorithm is $1/p$, hence we have
    \begin{equation*}
        \begin{aligned}
            \underset{\vec{a}^{(L)}}{\mathbb{E}}\left(
            \left|B_{n/2, s_0}^{(L)}\right|\right) &= \frac{2^{n/2} - 1}{p},\\
            &= \mathcal{O}\left(M^{1/4}\right).
        \end{aligned}
    \end{equation*}
    
    Since $\vec{a}^{(L)}$ and $\vec{a}^{(R)}$ are sets of random variables with the same distribution, 
    we have
    \begin{align*}
        \underset{\vec{a}^{(R)}}{\mathbb{E}}\left(
        \left|B_{n/2, s_1}^{(R)}\right|\right) = 
        \underset{\vec{a}^{(L)}}{\mathbb{E}}\left(
        \left|B_{n/2, s_0}^{(L)}\right|\right)
        &= \mathcal{O}\left(M^{1/4}\right).
    \end{align*}

    Similarly, for any vector $\vec{v} \in \mathsf{A}$, the probability of these vectors surviving
    both levels of the algorithm is $1/p^2q$, hence we have
    
    \begin{align*}
        \underset{\vec{a}}{\mathbb{E}}\left(
        \left|B_{s_0, s_1, s_2}\right|\right) = 
        \frac{\left(2^{n/2} - 1\right)^2}{p^2q} = 
        \mathcal{O}\left(M^{1/4}\right).
    \end{align*}

    \paragraph{Tail Upper Bound.}
    Hence, the expected size of all three lists is 
    $\mathcal{O}\left(M^{1/4}\right)$. 
    Using Markov's inequality, we can say the following
    \begin{equation*}
        \begin{aligned}
            \Pr_{\vec{a}^{(L)}}
            \left(\left|B_{n/2, s_0}^{(L)}\right| \geq 
            \poly(n) M^{1/4} \right)
            &\leq 
            \frac{\underset{\vec{a}^{(L)}}{\mathbb{E}}
            \left(\left|B_{n/2, s_0}^{(L)}\right|\right)}
            {\poly(n) M^{1/4}}\\
            &= \mathcal{O}\left(\frac{1}{\poly(n)}\right).
        \end{aligned}
    \end{equation*}
    Similarly, we have
    \begin{equation*}
        \begin{aligned}
            \Pr_{\vec{a}^{(R)}}\left(\left|B_{n/2, s_1}^{(R)}\right| \geq
            \poly(n) M^{1/4}\right)
            &= \mathcal{O}\left(\frac{1}{\poly(n)}\right), \\
            \Pr_{\vec{a}}\left(\left|B_{s_0, s_1, s_2}\right| \geq
            \poly(n) M^{1/4}\right)
            &= \mathcal{O}\left(\frac{1}{\poly(n)}\right).
        \end{aligned}
    \end{equation*}
    This completes the proof.
\end{proof}

\begin{remark}
    Combining the above tail bound for the running time and the success probability of the algorithm, 
    we can claim that the filtering algorithm for
    $\Omega(1)$ fraction of instances will have running time $\Oh^{\star}(2^{n/4})$ and the final list will have a collision. 
\end{remark}

\section{Integer average-case: Proof of \cref{thm:average-case-integer}}
\label{sec: average-case-integer}
In this section, we present the proof of \cref{thm:average-case-integer}. Firstly, 
recall the definition of the \textbf{ACI-RESS} problem.
\paragraph{Average-Case Integer Restricted Equal Subset Sum (ACI-RESS).} \label{prob: average_case_integer}
Given two integers $M$ and $n$ divisible by $4$ as input, such that $M \leq 2^n$.

Choose $n$ integers $a_1, a_2, \ldots, a_n$ uniformly randomly from $[M]$
and the task is to find a non-trivial integer collision, i.e., two \emph{distinct} sets 
$\vec{u}, \vec{v} \in \{0, 1\}^n$ such that $\vec{a} . \vec{u} = \vec{a}. \vec{v}$,
where $.$ denotes the standard dot product.

Similar to the Modular case, w.l.o.g. we can assume that 
$M = \Theta(2^n)$.
The algorithm for solving $\textbf{ACI-RESS}$ will remain the same as in the modular case. The inputs will have a different structure as
mentioned in \cref{alg: averange-case-integer}. 
We give \cref{alg: averange-case-integer} for \cref{thm:average-case-integer}, 
for completeness.

In the modular case, we observed that the sum of random variables modulo $p$, $q$, and $r$ is uniformly distributed over their corresponding supports, and to obtain a
collision modulo $M=pqr$, one must additionally ensure collision in the
residual $r$-part, after which the Chinese Remainder Theorem yields a
collision modulo $M$. In the integer setting, agreement modulo $pq$ remains the same as in the modular setting. Still, after that, one must additionally understand carry
propagation from the sum of mod-$pq$ variables and the collision over the remaining high-order bits. 
This is the bottleneck in analyzing the integer case.

\begin{algorithm}[h!]
\caption{Algorithm for~\cref{thm:average-case-integer}.}
\label{alg: averange-case-integer}
\KwIn{$(a_1, a_2, \ldots, a_n) \in [M]^n$; \textbf{Parameters:} $p, q \in 
\Theta(M^{1/4})$ co-prime,
$s_0, s_1 \in \mathbb{Z}_p$ and $s_2 \in \mathbb{Z}_q$\\}
\KwOut{A non-trivial integer collision or $\bot$}
Build two mod-$p$ DP tables storing the sizes of the following sets $\forall$ $s \in \mathbb{Z}_p$ and $\forall$ $1 \leq i \leq n/2$
\begin{align*}
    B_{i, s}^{(L)} &:= \left\{ \vec{u} \in \{0, 1\}^i\; \middle| \; 
    \sum_{j = 1}^{i}a_ju_j \equiv s \bmod p \right\}, \\
    B_{i, s}^{(R)} &:= \left\{ \vec{v} \in \{0, 1\}^i \; \middle| \; 
    \sum_{j = 1}^{i}a_{n/2 + j}v_j \equiv s \bmod p \right\}.
\end{align*}

Compute the bins $B_{n/2, s_0}^{(L)}$ and $B_{n/2, s_1}^{(R)}$ by backtracing the DP table; look for a collision in these sets; else continue.

Construct the following set
\begin{equation*}
    \begin{aligned}
        B_{s_0, s_1, s_2} &:= 
        \left\{ \vec{u} || \vec{v} \; \middle| \; 
        \begin{aligned}
            \vec{u} \in B_{n/2, s_0}^{(L)}, \; 
            \vec{v} \in B_{n/2, s_1}^{(R)},\; \\
            \vec{a}^{(L)} \cdot \vec{u} + \vec{a}^{(R)} \cdot \vec{v} \equiv s_2 \bmod q 
        \end{aligned}\right\}
    \end{aligned}
\end{equation*}

Look for a non-trivial integer collision in this list $B_{s_0, s_1, s_2}$.
If one is found, it outputs this collision; otherwise, it outputs $\bot$.
\end{algorithm}

\paragraph{Unstructured M.} 
Firstly, we divide $M$ by the product of the integers $p$ and $q$ (the parameters of the algorithm
\footnote{Parameters $p$ and $q$ are chosen such that $\floor{M/pq}$ is co-prime to $p$ and $q$ both.}). 
Let $r$ and $s$ be the quotient and remainder of the division as mentioned below
\begin{align*}
    M = r\cdot pq + s,
\end{align*}
where $0 < s < pq$. Recall that the input vector $\vec{a}$ is uniformly distributed over the set $[M]^n = \{1, 2, \ldots, M\}^n$.
Let us split all the random variables, $a_i$, into the lower $pq$ bits and the remaining as follows
\begin{align*}
    a_i = a_i^{(pq)} + a_i^{(r)}\cdot pq,
\end{align*}
where $a_i^{(pq)} \in [0, pq)$ and $a_i^{(r)} = [0, r]$ $\forall$ $1 \leq i \leq n$.

Now, we split the probability mass of the distribution of the 
random variable $a_i$ into two cases. First, the inner rectangle, i.e. 
$1 \leq a_i \leq pqr$, and second the truncated slice, i.e. $pqr < a_i \leq M$.
Let $\mathsf{E_i^{(U)}}$ and $\mathsf{E_i^{(T)}}$ be two disjoint exhaustive events that denote these
two cases, respectively. Then, we have
\begin{equation*}
    \begin{aligned}
        \Pr_{a_i}\left(\mathsf{E_i^{(T)}}\right) &= \frac{s}{M},\\
        \Pr_{a_i}\left(\mathsf{E_i^{(U)}}\right) &= 1 - \frac{s}{M}.
    \end{aligned}
\end{equation*}
Since $s < pq$, we have $s/M = \mathcal{O}\left(1/2^{n/2}\right)$.
Let $\mathsf{E^{(T)}} := \lor_{i = 1}^{n} \mathsf{E_i^{(T)}}$, then using union bound, we have
\begin{equation*}
    \begin{aligned}
        \Pr_{\vec{a}}\left(\mathsf{E^{(T)}}\right) \leq \frac{ns}{M} = \mathcal{O}\left(\frac{n}{2^{n/2}}\right).
    \end{aligned}
\end{equation*}
Hence, we have
\begin{align*}
    \Pr_{\vec{a}}\left(\mathsf{E^{(U)}}\right) \geq 1 - \mathcal{O}\left(\frac{n}{2^n}\right).
\end{align*}

Since we are only interested in computing the lower bound of the expectation and
the event $\mathsf{E^{(U)}}$ has almost the entire probability mass, 
it makes sense to consider events conditioned on the event $\mathsf{E^{(U)}}$.
In particular, for the remainder of the proof, we assume that $M = pqr$, and the random variables 
$\{a_{i}^{(pq)}\}_{i = 1}^{n}$ and $\{a_{i}^{(r)}\}_{i = 1}^{n}$ are independent and distributed 
uniformly on the integer sets $[1, pq]$ and $[0, r-1]$ respectively. 

\subsection{Success Probability}
\begin{lemma}\label{lemma: Filtering_algo_succ_integer}
   Let $M$ and $n$ be the same as defined in the above problem \textbf{ACI-RESS}. 
    Then the following holds for the \cref{alg: averange-case-integer}
    \begin{align*}
        \Pr_{\vec{a}}\left[\text{\cref{alg: averange-case-integer} succeeds}\right] =
        \Omega\left(\frac{1}{n^{2}}\right).
    \end{align*}
\end{lemma}

In the remainder of the section, we give a proof of \cref{lemma: Filtering_algo_succ_integer}.
Recall the following set from the modular case
\begin{align*}
    \mathsf{A} = \{\vec{u} \in \{0, 1\}^n \; | \; 
    \vec{u}^{(L)} \neq 0^{n/2}, \vec{u}^{(R)} \neq 0^{n/2}\}.
\end{align*}
The elements of the set $\mathsf{A}$ are ordered in a lexicographical sense, and 
$\mathsf{A} = \{\vec{u}_1, \vec{u}_2, \ldots \vec{u}_{N}\}$, where $\vec{u}_i < \vec{u}_j$ 
for any $1 \leq i < j \leq N$ 
and $N = (2^{n/2} - 1)^2 = 2^n - 2^{n/2 + 1} + 1$.

The following random variable over a random choice of $\vec{a}$ denotes the number of 
integer collisions that survive the filtering algorithm. To be more precise, we have
\begin{equation*}
    \begin{aligned}
        \mathcal{C}^{(\mathsf{int})} := \left|\left\{(\vec{u}_i, \vec{u}_j) \; \middle | \;  
        \begin{aligned}
            i < j, \; \vec{u}_i, \vec{u}_j \in  B_{s_0, s_1, s_2}, \\
            \floor{\frac{\vec{a} \cdot \vec{u}_i}{pq}} = \floor{\frac{\vec{a} \cdot \vec{u}_j}{pq}}
        \end{aligned}\right\}\right|.
    \end{aligned}
\end{equation*}
Our filtering algorithm succeeds for an instance of the problem \textbf{ACI-RESS} 
if the value of this non-negative random variable is at least $1$.
Hence, similar to the modular case, analyzing
the random variable $\mathcal{C}^{(\mathsf{int})}$ will give us the success probability of the \cref{alg: averange-case-integer}.

For any input instance $\vec{a}$, we can compute the value of $\mathcal{C}^{\mathsf{(int)}}$ 
by iterating through all
the possible pairs of binary vectors from the set $\mathsf{A}$, 
and considering whether the vectors pass through the 
filters of the algorithm and collide over integers. Hence, we define a binary random variable 
$\mathcal{C}^{(\mathsf{int})}_{ij}$ $\forall$ $1 \leq i < j \leq N$
such that its value is $1$ if and only if $\vec{u}_i$ and $\vec{u}_j$ pass through all the filters of 
the algorithm, and they collide over integers. To analyze this random variable 
$\mathcal{C}^{(\mathsf{int})}_{ij}$, we first define the following random variables that will
denote different sums of the input random variables. 
Let us fix a binary vector $\vec{u}_i \in \mathsf{A}$ with Hamming weight $w_i$, and define

\begin{equation*}
    \begin{aligned}
        S_{\mathsf{TS}, i} &:= \vec{a} \cdot \vec{u}_i 
        &(\text{Total sum over integers})\\
        S_{\mathsf{TS}, pq, i} &:= \vec{a}^{(pq)} \cdot \vec{u}_i 
        &(\text{Total sum lower $pq$ bits over integers})\\
        S_{\mathsf{TS}, i}^{(pq)} &:= \vec{a}^{(pq)} \cdot \vec{u}_i \mod pq 
        &(\text{Mod $pq$ bits of the Total sum $S_{\mathsf{TS}, pq, i}$})\\
        S_{\mathsf{C}, pq, i} &:= \floor{\frac{\vec{a}^{(pq)}. \vec{u}_i}{pq}} 
        = \frac{\vec{a}^{(pq)}. \vec{u}_i - (\vec{a}^{(pq)}.\vec{u}_i \mod pq)}{pq} 
        &(\text{Carry bits from $S_{\mathsf{TS}, pq, i}$})\\
        S_{\mathsf{TS}, r, i} &:= \vec{a}^{(r)} \cdot \vec{u}_i 
        &(\text{Remaining bits of the Total sum $S_{\mathsf{TS}, i}$})\\
        S_{\mathsf{LS}, i}^{(p)} &:= \vec{a}^{(p)}{}^{(L)} \cdot \vec{u}_i^{(L)} \mod p 
        &(\text{Left sum Mod $p$})\\
        S_{\mathsf{RS}, i}^{(p)} &:= \vec{a}^{(p)}{}^{(R)} \cdot \vec{u}_i^{(R)} \mod p
        &(\text{Right sum Mod $p$})
    \end{aligned}
\end{equation*}

A few of the variables are related as follows
\begin{equation*}
    \begin{aligned}
        S_{\mathsf{TS}, i} = 
        \vec{a} \cdot \vec{u}_i &= \left(\vec{a}^{(pq)} + \vec{a}^{(r)} \cdot pq\right)
        \cdot \vec{u}_i \\
        &= \left(\underbrace{\vec{a}^{(pq)} \cdot \vec{u}_i \mod pq}_{\mod pq 
        \text{ bits}} + 
        \underbrace{\frac{\vec{a}^{(pq)}\cdot \vec{u}_i - (\vec{a}^{(pq)}.\vec{u}_i 
        \mod pq)}{pq}}_{\text{carry bits}} \cdot pq\right) + 
        \left(\underbrace{\vec{a}^{(r)} \cdot \vec{u}_i}_{\text{remaining bits}}\right) \cdot pq \\
        &= \left(S_{\mathsf{TS}, i}^{(pq)} + 
        S_{\mathsf{C}, pq, i} \cdot pq \right) + S_{\mathsf{TS}, r, i} \cdot pq\\
        &= S_{\mathsf{TS}, pq, i} + S_{\mathsf{TS}, r, i} \cdot pq.
    \end{aligned}
\end{equation*}

The above sum in the bracket can be split as follows
\begin{equation*}
    \begin{aligned}
        S_{\mathsf{TS}, i} &= S_{\mathsf{TS}, pq, i} + S_{\mathsf{TS}, r, i} \cdot pq, \\
        &= \left(S_{\mathsf{TS}, i}^{(pq)} + 
        S_{\mathsf{C}, pq, i} \cdot pq \right) + S_{\mathsf{TS}, r, i} \cdot pq.
    \end{aligned}
\end{equation*}

Then, we redefine $\mathcal{C}^{(\mathsf{int})}_{ij} = 1$ if and only if 
all the following events occur simultaneously and is $0$ otherwise. 
\begin{enumerate}
    \item \textit{$\vec{u}_i, \vec{u}_j$ survives the Filtering algorithm}, i.e. 
    \begin{align*}
        \mathsf{E}_{\mathsf{LS}, \vec{a}, i}^{(p)} &: S_{\mathsf{LS}, i}^{(p)} \equiv s_0 \mod p,\\
        \mathsf{E}_{\mathsf{LS}, \vec{a}, j}^{(p)} &: S_{\mathsf{LS}, j}^{(p)} \equiv s_0 \mod p,\\
        \mathsf{E}_{\mathsf{TS}, \vec{a}, i}^{(pq)} &: S_{\mathsf{TS}, i}^{(pq)} \equiv \beta \mod pq,\\
        \mathsf{E}_{\mathsf{TS}, \vec{a}, j}^{(pq)} &: S_{\mathsf{TS}, j}^{(pq)} \equiv \beta \mod pq.
    \end{align*}
    \item \textit{$\vec{u}_i, \vec{u}_j$ collide over the Carries}, i.e. 
    \begin{align*}
        \mathsf{E}_{\mathsf{C}, pq, \vec{a}, i, j} &: S_{\mathsf{C}, pq, i} = 
        S_{\mathsf{C}, pq, j}.
    \end{align*}
    \item \textit{$\vec{u}_i, \vec{u}_j$ collide over the remaining bits}, i.e.
    \begin{align*}
        \mathsf{E}_{\mathsf{TS}, r, \vec{a}, i, j} &: S_{\mathsf{TS}, r, i} = S_{\mathsf{TS}, r, j}.
    \end{align*}
\end{enumerate}

where $\beta \equiv s_0 + s_1 \mod p$ and $\beta \equiv s_2 \mod q$. The first four events make sure that the vectors
$\vec{u}_i$ and $\vec{u}_j$ survive the filtering algorithm, and the remaining events constraint $\vec{u}_i$
and $\vec{u}_j$ to collide over the carry bits and the remaining bits respectively. 
Here, the reader should observe that the above constraints are sufficient, but not necessary 
for the collision over integers.

\medskip

Now, similar to the modular case, we will compute bounds on the first and second moments 
of the random variable $\mathcal{C}^{(\mathsf{int})}$ to prove~\cref{lemma: Filtering_algo_succ_integer}.

\subsubsection{First Moment}
\begin{lemma}\label{lemma: first_moment_integer}
    Let $M$ and $n$ be the same as defined in the above problem \textbf{ACI-RESS}. 
    Then the following holds for the \cref{alg: averange-case-integer}
    \begin{align*}
        \underset{\vec{a}}{\mathbb{E}}\left(\mathcal{C}^{\mathsf{(int)}}\right) = 
        \Omega\left(\frac{1}{n}\right).
    \end{align*}
\end{lemma} 
\begin{proof}
    To compute a lower bound of the first moment of the random variable $\mathcal{C}^{\mathsf{(int)}}$, we first look at its relation with the 
    binary random variables $\mathcal{C}_{ij}^{\mathsf{(int)}}$ as follows
    \begin{equation*}
        \begin{aligned}
            \mathcal{C}^{\mathsf{(int)}} &\geq \sum_{1 \leq i < j \leq N} \mathcal{C}_{ij}^{\mathsf{(int)}}.
        \end{aligned}
    \end{equation*}
    The inequality comes because we work with the sufficient condition for collision, then the necessary condition in the definition of the variable $\mathcal{C}_{ij}^{\mathsf{(int)}}$.
    Using the linearity of expectation, we have the following
    \begin{equation*}
        \begin{aligned}
            \underset{\vec{a}}{\mathbb{E}}\left(\mathcal{C}^{\mathsf{(int)}}\right) &\geq
            \sum_{i < j} \underset{\vec{a}}{\mathbb{E}}\left(\mathcal{C}_{ij}^{\mathsf{(int)}}\right).
        \end{aligned}
    \end{equation*}

    \paragraph{Expectation for overlapping vectors.}
    Let us make the following observation. Fix $\vec{u}_i, \vec{u}_j \in \mathsf{A}$.
    Then, we split the set $[n]$ into three parts as follows. The first part contains the indices where
    the vector $\vec{u}_i$ is $1$ and the vector $\vec{u}_j$ is $0$. Similarly, the second part
    corresponds to $0$ and $1$, and the last part corresponds to $1$ and $1$. This will give us
    three disjoint subsets of the set $[n]$. These are the only indices that are active for these
    two binary vectors. 

    Observe that indices at which both vectors have a $1$ in the same position necessarily contribute to a collision, regardless of the specific input instance. Consequently, it suffices to analyze the remaining components. Since these components correspond to disjoint index sets, the problem effectively reduces to analyzing vectors with non-overlapping supports.

    Now, we formalize this argument.
    Define three auxiliary binary vectors $\vec{u}_{ij}^{(10)}, \vec{u}_{ij}^{(01)}$, and $\vec{u}_{ij}^{(11)}$ as follows
    \begin{equation}\label{eq: aux_vectors}
        \begin{aligned}
            u_{ijl}^{(10)} &:= 
            \begin{cases}
                1 & u_{il} = 1, u_{jl} = 0\\
                0 & \text{else}
            \end{cases} \qquad \forall \;  1 \leq l \leq n\\
            u_{ijl}^{(01)} &:= 
            \begin{cases}
                1 & u_{il} = 0, u_{jl} = 1\\
                0 & \text{else}
            \end{cases} \qquad \forall \;  1 \leq l \leq n \\
            u_{ijl}^{(11)} &:= 
            \begin{cases}
                1 & u_{il} = 1, u_{jl} = 1\\
                0 & \text{else}
            \end{cases} \qquad \forall \;  1 \leq l \leq n
        \end{aligned}
    \end{equation}
    and let $w_{ij}^{(10)}$, $w_{ij}^{(01)}$, and 
    $w_{ij}^{(11)}$ denote the Hamming weights of the vectors $\vec{u}_{ij}^{(10)}$, 
    $\vec{u}_{ij}^{(01)}$, and $\vec{u}_{ij}^{(11)}$ respectively. 
    Let $w_{ij}^{(10)}, w_{ij}^{(01)} \in (cn, n]$
    for a constant $1/8 \leq c < 1/4$. That implies that the non-overlapping parts have big enough weights. 
    Now, we define another binary random variable $\mathcal{D}_{ij}^{(int)}$ over 
    the random choice of $\vec{a}$. The variable $\mathcal{D}_{ij}^{(int)} = 1$ if  
    the vectors $\vec{u}_i$ and $\vec{u}_j$ survive the algorithm, and the auxiliary vectors
    $\vec{u}_{ij}^{(10)}$ and $\vec{u}_{ij}^{(01)}$ collide over the carry bits and the remaining bits.
    
    Formally, we define $\mathcal{D}_{ij}^{(int)} = 1$ if
    the auxiliary vectors $\vec{u}_{ij}^{(10)}, \vec{u}_{ij}^{(01)}$, and $\vec{u}_{ij}^{(11)}$ satisfy the following equations 
    \begin{equation*}
        \begin{aligned}
            \mathsf{F}_{\mathsf{LS}, \vec{a}, 10}^{(p)} &: \vec{u}_{ij}^{(10)}{}^{(L)}\cdot 
            \vec{a}^{(p)}{}^{(L)} 
            \equiv s_0 - \vec{u}_{ij}^{(11)}{}^{(L)}\cdot \vec{a}^{(p)}{}^{(L)} \mod p,\\
            \mathsf{F}_{\mathsf{LS}, \vec{a}, 01}^{(p)} &: \vec{u}_{ij}^{(01)}{}^{(L)}\cdot 
            \vec{a}^{(p)}{}^{(L)} 
            \equiv s_0 - \vec{u}_{ij}^{(11)}{}^{(L)}\cdot \vec{a}^{(p)}{}^{(L)} \mod p,\\
            \mathsf{F}_{\mathsf{TS}, \vec{a}, 10}^{(pq)} &: \vec{u}_{ij}^{(10)} \cdot \vec{a}^{(pq)}
            \equiv \beta - \vec{u}_{ij}^{(11)} \cdot \vec{a}^{(pq)}  \mod pq,\\
            \mathsf{F}_{\mathsf{TS}, \vec{a}, 01}^{(pq)} &: \vec{u}_{ij}^{(01)} \cdot \vec{a}^{(pq)}
            \equiv \beta - \vec{u}_{ij}^{(11)} \cdot \vec{a}^{(pq)}  \mod pq,
        \end{aligned}
    \end{equation*}
    \begin{equation*}
        \begin{aligned}
            \mathsf{F}_{\mathsf{OB}, pq, \vec{a}, i, j} &: \floor{\frac{\vec{u}_{ij}^{(10)} \cdot \vec{a}^{(pq)}}{pq}} = 
            \floor{\frac{\vec{u}_{ij}^{(01)} \cdot \vec{a}^{(pq)}}{pq}}\\
            \mathsf{F}_{\mathsf{TS}, r, \vec{a}, i, j} &: 
            \vec{u}_{ij}^{(10)} \cdot \vec{a}^{(r)} = \vec{u}_{ij}^{(01)} \cdot \vec{a}^{(r)}.
        \end{aligned}
    \end{equation*}
    
    \begin{claim}\label{claim: D_ij_int_C_ij_int}
        For any $1 \leq i< j  \leq N$, the following holds 
        \begin{align*}
            \mathcal{C}_{ij}^{(int)} \geq \mathcal{D}_{ij}^{(int)}.
        \end{align*}
    \end{claim}
    \begin{claimproof}
        Let $\mathcal{D}_{ij}^{(int)} = 1$.
        Then the events $\mathsf{F}_{\mathsf{LS}, \vec{a}, 10}^{(p)}, \mathsf{F}_{\mathsf{LS}, \vec{a}, 01}^{(p)}
        , \mathsf{F}_{\mathsf{TS}, \vec{a}, 10}^{(p)}$, and $\mathsf{F}_{\mathsf{TS}, \vec{a}, 01}^{(p)}$ trivially
        implies the survival of the vectors $\vec{u}_i$ and $\vec{u}_j$. And, the 
        event $\mathsf{F}_{\mathsf{TS}, r, \vec{a}, i, j}$ implies collision
        over the remaining $r$ bits. Since $\mathsf{F}_{\mathsf{OB}, pq, \vec{a}, i, j} = 1$,
        and $\vec{u}_{ij}^{(01)} \cdot \vec{a}^{(pq)} \equiv_{pq} \vec{u}_{ij}^{(10)} \cdot \vec{a}^{(pq)}$
        we have the following
        \[
            \vec{u}_{ij}^{(01)} \cdot \vec{a}^{(pq)} = \vec{u}_{ij}^{(10)} \cdot \vec{a}^{(pq)}.
        \]
        This gives us the following equality over the carries too
        \begin{align*}
            \floor{\frac{\vec{u}_i \cdot \vec{a}^{(pq)}}{pq}} &= 
            \floor{\frac{ \vec{u}_{ij}^{(10)} \cdot \vec{a}^{(pq)} + 
            \vec{u}_{ij}^{(11)}\cdot \vec{a}^{(pq)}}{pq}}\\
            &= \floor{\frac{ \vec{u}_{ij}^{(01)} \cdot \vec{a}^{(pq)} + 
            \vec{u}_{ij}^{(11)}\cdot \vec{a}^{(pq)}}{pq}} \\
            &= \floor{\frac{\vec{u}_j \cdot \vec{a}^{(pq)}}{pq}}.
        \end{align*}
        This implies that $\vec{u}_i$ and $\vec{u}_j$ collides over the carries too,
        and we have $\mathcal{C}_{ij}^{(int)} = 1$.
    \end{claimproof}

    \medskip
    Due to Claim~\ref{claim: D_ij_int_C_ij_int}, the task of analyzing the random variable $\mathcal{C}_{ij}^{(int)}$
    reduces to the case of vectors with non-overlapping supports. 
    This reduction simplifies the analysis and allows us to avoid dealing
    with dependencies arising from overlapping coordinates of the two vectors.
    
    \begin{lemma}\label{lemma: first_moment_integer_ij_no_overlap}
        Fix $\vec{u}_i, \vec{u}_j \in \mathsf{A}$ such that they do not overlap, 
        i.e. $\|\vec{u}_i\|_1 + \|\vec{u}_j\|_1 = \|\vec{u}_i \xor \vec{u}_j\|_1$. Let 
        $w_i, w_j, w_{l, i}, w_{r, i}$, $w_{l, j}$, and $w_{r, j}$
        denote the Hamming weights of the vectors $\vec{u}_i, \vec{u}_j, \vec{u}_i^{(L)}, \vec{u}_i^{(R)}, \vec{u}_j^{(L)}$, and 
        $\vec{u}_j^{(R)}$ respectively. Assume that $\alpha n < w_i, w_j < (1-\alpha)n$,
        $|w_i - w_j| = \mathcal{O}(\sqrt{n})$ and 
        $\alpha n/2 < w_{l, i}, w_{r, i}, w_{l, j}, w_{r, j} < (1- \alpha)n/2$
        for some constant $1/8 \leq \alpha < 1/4$. Then, we have the following
        \begin{align*}
            \underset{\vec{a}}{\mathbb{E}}\left(\mathcal{C}_{ij}^{\mathsf{(int)}}\right) = 
            \Omega\left(\frac{1}{np^4q^2r}\right).
        \end{align*}
    \end{lemma}
    \begin{proof}
        Since $\vec{a}^{(pq)}$ and $\vec{a}^{(r)}$ are independent of each other, 
        $\underset{\vec{a}}{\mathbb{E}}\left(\mathcal{C}_{ij}^{\mathsf{(int)}}\right)$ is the 
        product of the following two probabilities
        \begin{equation}\label{eq: E_C_ij_int_split}
            \begin{aligned}
                \Pr_{\vec{a}^{(pq)}} \left(\mathsf{E}_{\mathsf{LS}, \vec{a}, i}^{(p)} \wedge 
                \mathsf{E}_{\mathsf{LS}, \vec{a}, j}^{(p)} \wedge \mathsf{E}_{\mathsf{TS}, \vec{a}, i}^{(pq)} \wedge \mathsf{E}_{\mathsf{TS}, \vec{a}, j}^{(pq)} \wedge 
                \mathsf{E}_{\mathsf{C}, pq, \vec{a}, i, j}\right),
            \end{aligned}
        \end{equation}
        and 
        \begin{align*}
            \Pr_{\vec{a}^{(r)}} \left(\mathsf{E}_{\mathsf{TS}, r, \vec{a}, i, j}\right).
        \end{align*}
        Now, we analyze the $\vec{a}^{(pq)}$ and $\vec{a}^{(r)}$ events separately in the following 
        proposition.
        \begin{proposition} \label{prop: events_vec_a_probabilities}
            Fix $\vec{u}_i, \vec{u}_j \in \mathsf{A}$ such that they satisfy the same constraints as
            \cref{lemma: first_moment_integer_ij_no_overlap}. Then, we have the following
            \begin{enumerate}
                \item 
                    \begin{align*}
                        &\Pr_{\vec{a}^{(pq)}} \left(\mathsf{E}_{\mathsf{LS}, \vec{a}, i}^{(p)} \wedge 
                        \mathsf{E}_{\mathsf{LS}, \vec{a}, j}^{(p)} \wedge \mathsf{E}_{\mathsf{TS}, \vec{a}, i}^{(pq)} \wedge \mathsf{E}_{\mathsf{TS}, \vec{a}, j}^{(pq)} \wedge 
                        \mathsf{E}_{\mathsf{C}, pq, \vec{a}, i, j}\right) \\
                        &= \Omega\left(\frac{1}{p^4q^2\sqrt{n}}\right).
                    \end{align*}
                \item 
                    \begin{align*}
                        \Pr_{\vec{a}^{(r)}}\left(\mathsf{E}_{\mathsf{TS}, r, \vec{a}, i, j}\right)
                        = \Omega\left(\frac{1}{r\sqrt{n}}\right).
                    \end{align*}
            \end{enumerate}
        \end{proposition}
        \begin{proof}
            Refer to \cref{sec: events_vec_a_probabilities}.
        \end{proof}
        
        Using \cref{prop: events_vec_a_probabilities} in equation $(\ref{eq: E_C_ij_int_split})$, we get the desired
        result. This completes the proof.
    \end{proof} 
    Now that we have computed a lower bound on the expectation of the random variable $\mathcal{C}_{ij}^{(\mathsf{int})}$ for vectors with constrained Hamming weights, we are left with the problem of counting such constrained vectors.
    \medskip 
    \paragraph{Counting "Big Weight" Vectors.} 
    To apply the Local Central Limit Theorem effectively, the vectors under consideration must have sufficiently large Hamming weight so that the distribution of their associated sums converges to a normal distribution with a small approximation error. In our setting, this requires counting pairs of vectors whose non-overlapping components have Hamming weight $\Theta(n)$ on both the left and right halves of the vector. 
    To enumerate such pairs, we rely on the following proposition.
    
    \begin{proposition}\label{prop: big_weight_vec_count}
        For any large integer $m$ and constants $1/8 \leq \alpha < 1/4$ and $c > 1/\sqrt{2}$, 
        we have the following
        \begin{align*}
            \sum_{t_1, t_2}\binom{m}{t_1} \binom{m - t_1}{t_2} 2^{m - (t_1 + t_2)} = \Theta\left(4^m\right),
        \end{align*}
        where $\alpha m \leq t_1, t_2 < m$, $t_1 + t_2 \leq m$, and $|t_1 - t_2| \leq c\sqrt{m}$.
    \end{proposition}
    \begin{proof}
        Refer \cref{sec: big_weight_vec_count}.
    \end{proof}

    Using the above proposition, we conclude that there exists a set $\mathsf{T} \subseteq \mathsf{A} \times \mathsf{A}$ 
    such that for any element $(\vec{u}_l, \vec{u}_k) \in \mathsf{T}$, 
    their corresponding auxiliary binary vectors $\vec{u}_{lk}^{(10)}, \vec{u}_{lk}^{(01)}$, and $\vec{u}_{lk}^{(11)}$, as defined above in the equation $(\ref{eq: aux_vectors})$, 
    will satisfy the following constraints
    \begin{equation}\label{eq: aux_vectors_weight_bounds}
        \begin{aligned}
            \frac{n}{16} &\leq \left\lVert \vec{u}_{lk}^{(10)}{}^{(L)} \right\rVert_1, 
            \left\lVert \vec{u}_{lk}^{(10)}{}^{(R)} \right\rVert_1, \\
            \frac{n}{16} &\leq \left\lVert \vec{u}_{lk}^{(01)}{}^{(L)} \right\rVert_1, 
            \left\lVert \vec{u}_{lk}^{(01)}{}^{(R)} \right\rVert_1, \\
            &\left|\left\lVert \vec{u}_{lk}^{(10)} \right\rVert_1 - 
            \left\lVert \vec{u}_{lk}^{(01)} \right\rVert_1\right| \leq c\sqrt{n}, \text{ for some constant
            $c > 1/\sqrt{2}$}
        \end{aligned}
    \end{equation} 
    and $|\mathsf{T}| = \Theta(4^n)$.
    Firstly, using Claim~\ref{claim: D_ij_int_C_ij_int}, for every pair $(\vec{u}_l, \vec{u}_k) \in \mathsf{T}$, we have 
    the following
    \begin{align}\label{eq: E_C_lk_D_lk}
        \underset{\vec{a}}{\mathbb{E}}\left(\mathcal{C}_{lk}^{\mathsf{(int)}}\right) \geq 
        \underset{\vec{a}}{\mathbb{E}}\left(\mathcal{D}_{lk}^{\mathsf{(int)}}\right).
    \end{align}
    Since the vectors $\vec{u}_{lk}^{(10)}$ and $\vec{u}_{lk}^{(01)}$ do not overlap, and 
    using equations $(\ref{eq: aux_vectors_weight_bounds})$, we can say these vectors satisfy the weight constraints 
    of the \cref{lemma: first_moment_integer_ij_no_overlap}, therefore we get the following
    \begin{align}\label{eq: E_D_lk_lower_bound}
        \underset{\vec{a}}{\mathbb{E}}\left(\mathcal{D}_{lk}^{\mathsf{(int)}}\right) = 
        \Omega\left(\frac{1}{np^4q^2r}\right).
    \end{align}
    Using equations $(\ref{eq: E_C_lk_D_lk})$, $(\ref{eq: E_D_lk_lower_bound})$, and
    the fact that $|T| = \Theta(4^n)$, $p, q = \Theta\left(M^{1/4}\right)$, and
    $M = pqr = \Omega\left(2^n\right)$, we get the following
    \begin{equation*}
        \begin{aligned}
            \underset{\vec{a}}{\mathbb{E}}\left(\mathcal{C}^{\mathsf{(int)}}\right) &\geq
            \sum_{i < j} \underset{\vec{a}}{\mathbb{E}}\left(\mathcal{C}_{ij}^{\mathsf{(int)}}\right) \\
            &> \sum_{(\vec{u}_l, \vec{u}_k) \in \mathsf{T}} 
            \underset{\vec{a}}{\mathbb{E}}\left(\mathcal{C}_{lk}^{\mathsf{(int)}}\right)\\
            &\geq \sum_{(\vec{u}_l, \vec{u}_k) \in \mathsf{T}} 
            \underset{\vec{a}}{\mathbb{E}}\left(\mathcal{D}_{lk}^{\mathsf{(int)}}\right)\\
            &= \Omega\left(\frac{4^n}{np^4q^2r}\right) \\
            &= \Omega\left(\frac{1}{n}\right).
        \end{aligned}
    \end{equation*}
    \begin{remark}
        Recall that we assume $M = pqr$, and that the random vectors $a_i^{(pq)}$
        and $a_i^{(r)}$ are independent and uniformly distributed over the sets $[1, pq]$ and 
        $[0, r-1]$, respectively. This assumption holds conditioned 
        on the event that $1 \leq a_i \leq pqr$ $\forall$ $1 \leq i \leq n$, given $M = pqr + s$. 
        Since the probability of this event is at least $1 - n/2^{n/2}$,
        incorporating this factor into the above calculation of the first moment does not 
        change the asymptotic bound, and we obtain the same result.
        This argument is explained in detail in \cref{sec: discussion_M}.
    \end{remark}
    This completes the proof of \cref{lemma: first_moment_integer}.
\end{proof}

\subsubsection{Second Moment}\label{sec: second_moment_integer}
In the previous subsection, while analyzing the first moment of the random variable $\mathcal{C}^{\mathsf{(int)}}$,
we assumed that $a_i$ are random integers chosen uniformly from the set $[1, pqr]$. We ignored the 
truncated slice when $pqr < a_i \leq pqr + s$, since the probability of this event is negligible, and
we were interested in the lower bound of the first moment. Now, we need to compute 
an upper bound of the second moment of the random variable $\mathcal{C}^{\mathsf{(int)}}$.

\begin{lemma}\label{lemma: second_moment_integer}
    Let $M$ and $n$ be the same as defined in the problem \textbf{ACI-RESS}. 
    Then the following holds for the \cref{alg: averange-case-integer}
    \begin{align*}
        \underset{\vec{a}}{\mathbb{E}}\left(\mathcal{C}^{\mathsf{(int)}}{}^2\right) = 
        \mathcal{O}\left(1\right).
    \end{align*}
\end{lemma} 
\begin{proof}
    Let us first define another random variable that denotes the number of non-trivial modular collisions in the final 
    list $B_{s_0, s_1, s_2}$ as follows
    \begin{equation*}
        \begin{aligned}
            \mathcal{C}^{(\mathsf{pqr})} := \left|\left\{(\vec{u}_i, \vec{u}_j) \; \middle | \;  
            \begin{aligned}
                i < j, \; \vec{u}_i, \vec{u}_j \in  B_{s_0, s_1, s_2}, \\
                \vec{a}^{(r)}\cdot\vec{u}_i \equiv 
                \vec{a}^{(r)}\cdot\vec{u}_j \mod r.
            \end{aligned}\right\}\right|.
        \end{aligned}
    \end{equation*}
    Now, we make the following crucial observation relating the two random variables $\mathcal{C}^{(\mathsf{int})}$ and $\mathcal{C}^{(\mathsf{pqr})}$
    as follows
    \begin{align*}
        \mathcal{C}^{(\mathsf{int})} \leq \mathcal{C}^{(\mathsf{pqr})}.
    \end{align*}
    This implies
    \begin{align*}
        \underset{\vec{a}}{\mathbb{E}}\left(\mathcal{C}^{\mathsf{(int)}}{}^2\right) \leq
        \underset{\vec{a}}{\mathbb{E}}\left(\mathcal{C}^{\mathsf{(pqr)}}{}^2\right).
    \end{align*}

    In the modular case, we analyzed the second moment of the random
    variable $\mathcal{C}^{\mathsf{(pqr)}}$ given the random vector $\vec{a}$ is uniformly
    distributed over the ring $\mathbb{Z}_{pqr}$. 

    In the integer case, the input random vector $\vec{a}$ is uniformly distributed over the 
    set $[M]^n$, here we can observe that
    the variables $a_i^{(p)} := a_i \mod p$, $a_i^{(q)} := a_i \mod q$,
    and $a_i^{(r)} := a_i \mod r$ has the following distribution
    \begin{equation*}
        \begin{aligned}
            \Pr_{a_i}\left(a_i^{(p)} = \alpha\right) &= 
            \frac{1}{p}\cdot\left(1 + \mathcal{O}\left(\frac{p}{M}\right)\right)\\
            \Pr_{a_i}\left(a_i^{(q)} = \beta\right) &= 
            \frac{1}{q}\cdot\left(1 + \mathcal{O}\left(\frac{q}{M}\right)\right)\\
            \Pr_{a_i}\left(a_i^{(r)} = \gamma\right) &= 
            \frac{1}{r}\cdot\left(1 + \mathcal{O}\left(\frac{r}{M}\right)\right),
        \end{aligned}
    \end{equation*}
    $\forall$ $1 \leq i \leq n$, $\alpha \in \mathbb{Z}_p$, 
    $\beta \in \mathbb{Z}_q$, and $\gamma \in \mathbb{Z}_r$.
    Since, $p, q$, and $r$ are significantly smaller than the input integer $M$, 
    we can safely assume that the random variables $\vec{a}^{(p)}$, $\vec{a}^{(q)}$, and 
    $\vec{a}^{(r)}$ are uniformly distributed over the rings $\mathbb{Z}_{p}$, $\mathbb{Z}_{q}$,
    and $\mathbb{Z}_{r}$ respectively.
    Now, using the second-moment computation from the modular case, we obtain the desired result.
\end{proof}

\subsubsection{Concentration}
Similar to the modular case, we use the Paley-Zygmund inequality to prove that the random
variable $\mathcal{C}^{\mathsf{(int)}}$ is concentrated around its mean. Using \cref{lemma: first_moment_integer} and \cref{lemma: second_moment_integer} and this inequality, 
we get the desired success probability of the
\cref{alg: averange-case-integer}. This completes the proof of \cref{lemma: Filtering_algo_succ_integer}.

\subsection{Complexity Analysis}
\begin{lemma}\label{lemma: complexity_integer}
    Let $M$ and $n$ be the same as defined in the problem \textbf{ACI-RESS}.
    Then \cref{alg: averange-case-integer} runs in expected time 
    $\poly(n) M^{1/4}$ over the random choice of the input vector
    $\vec{a}$. Let $\mathsf{T}$ be the running time of the algorithm. Then, we also have the following tail bound
    \begin{align*}
        \Pr_{\vec{a}} \left(\mathsf{T} \geq \poly(n) \cdot M^{1/4}\right) =  
        \mathcal{O}\left(\frac{1}{\poly(n)}\right).
    \end{align*}
\end{lemma}
\begin{proof}
    Since the filtering algorithm only depends on the lower $pq$ bits of the
    input variables, 
    and as we observed in \cref{sec: second_moment_integer}, the input vector mod $p$
    and $q$ is asymptotically uniform over the corresponding support; therefore, the expected sizes of the list generated
    in the algorithm asymptotically remain the same as in the modular case. 
    Hence, the time complexity in the integer case asymptotically remains the same as the modular case.
\end{proof}

\subsection{Distribution of the sum of Integer random Variables}
In the integer setting, we repeatedly encounter the distribution of a sum of a large number of i.i.d. integer-valued random variables. Local Central Limit Theorem (LCLT) gives a good
approximation of this distribution. This theorem has been 
studied extensively in probability theory and statistics (see, e.g. \cite{ibragimov1965independent}, \cite{mitzenmacher2017probability}, \cite{lawler2010random}, \cite{feller1991introduction}). 
In our analysis, we make use of the following corollary of the LCLT.

\begin{theorem} \label{cor: Theta_LCLT} \cite{mitzenmacher2017probability} 
    \textbf{[Sharper version of LCLT.]} \label{thm:LCLT}
    Let $X_1, X_2, \ldots, X_n$ be a sequence of I.I.D, distributed over the set $[M]$ for some integer $M \geq 1$ 
    with $\mathbb{E}(X_i) = \mu$
    and $\mathsf{Var}(X_i) = \sigma^2 < \infty$. Then
    there exists $c_1, c_2, n_0, C > 0$ such that for all $n \geq n_0$ and for all integers $k$ satisfying 
    \begin{align*}
        |k - n\mu| \leq C\sigma \sqrt{n} \qquad (\text{Central Window of Gaussian}),
    \end{align*}
    we have
    \begin{align*}
        \frac{c_1}{\sigma\sqrt{n}} \leq \Pr\left(\Bar{X} = k\right) \leq 
        \frac{c_2}{\sigma\sqrt{n}},
    \end{align*}
    and hence
    \begin{align*}
        \Pr\left(\Bar{X} = k\right) = \Theta\left(\frac{1}{\sigma \sqrt{n}}\right).
    \end{align*}
\end{theorem}

\begin{remark}
The assumptions of the local central limit theorem remain valid when
$M = M(n)$, since the variables remain uniformly bounded and their variance is strictly positive for all $n$. 
\end{remark}

Recall that, in the integer case, $a_{i}^{(pq)} \in [0, pq)$
and $a_{i}^{(r)} \in [0, r)$. These are integers and not field elements.
Let $\mu_{pq}$ and $\sigma_{pq}^2$ denote the mean and the variance of the random variables $a_{i}^{(pq)}$ $\forall$ $1 \leq i \leq n$, then we have
\begin{equation*}
    \begin{aligned}
        \mu_{pq} &:= \mathbb{E}\left(a_i^{(pq)}\right) = \frac{pq - 1}{2} = \Theta(pq), \\
        \sigma_{pq}^2 &:= \mathsf{Var}\left(a_i^{(pq)}\right) = \Theta((pq)^2) < \infty.
    \end{aligned}
\end{equation*}
Since the random variable $\{a_i^{(pq)}\}_{i = 1}^{n}$ are I.I.D. and their variance is bounded, we can use~\cref{cor: Theta_LCLT} to compute 
the distribution of their sum as follows

\begin{corollary}\label{cor: LCLT_TS_pq}
    For a fixed vector $\vec{u}_i \in \mathsf{A}$ with Hamming weight $w_i = \Theta(n)$, for any
    $s_{pq} = w_i \mu_{pq} + \mathcal{O}(\sigma_{pq}\sqrt{w_i})$ we have the following 
    \begin{equation*}
        \begin{aligned}
            \Pr_{\vec{a}^{(pq)}} 
            \left(S_{\mathsf{TS}, pq, i} = \vec{a}^{(pq)} \cdot \vec{u}_i = s_{pq}\right)
            &= \Theta\left(\frac{1}{\sigma_{pq}\sqrt{w_i}}\right) \\
            &= \Theta\left(\frac{1}{pq\sqrt{n}}\right).
        \end{aligned}
    \end{equation*}
\end{corollary}

Similarly, we can compute the distribution of the sum of the random
variables $\{a_i^{(r)}\}_{i = 1}^{n}$.
In analyzing the expectation of the 
random variable $\mathcal{C}^{(\mathsf{int})}$, the bottleneck problem is to understand the 
interdependency of the following events depending on the vector $\vec{a}^{(pq)}$, 
\begin{itemize}
    \item $S_{\mathsf{LS}, i}^{(p)} = \vec{a}^{(p)}{}^{(L)} \cdot \vec{u}_i^{(L)} 
    \mod p = k$,
    \item $S_{\mathsf{TS}, i}^{(pq)} = \vec{a}^{(pq)} \cdot \vec{u}_i = l$,
    \item $S_{\mathsf{C}, pq, i} = \floor{\frac{\vec{a}^{(pq)}\cdot \vec{u}_i}{pq}} = c$.
\end{itemize}
for some $k \in [0, p)$, $l \in [0, pq)$, and $c = [0, w_i)$.
The following lemma provides a detailed account of the interdependencies among these three events.
\begin{lemma}\label{lemma: pq_events_analysis}
    Let $n$, $p$, $q$, $\vec{a}^{(pq)}$ be the same as defined earlier. Fix $\vec{u}_i \in \mathsf{A}$, let $w_i, w_{l, i}$, and $w_{r, i}$
    denote the Hamming weights of the vectors $\vec{u}_i, \vec{u}_i^{(L)}$, and $\vec{u}_i^{(R)}$
    respectively. Assume that $\alpha n < w_i < (1-\alpha)n$ and $\alpha n/2 < w_{l, i}, w_{r, i} < (1- \alpha )n/2$ where $1/8 \leq \alpha < 1/4$ is some constant. Then
    \begin{enumerate}
        \item For all $c = w_i/2 + \mathcal{O}(\sqrt{w_i})$, we have
            \begin{align*}
                \Pr_{\vec{a}^{(pq)}}\left(S_{\mathsf{C}, pq, i} = c\right) = 
                \Theta\left(\frac{1}{\sqrt{w_i}}\right).
            \end{align*}
        \item For all $l \in [0, pq)$ and $c = w_i/2 + \Theta(\sqrt{w_i})$, we have
            \begin{align*}
                \Pr_{\vec{a}^{(pq)}}\left(S_{\mathsf{TS}, i}^{(pq)} = l \; \middle |\; 
                S_{\mathsf{C}, pq, i} = c\right) = \Theta\left(\frac{1}{pq}\right).
            \end{align*}
        \item For all $k \in [0, p)$, $l \in [0, pq)$, and 
        $c = w_i/2 + \Theta(\sqrt{w_i})$, we have
            \begin{align*}
                \Pr_{\vec{a}^{(pq)}}\left(
                S_{\mathsf{LS}, i}^{(p)} = k, 
                S_{\mathsf{TS}, i}^{(pq)} = l, 
                S_{\mathsf{C}, pq, i} = c\right) \\ = \Theta\left(\frac{1}{p^2q\sqrt{w_i}}\right).
            \end{align*}
    \end{enumerate}
\end{lemma} 

\begin{proof}
    Leveraging the smoothness of the sum distributions, we first show that the carry bits converge in distribution to a discrete normal distribution with mean $w_i/2$ and variance $\Theta(w_i)$,
    and henceforth we prove that the total sum modulo $pq$
    and the left sum modulo $p$ condition on the carry bits also 
    converge to the uniform distribution, as $n$ tends to infinity. 
    Finally, we show that the following three events 
    are asymptotically independent
    \begin{itemize}
        \item $\sum_{j = 1}^{n/2} a_j^{(pq)} \cdot u_{ij} \equiv_{p} a$,
        \item $\vec{a}^{(pq)} \cdot \vec{u}_i \equiv_{pq} b$,
        \item $\floor{\frac{\vec{a}^{(pq)} \cdot \vec{u}_i}{pq}} = c$.
    \end{itemize}
    Refer to \cref{sec: pq_events_analysis} for a detailed proof. 
\end{proof}

\subsubsection{Proof of \cref{prop: events_vec_a_probabilities}}\label{sec: events_vec_a_probabilities}
\begin{enumerate}
    \item Since $\vec{u}_i$ and $\vec{u}_j$ do not overlap, we can split the 
        independent events as follows
        \begin{equation*}
            \begin{aligned}
                &\Pr_{\vec{a}^{(pq)}} \left(\mathsf{E}_{\mathsf{LS}, \vec{a}, i}^{(p)} \wedge 
                \mathsf{E}_{\mathsf{LS}, \vec{a}, j}^{(p)} \wedge \mathsf{E}_{\mathsf{TS}, \vec{a}, i}^{(pq)} \wedge \mathsf{E}_{\mathsf{TS}, \vec{a}, j}^{(pq)} \wedge 
                \mathsf{E}_{\mathsf{C}, pq, \vec{a}, i, j} \right) \\ 
                &= 
                \Pr_{\vec{a}^{(pq)}} \left(\mathsf{E}_{\mathsf{LS}, \vec{a}, i}^{(p)} \wedge 
                \mathsf{E}_{\mathsf{LS}, \vec{a}, j}^{(p)} \wedge \mathsf{E}_{\mathsf{TS}, \vec{a}, i}^{(pq)} \wedge \mathsf{E}_{\mathsf{TS}, \vec{a}, j}^{(pq)} \wedge 
                \left(S_{\mathsf{C}, pq, i} = S_{\mathsf{C}, pq, j}\right)
                \right)\\
                &= \sum_{c}
                \Pr_{\vec{a}^{(pq)}} \left(\mathsf{E}_{\mathsf{LS}, \vec{a}, i}^{(p)} \wedge 
                \mathsf{E}_{\mathsf{LS}, \vec{a}, j}^{(p)} \wedge \mathsf{E}_{\mathsf{TS}, \vec{a}, i}^{(pq)} \wedge \mathsf{E}_{\mathsf{TS}, \vec{a}, j}^{(pq)} \wedge 
                \left(S_{\mathsf{C}, pq, i} = S_{\mathsf{C}, pq, j} = c\right)
                \right)\\
                &= \sum_{c}
                \Pr_{\vec{a}^{(pq)}}\left(
                    S_{\mathsf{LS}, i}^{(p)} = s_0, 
                    S_{\mathsf{TS}, i}^{(pq)} = \beta, 
                    S_{\mathsf{C}, pq, i} = c\right) \cdot 
                    \Pr_{\vec{a}^{(pq)}}\left(
                    S_{\mathsf{LS}, j}^{(p)} = s_0, 
                    S_{\mathsf{TS}, j}^{(pq)} = \beta, 
                    S_{\mathsf{C}, pq, j} = c\right)\\
                &> \sum_{c = W^{(C)}}
                \Pr_{\vec{a}^{(pq)}}\left(
                    S_{\mathsf{LS}, i}^{(p)} = s_0, 
                    S_{\mathsf{TS}, i}^{(pq)} = \beta, 
                    S_{\mathsf{C}, pq, i} = c\right) \cdot 
                    \Pr_{\vec{a}^{(pq)}}\left(
                    S_{\mathsf{LS}, j}^{(p)} = s_0, 
                    S_{\mathsf{TS}, j}^{(pq)} = \beta, 
                    S_{\mathsf{C}, pq, j} = c\right),
            \end{aligned}
        \end{equation*}
        where $W^{(C)}$ is the intersection of the Gaussian central windows of 
        the variables $S_{\mathsf{C}, pq, i}$ and $S_{\mathsf{C}, pq, j}$.
        Then, using \cref{lemma: pq_events_analysis}, we have the following
        \begin{equation*}
            \begin{aligned}
                \Pr_{\vec{a}^{(pq)}} \left(\mathsf{E}_{\mathsf{LS}, \vec{a}, i}^{(p)} \wedge 
                \mathsf{E}_{\mathsf{LS}, \vec{a}, j}^{(p)} \wedge \mathsf{E}_{\mathsf{TS}, \vec{a}, i}^{(pq)} 
                \wedge \mathsf{E}_{\mathsf{TS}, \vec{a}, j}^{(pq)} \wedge 
                \mathsf{E}_{\mathsf{C}, pq, \vec{a}, i, j} \right)
                &> \sum_{t = W^{(C)}} \left(
                 \frac{1}{p^2q} \cdot \Theta\left(\frac{1}{\sqrt{n}}\right)
                \right)^2 \\
                &= \sum_{t = W^{(C)}} \Theta\left(\frac{1}{p^4q^2 n}\right) \\
                &= \Theta\left(\frac{1}{p^4q^2\sqrt{n}}\right),
            \end{aligned}
        \end{equation*}
    since the window $W^{(C)}$ has $\Theta(\sqrt{n})$
    many points. This completes the proof of part $[1]$.
    \item Firstly, we observe the following
        \begin{equation*}
            \begin{aligned}
                \mu_{r} &:= \mathbb{E}\left(a_{j}^{(r)}\right) = 
                \frac{r-1}{2} = \Theta(r), \\
                \sigma_{r}^2 &:= \mathsf{Var}\left(a_{j}^{(r)}\right) = \Theta(r^2).
            \end{aligned}
        \end{equation*}
        Similar to the previous part,
        since the vectors $\vec{u}_i$ and $\vec{u}_j$ do not overlap, we have 
        the following
        \begin{equation*}\label{eq: lemma_5_2_1_1}
            \begin{aligned}
                &\Pr_{\vec{a}^{(r)}}\left(S_{\mathsf{TS}, r, i} = S_{\mathsf{TS}, r, j}\right) \\
                &= \sum_{t} 
                \Pr_{\vec{a}^{(r)}}\left(S_{\mathsf{TS}, r, i} = t\right)
                \cdot \Pr_{\vec{a}^{(r)}}\left(S_{\mathsf{TS}, r, j} = t\right).\\
            \end{aligned}
        \end{equation*}
        Let $W_i^{(r)} := w_i \mu_r + \mathcal{O}(\sigma_r \sqrt{w_i})$ and 
        $W_j^{(r)} := w_j \mu_r + \mathcal{O}(\sigma_r \sqrt{w_j})$ be the Gaussian central windows of the random variables $S_{\mathsf{TS}, r, i}$
        and $S_{\mathsf{TS}, r, j}$ respectively. 
        To capture most of the mass, we take the intersection of the two windows. Let $W^{(r)} := 
        W_i^{(r)} \cap W_j^{(r)}$, then we have the following
        \begin{align*}
            &\Pr_{\vec{a}^{(r)}}\left(S_{\mathsf{TS}, r, i} = S_{\mathsf{TS}, r, j}\right) \\
            &= \sum_{t} 
                \Pr_{\vec{a}^{(r)}}\left(S_{\mathsf{TS}, r, i} = t\right)
                \cdot \Pr_{\vec{a}^{(r)}}\left(S_{\mathsf{TS}, r, j} = t\right) \\
            &> \sum_{t \in W^{(r)}} \Pr_{\vec{a}^{(r)}}\left(S_{\mathsf{TS}, r, i} = t\right)
                \cdot \Pr_{\vec{a}^{(r)}}\left(S_{\mathsf{TS}, r, j} = t\right).
        \end{align*}

        Now, using \cref{thm:LCLT} to compute the distributions of the 
        random variables $S_{\mathsf{TS}, r, i}$ and $S_{\mathsf{TS}, r, j}$,
         we have the following 
        \begin{align*}
            \Pr_{\vec{a}^{(r)}}\left(S_{\mathsf{TS}, r, i} = S_{\mathsf{TS}, r, j}\right) &= 
            \sum_{t \in W^{(r)}} \Theta\left(\frac{1}{\sigma_r^2 \sqrt{w_iw_j}}\right).
        \end{align*}
        And, since there are $\Theta(\sigma_r \sqrt{n})$ many points in the window $W^{(r)}$ 
        and $w_i, w_j = \Theta(n)$, we have the desired result. 
        This completes the proof of part $[2]$.
\end{enumerate}

\section{Conclusion and Future Directions}
\begin{strip}
    \centering
    \includegraphics[scale=0.6944]{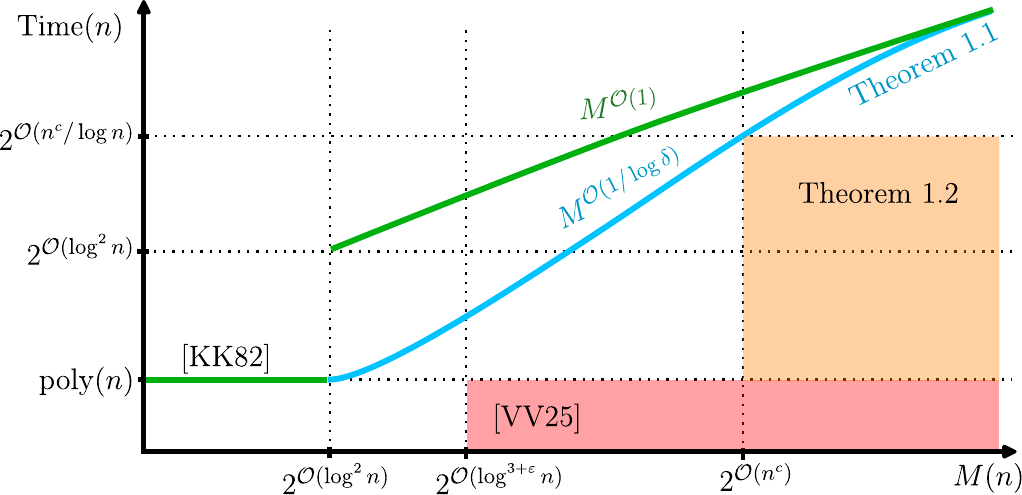}
    \captionof{figure}{Overview of our results. Green lines depict the previous upper bounds for the problem.
    Blue lines depict~\cref{thm:subexponential}.
    The result of Vafa and Vaikuntanathan~\cite{VV25} forbids a polynomial-time (even average case) algorithm for $M = 2^{\Omega(\log^{3+\eps}n)}$. \cref{thm:lowerbound} shows a tight lowerbound in the setting when $M(n) = 2^{\Theta(n^c)}$
    for constant $c \in (0,1)$.
    \label{fig:overview}
    }
\end{strip}
We study the complexity of the \PHES and make progress on two central open
questions: whether one can achieve subexponential-time algorithms and whether
the $\Os(M^{1/3})$ barrier can be improved. Our contributions are
\begin{enumerate}[label=(\roman*)]
    \item Deterministic $\poly(n) \cdot M^{\Oh(1/\log{\delta})}$-time algorithm when $\delta = n/\log{M}$ is the density of the instance (\cref{thm:subexponential}).
    \item Conditional lower bound showing that no $M^{o(1/\log{\delta})}$-time algorithm exists assuming currently based algorithms for lattice problems are tight (\cref{thm:lowerbound}).
    \item  Average-case, $\Oh^{\star}(2^{n/4})$-time algorithms when $M = \Theta(2^n)$ 
    (\cref{thm:average-case-integer} and \cref{thm:average-case-modular})
\end{enumerate}

Figure~\ref{fig:overview} gives a pictorial view of our results.
We also show that no $2^{o(n/\log n)}$-time algorithm exists for \PHES
(conditionally answering open question of Jin~\cite{jin-slides} up to $\log{n}$
factor in the exponent).

Several intriguing directions remain open. Most notably, closing the remaining
gap to fully resolve whether \PHES admits a $2^{o(n)}$-time algorithm remains a
major challenge. Our lower bound relies on a technique of Vafa and
Vaikuntanathan~\cite{VV25}, who showed that no polynomial-time algorithm solves
Number Balancing Problem with $\kappa(n) = 2^{-\Omega(\log^{3+\eps}(n))}$.
Although our lower bound is tight, it is tight only in the setting of
exponential running times and does not improve the lower bound of Vafa and
Vaikuntanathan~\cite{VV25}. In fact, our reduction, when plugged in with their
hypothesis and parameters, would also exclude $\kappa(n) =
2^{-\Oh(\log^{3+\eps}(n))}$-approximation in polynomial time, but it would be
worse because the lower bound of~\cite{VV25} prohibits even average-case
algorithms. Can one improve these lower bounds in the setting of polynomial-time
algorithms for \PHES?

Finally, we ask whether average-case algorithms for \PHES in the regime $\delta
= 1$ can be extended to the worst-case setting. We conjecture that the
structural insights of Jin, Williams, and Zhang~\cite{JWZ25} and our analysis in
the average-case setting can be combined to yield an $\Os(2^{n/4})$-time
worst-case algorithm. Also, can we improve the average-case running time exponent of $\Os(2^{n/4})$ below $1/4$?

\appendix

\section{Sampling Lattice Points}\label{sec:sampling}

\begin{definition}[Discrete Gaussian Distribution over a Lattice]
For any vector $\mat c \in \R^n$, real $s > 0$ and a lattice $\mathcal{L} \subset \R^n$, the probability distribution $\mathcal{D}_{\mathcal{L}, s,\mat c}$ over $\mathcal{L}$ is defined as
\[\mathcal{D}_{\mathcal{L}, s, \mat c}(x) = \frac{\rho_{s, \mat c}(x)}{\rho_{s,\mat c}(\mathcal{L})}, \;\;\; \forall \mat x \in \mathcal{L}\]
where $\rho_{s, \mat c}(\mat x) = e^{-\pi \norm{(\mat x-\mat c)/s}^2}$ and $\rho_{s,\mat c}(\mathcal{L}) = \sum_{\mat x \in \mathcal{L}} \rho_{s,\mat c}(\mat x)$.
\end{definition}

For simplicity, we use $\mathcal{D}_{\L, s}$ to denote $\mathcal{D}_{\L, s,\mat 0}$ and $\rho_{s}$ to denote $\rho_{s,\mat 0}$.
\begin{definition}[\cite{MR07}]
    For any $n$-dimensional lattice $\L$ and positive real $\epsilon>0$, the smoothing parameter $\eta_{\epsilon}(\L)$ is the smallest real $s>0$ such that $\rho_{1/s} (\L^* \setminus \{0\}) \leq \epsilon$. Here, $\L^* = \{\mat y \in \R^n \;|\: \langle \mat x, \mat y \rangle \in \Z, \forall \mat x \in \L\}$ is the dual lattice of $\L$.
\end{definition}

\begin{theorem}[{\cite[Lemma 3.1]{GPV08}}]
    \label{thm:smoothing-parameter-bound}
    For any $n$-dimensional lattice $\L$ spanned by a basis $\mathbf{B}$ and positive real $\epsilon>0$, we have
    \[\eta_{\epsilon}(\L) \leq B \sqrt{\frac{\log(2n(1 + 1/\epsilon))}{\pi}}\]
    where $B = \norm{\mathbf{B}}_2$. In particular, for $\epsilon = 2^{-n}$, we have $\eta_{\epsilon}(\L) \leq 2B\sqrt{n}$.
\end{theorem}

\begin{theorem}[Discrete Gaussian sampling from a basis~{\cite[Theorem~4.1]{GPV08}}]\label{thm:dgs_basis}
Let $\mathcal{L} \subset \R^n$ be a full-rank lattice and let $\mathbf{B}=(\mat b_1,\dots, \mat b_n)$ be any basis of $\mathcal{L}$.
Let $\widetilde{\mathbf{B}}=(\widetilde{\mat b}_1,\dots,\widetilde{ \mat b}_n)$ be the Gram--Schmidt orthogonalization of $\mathbf{B}$. Fix any center $c\in\R^n$ and any parameter
\[
s \;\ge\; \|\widetilde{\mathbf{B}}\|_2\cdot \omega(\sqrt{\log n}).
\]
Then there exists a randomized algorithm $\mathsf{SampleD}(B,s,c)$ running in time
$\poly(n,\mathrm{bitlength}(\mathbf{B}),\log s)$ 
that outputs a vector $\mat x\in\mathcal{L}$ whose distribution is within
negligible statistical distance of the discrete Gaussian distribution $\mathcal{D}_{\L,s,\mat c}$ over $\L$.
\end{theorem}

The goal of this section is to prove the following theorem:

\thmsampling*

\begin{proof}[Proof of~\cref{thm:wide_sampling_shift_invariant}]
We describe the sampler and then prove the claimed properties.

\paragraph{The sampler.}
Sample $\mathbf{x}^{(1)},\ldots,\mathbf{x}^{(N)}$ independently from the discrete Gaussian distribution
$\mathcal{D}_{\L,s}$ over $\L$, and output these samples.
This can be performed using the algorithm in \cref{thm:dgs_basis} in time
$\poly(n,\mathrm{bitlength}(\mathbf{B}), \log(s))$ given a basis $\mathbf{B}$ of $\L$ and $s$ which is clearly satisfy the theorem condition.
Hence $N$ independent samples can be generated in time $N\cdot \poly(n,\mathrm{bitlength}(\mathbf{B}), \log(s))$.

\paragraph{Support in an $\ell_\infty$-ball.}
Fix any single sample $x \leftarrow \mathcal{D}_{\L,s}$.
We bound $\Pr[\|x\|_\infty > R]$. A standard tail bound for (discrete) Gaussian measures\footnote{Refer~{\cite[Theorem~12]{DD18}}} implies that for any $t\ge 1$,
\begin{equation*}\label{eq:gauss_tail_1d}
\Pr\big[ \norm{\mat x}_2 \ge t s \big] \;\le\; 2 e^{-\pi t^2}
\end{equation*}

Therefore, for $N$ independent samples $\mathbf{x}^{(1}),\ldots,\mathbf{x}^{(N)}$,
another union bound gives
\begin{align*}
&\Pr\big[ \exists j\in[N]: \|\mathbf{x}^{(j)}\|_\infty > R \big] \\
&\le \Pr\big[ \exists j\in[N]: \|\mathbf{x}^{(j)}\|_\infty > \sqrt{n\log(N)}s \big]\\
&\le \Pr\big[ \exists j\in[N]: \|\mathbf{x}^{(j)}\|_2 > \sqrt{n\log(N)}s \big]\\
&\le 2N\,e^{-\pi n\log(N)}
\;\le \;
2^{-n},
\end{align*}
and hence (except with negligible probability) all outputs lie in $\L\cap \Ball_\infty(0,R)$.

\end{proof}

\begin{lemma}
\label{lem:stat-distance-DGS}
Let $\D$ denote the distribution $\mathcal{D}_{\L, s}$ and $\mat v^*$ be a shortest vector of $\L$. Then, the statistical distance between $\D$ and $\D + \mat v^\ast$ is upper bounded by
\[
\Delta\!\left(\D,\ \D+ \mat v^\ast\right)
\;\le\; \sqrt{\frac{\pi}{2}}\,\frac{\|\mat v^*\|}{s}.
\]
\end{lemma}

\begin{proof}
We bound $\Delta(\D,\D+ \mat v^*)$ for a shortest nonzero vector $\mat v^* \in\L$ given by Pinsker's inequality,
\[
\Delta(\D,\D+ \mat v^*)\le \sqrt{\frac12 \mathcal{D}_{\mathrm{KL}}(\D\|\D+ \mat v^*)}.
\]
where,
\begin{align*}
    \mathcal{D}_{\mathrm{KL}}(\D\|\D+ \mat v^*) &=
    \sum_{\mat x\in\L} \D( \mat x)\log\frac{\D(\mat x)}{\D(\mat x - \mat v^*)} \\
    &= \frac{\pi}{s^2}\sum_{\mat x\in\L} \D(\mat x)\bigl(\|\mat x- \mat v^*\|^2-\|\mat x\|^2\bigr),
\end{align*}
where we used that fact that $\rho_s(\L) = \rho_s(\L +\mat v^*)$. Since $\|\mat x-\mat v^*\|^2-\|\mat x\|^2=-2\langle \mat x, \mat v^*\rangle+\|\mat v^*\|^2$ and $\mathcal{D}_{\L,s}$ is symmetric, we have $\underset{\mat x\sim \D}{\mathbb{E}}[\mat x]=0$. Hence $ \mathcal{D}_{\mathrm{KL}}(\D\|\D+ \mat v^*)=\frac{\pi \|\mat v^*\|^2}{s^2}$. Therefore
\[
\Delta(\mathcal{D}_{\L,s},\mathcal{D}_{\L,s}+ \mat v^*)\le \sqrt{\frac{\pi}{2}}\,\frac{\|\mat v^*\|}{s}.
\]
\end{proof}

\section{Proof of \cref{lem:short-basis}}
\label{sec:proof-short-basis}

In this section, we prove the following statement. 

\lemshortbasis*

This statement is folklore and follows from the previous works. As we did not
find the statement with a bound on $\lambda_n$, we include it for completeness. We
start by showing that a $\gamma(n)$ approximate SVP algorithm can be used to
build a $\sqrt{n} \gamma(n)$ approximate SIVP algorithm.

\begin{theorem}
\label{thm:SIVP-SVP}
Let $\gamma : \N \to \R_{\ge 1}$ be non-decreasing. Suppose there is an oracle that, on input any rank-$m$ lattice $\mathcal{M}$, returns a nonzero vector $\mat x \in \mathcal{M}$ such that
\[
\|\mat x\|_2 \le \gamma(m)\,\lambda_1(\mathcal{M}).
\]
Then there is a polynomial-time algorithm that, given an $n$-dimensional lattice $\mathcal{L}$, makes $n$ oracle calls and outputs $n$ linearly independent vectors $\mat b_1,\dots,\mat b_n \in \mathcal{L}$ such that
\[
\|\mat b_i\|_2 \le \sqrt{n}\,\gamma(n)\,\lambda_n(\mathcal{L})
\qquad\text{for all } i\in[n].
\]
In particular, this gives a polynomial-time reduction from $\gamma(n)$-SVP to $\sqrt{n}\,\gamma(n)$-SIVP.
\end{theorem}

\begin{proof}
We construct the vectors $\mat b_1,\dots,\mat b_n$ iteratively. Assume that $\mat b_1,\dots, \mat b_{i-1}\in\mathcal{L}$ have already been constructed and are
linearly independent. Let $W_{i-1}:=\mathsf{span}(\mat b_1,\dots,\mat b_{i-1})$, and let $\pi_i : \R^n \to W_{i-1}^{\perp}$ be the orthogonal projection onto the orthogonal complement of $W_{i-1}$. Consider the projected lattice
\[
\mathcal{L}_i := \pi_i(\mathcal{L}) \subseteq W_{i-1}^{\perp}.
\]
Since $\dim(W_{i-1}^{\perp})=n-i+1$, the lattice $\mathcal{L}_i$ has rank $n-i+1$. We call the $\gamma$-SVP oracle on $\mathcal{L}_i$, obtaining a nonzero vector
$\mat y_i \in \mathcal{L}_i$ such that
\[
\|\mat y_i\|_2 \le \gamma(n-i+1)\,\lambda_1(\mathcal{L}_i).
\]
Choose any lift $\mat x_i \in \mathcal{L}$ such that $\pi_i(\mat x_i)= \mat y_i$. Now size-reduce $\mat x_i$ with respect to $\mat b_1,\dots, \mat b_{i-1}$, i.e., write $\mat x_i = \sum_{j=1}^{i-1} \alpha_j \mat b_j + \mat r_i,$ where $\mat r_i \in \R^n$, and let $\mat b_i := \mat x_i - \sum_{j=1}^{i-1} \lfloor \alpha_j \rceil\, \mat b_j$. Then $\mat b_i \in \mathcal{L}$, and since subtracting vectors in $W_{i-1}$ does not change the projection onto $W_{i-1}^{\perp}$, we still have $\pi_i(\mat b_i)=\mat y_i$. In particular, $\mat y_i\neq 0$ implies $\mat b_i\notin W_{i-1}$, so $\mat b_1,\dots, \mat b_i$ are linearly independent.

Let $\mat b_1^*,\dots,\mat b_i^*$ denote the Gram--Schmidt orthogonalization of
$\mat b_1,\dots,\mat b_i$. By construction, $\mat b_i^*=\pi_i(\mat b_i)= \mat y_i$, and the size-reduction step ensures that the Gram--Schmidt coefficients satisfy
\[
|\mu_{i,j}| \le \frac12
\qquad\text{for all } j<i.
\]

We now bound $\|\mat b_i^*\|_2$. By definition of $\lambda_i(\mathcal{L})$, there exist
$i$ linearly independent vectors $\mat v_1,\dots,\mat v_i\in\mathcal{L}$ such that $\|\mat v_t\|_2 \le \lambda_i(\mathcal{L})$ for all $t\in[i]$. Since $W_{i-1}$ has dimension $i-1$, not all of these $i$ vectors can lie in $W_{i-1}$. Hence there exists $t\in[i]$ such that $\pi_i(\mat v_t)\neq 0$. Therefore, $\pi_i(\mat v_t)$ is a nonzero vector of $\mathcal{L}_i$, and so
\[
\lambda_1(\mathcal{L}_i) \le \|\pi_i(\mat v_t)\|_2 \le \|\mat v_t\|_2 \le \lambda_i(\mathcal{L}).
\]
Using the oracle guarantee and the monotonicity of $\gamma$, we obtain
\begin{align*}
    \|\mat b_i^*\|_2
    = \|\mat y_i\|_2
    &\le \gamma(n-i+1)\,\lambda_1(\mathcal{L}_i),\\
    &\le \gamma(n)\,\lambda_i(\mathcal{L}),\\
    &\le \gamma(n)\,\lambda_n(\mathcal{L}).
\end{align*}
Finally, since $\mat b_i = \mat b_i^* + \sum_{j=1}^{i-1} \mu_{i,j} \mat b_j^*$, the orthogonality of the Gram-Schmidt vectors gives
\[
\|\mat b_i\|_2^2
= \|\mat b_i^*\|_2^2 + \sum_{j=1}^{i-1} \mu_{i,j}^2 \|\mat b_j^*\|_2^2.
\]
Using $|\mu_{i,j}|\le 1/2$ and $\|\mat b_j^*\|_2 \le \gamma(n)\lambda_n(\mathcal{L})$,
we get
\begin{align*}
    \|\mat b_i\|_2^2 &\le \gamma(n)^2\lambda_n(\mathcal{L})^2 + 
    \frac14(i-1)\gamma(n)^2\lambda_n(\mathcal{L})^2,\\
    & \le n\,\gamma(n)^2\lambda_n(\mathcal{L})^2.
\end{align*}
Therefore,
$
\|\mat b_i\|_2 \le \sqrt{n}\,\gamma(n)\,\lambda_n(\mathcal{L})
$
for all $i \in [n]$
\end{proof}

Next, we use the best algorithms for the approximate Shortest Vector Problem.
This follows from the block reduction introduced by Schnorr~\cite{schnorr1987hierarchy}.

\begin{theorem}[\cite{schnorr1987hierarchy,schnorr1994lattice,GN08,ALNS20}] \label{thm:ALS-approx-SVP}
Let $\kappa \ge 1$. There exists a randomized algorithm that, given an $n$-dimensional lattice $\mathcal{L}$, computes a nonzero vector $\mat v \in \mathcal{L}$ in time $2^{\Oh(n/\kappa)}$ such that
\[
\|\mat v\|_2 \leq n^{\kappa} \cdot \lambda_1(\mathcal{L}) .
\]
\end{theorem}

To the best of our knowledge, the currently best algorithm actually runs in time
$\Os(2^{\frac{0.802n}{\kappa+1}})$ and returns $n^\kappa$-approximation
(see~\cite[Table 1]{ALNS20} for historical overview).
Nevertheless,~\cref{thm:ALS-approx-SVP} is sufficient as the exact dependence on
behind $\Oh(\cdot)$ does not matter for us.

Combining \cref{thm:SIVP-SVP} and \cref{thm:ALS-approx-SVP} gives \cref{lem:short-basis}.

\subsection{Variant of \cref{thm:SIVP-SVP} for $n$-SMP}\label{sec:nsmp-variant}

Also, we need the following slightly different variant of~\cref{thm:SIVP-SVP}, when the oracle solves $n$-SMP.
\begin{theorem}\label{thm:SIVP-SVP-variant}
Let $\gamma : \N \to \R_{\ge 1}$ be non-decreasing. Suppose there is an oracle that, on input any rank-$m$ lattice $\mathcal{M}$, returns a nonzero vector $\mat x \in \mathcal{M}$ such that
\[
\|\mat x\|_2 \le \gamma(m)\,\lambda_m(\mathcal{M}).
\]
Then there is a polynomial-time algorithm that, given an $n$-dimensional lattice $\mathcal{L}$, makes $n$ oracle calls and outputs $n$ linearly independent vectors $\mat b_1,\dots, \mat b_n \in \mathcal{L}$ such that
\[
\|\mat b_i\|_2 \le \sqrt{n}\,\gamma(n)\,\lambda_n(\mathcal{L})
\qquad\text{for all } i\in[n].
\]
\end{theorem}

\begin{proof}[Proof sketch]
The proof is modified only in the step where we bound the oracle output on the projected lattice. At iteration $i$, let $\mathcal{L}_i := \pi_i(\mathcal{L})$, where $\pi_i$ is the orthogonal projection onto $W_{i-1}^{\perp}$ and $W_{i-1}=\mathsf{span}(\mat b_1,\dots,\mat b_{i-1})$. Then $\mathcal{L}_i$ has rank $m=n-i+1$. So the oracle now returns a nonzero vector $\mat y_i \in \mathcal{L}_i$ satisfying
\[
\|\mat y_i\|_2 \le \gamma(m)\,\lambda_m(\mathcal{L}_i).
\]
The rest of the construction of $\mat b_i$ from a lift of $\mat y_i$ is unchanged. The only point to reprove is the bound on $\|\mat b_i^*\|_2=\|\mat y_i\|_2$. For this, take $n$ linearly independent vectors $\mat v_1,\dots,\mat v_n \in \mathcal{L}$ such that
\[
\|\mat v_t\|_2 \le \lambda_n(\mathcal{L})
\qquad\text{for all } t\in[n].
\]
Since $W_{i-1}$ has dimension $i-1$, at most $i-1$ of these vectors can lie in $W_{i-1}$. Hence at least $n-(i-1)=m$ of the projected vectors $\pi_i(v_t)$ are nonzero. Moreover, these projected vectors are linearly independent in $\mathcal{L}_i$. Indeed, if a nontrivial linear combination of them were zero, then the same combination of the corresponding $\mat v_t$'s would lie in $W_{i-1}$, contradicting the linear independence of the $\mat v_t$'s together with $\dim(W_{i-1})=i-1$.

Therefore, $\mathcal{L}_i$ contains $m$ linearly independent vectors, each of norm at most $\lambda_n(\mathcal{L})$, because orthogonal projection cannot increase Euclidean norm. It follows that $\lambda_m(\mathcal{L}_i) \le \lambda_n(\mathcal{L})$. Hence
\[
\|\mat b_i^*\|_2=\|\mat y_i\|_2
\le \gamma(m)\lambda_m(\mathcal{L}_i)
\le \gamma(n)\lambda_n(\mathcal{L}),
\]
where we used the monotonicity of $\gamma$.

From this point onward, the remainder of the proof is the same: the size-reduction step gives $|\mu_{i,j}|\le 1/2$, and therefore
\[
\|\mat b_i\|_2^2
= \|\mat b_i^*\|_2^2 + \sum_{j=1}^{i-1} \mu_{i,j}^2 \|\mat b_j^*\|_2^2
\le n\,\gamma(n)^2 \lambda_n(\mathcal{L})^2,
\]
which implies
\[
\|\mat b_i\|_2 \le \sqrt{n}\,\gamma(n)\,\lambda_n(\mathcal{L}).
\]
\end{proof}

\section{Statistical Distance}

\begin{lemma}
\label{lem:stat-dist-prod-binary}
Let $D_0$ and $D_1$ be two distributions on the same sample space such that $\Delta(D_0,D_1) \leq \epsilon$. Consider the following two distributions on $n$-tuples.

\begin{itemize}
    \item $P := D_0 \times D_0 \times \cdots \times D_0 = D_0^n$.
    \item $Q$: sample independent uniform bits $b_1,\dots,b_n \in \{0,1\}$, and then output an $n$-tuple whose $i$-th coordinate is sampled independently from $D_{b_i}$.
\end{itemize}

Then $\Delta(P,Q) \leq n\epsilon/2$.
\end{lemma}

\begin{proof}
We define hybrid distributions $H_0,H_1,\dots,H_n$ by $H_k := D_1^k \times D_0^{\,n-k}$, for $k=0,1,\dots,n$. In particular, $H_0 = D_0^n = P$ and $H_n = D_1^n$. However, the distribution $Q$ is not equal to $H_n$. Instead, it is the product distribution
\[
Q = \left(\frac{D_0+D_1}{2}\right)^n,
\]
since each coordinate is sampled independently from $D_0$ or $D_1$ with equal probability. We now prove the desired bound by changing one coordinate at a time. Define
\[
G_k := \left(\frac{D_0+D_1}{2}\right)^k \times D_0^{\,n-k},
\]
for $k=0,1,\dots,n$. Then $G_0 = P$ and $G_n = \left(\frac{D_0+D_1}{2}\right)^n = Q$. By the triangle inequality,
\[
\Delta(P,Q)
=
\Delta(G_0,G_n)
\leq
\sum_{k=0}^{n-1} \Delta(G_k,G_{k+1}).
\]
Now $G_k$ and $G_{k+1}$ differ only in the $(k+1)$-st coordinate, so
\[
\Delta(G_k,G_{k+1})
=
\Delta\!\left(D_0,\frac{D_0+D_1}{2}\right).
\]
Using the definition of statistical distance,
\begin{align*}
    \Delta\!\left(D_0,\frac{D_0+D_1}{2}\right)
    &=
    \frac{1}{2}\sum_x \left|D_0(x)-\frac{D_0(x)+D_1(x)}{2}\right|,\\
    &=
    \frac{1}{4}\sum_x |D_0(x)-D_1(x)|,\\
    &= \frac{1}{2}\Delta(D_0,D_1), \\
    &\leq \frac{\epsilon}{2}.
\end{align*}
Therefore,
\[
\Delta(P,Q)
\leq
\sum_{k=0}^{n-1} \frac{\epsilon}{2}
=
\frac{n\epsilon}{2}.
\]
\end{proof}

Notice that the hypothesis of Vafa and Vaikuntanathan~\cite{VV25} considered SIVP, whereas~\cref{hyp:weak-sivp} concerns nSMP.
We show that these hypotheses are tightly related.  Suppose there exist constants $\eta, c_1, c_2 > 0$ such that nSMP with approximation factor $2^{(c_1 n^{1/2-\eta})}$ can be solved in time $2^{(c_2 n^{1/2-\eta})}$. Then, by~\cref{thm:SIVP-SVP-variant}, this implies a polynomial-time reduction to SIVP with an approximation factor
\[
    \sqrt{n}\cdot 2^{(c_1 n^{1/2-\eta})}
    \;=\;
    2^{(c_1 n^{1/2-\eta} + \log n/2)},
\]
and running time
\[
    \poly(n)\cdot 2^{(c_2 n^{1/2-\eta})}
    \;\leq\;
    2^{(c_2 n^{1/2-\eta} + c_2' \log n)}
\]
for some constant $c_2' > 0$. In particular, for all sufficiently large $n$, this yields an algorithm for $2^{2c_1 n^{(1/2)-\eta}}$ approximate SIVP running in time $2^{(2 c_2 n^{1/2-\eta})}$.

\section{Additional Proofs for \cref{sec: average-case-modular}}

\subsection{Proof of Claim~\ref{claim: Events_a_r_second_moment}} \label{sec: Events_a_r_second_moment}

\begin{proof}
Both $\vec{w}, \vec{w}'$ lie in $\{0,\pm1\}^n$ and are non-zero, since
$\vec{u}_i \ne \vec{u}_j$ and $\vec{u}_l \ne \vec{u}_k$.

\emph{A counting identity.} For $\vec{c} \in \Z_r^n$, the inner product
$\vec{c}\cdot\vec{a}^{(r)}$ is uniform on $\Z_r$ if some $c_t$ is a unit mod $r$.
Hence, in general, the map $\vec{a}^{(r)} \mapsto (\vec{w}\cdot\vec{a}^{(r)}, \vec{w}'\cdot\vec{a}^{(r)})$
is a group homomorphism $\Z_r^n \to \Z_r^2$. Under the uniform distribution on $\Z_r^n$, the
probability that it maps to $(0,0)$ equals $1/|\mathrm{Im}|$, and
$|\mathrm{Im}| = r^2/N$, where
\[
  N := \#\{(x,y) \in \Z_r^2 : x\vec{w} + y\vec{w}' \equiv_r \vec{0}\}
\]
(the image and this annihilator are dual to each other). Thus $P = N/r^2$.

\emph{Computing $N$.} Fix a coordinate $s$ with $w_s = \sigma \in \{\pm1\}$ and let
$b := w'_s \in \{0,\pm1\}$. The $s$-th coordinate of $x\vec{w} + y\vec{w}' \equiv \vec{0}$
forces $x \equiv -\sigma b y$. Substituting this into coordinate $t$ gives $m_t\, y \equiv_r 0$, where
\[
  m_t := w'_t - \sigma b\, w_t \in \{0, \pm1, \pm2\}
\]
(up to sign, $m_t$ is the $2\times 2$ minor of rows $\vec{w},\vec{w}'$ on columns $s,t$).
Therefore $N = \#\{y \in \Z_r : m_t y \equiv_r 0 \ \forall t\} = \gcd(g, r)$, where
$g := \gcd_t(m_t, r)$.

\emph{Case 1.} If $\vec{w}' \equiv \varepsilon \vec{w} \pmod r$ for some $\varepsilon \in \{\pm1\}$,
then $b \equiv \varepsilon\sigma$ and every $m_t \equiv 0$. Hence $g = r$, so $N = r$ and $P = 1/r$.

\emph{Case 2.} Suppose all $m_t \equiv 0 \pmod r$. If $r \ge 3$, then $m_t \in \{0,\pm1,\pm2\}$
forces $m_t = 0$ for all $t$, i.e.\ $\vec{w}' = \sigma b\, \vec{w}$ over $\Z$; since
$\vec{w}' \ne \vec{0}$; this gives $\vec{w}' = \pm\vec{w}$. If $r = 2$, it gives
$\vec{w}' \equiv \vec{w} \pmod 2$. Either way we are in Case 1. So in Case 2 some $m_t$ is
non-zero mod $r$, and since $|m_t| \le 2$, we get $g \in \{1,2\}$. Hence $N \in \{1,2\}$,
and $N = 2$ requires $2 \mid r$. Thus $P = N/r^2 \in \{1/r^2, 2/r^2\}$, with $P = 1/r^2$ for odd~$r$.
\end{proof}

\subsection{Proof of Claim~\ref{claim: rank_argument}}\label{sec: rank_argument}
    Let $\vec{w}$ and $\vec{x}$ are two non-zero binary vectors such that $\vec{w} \not\equiv_s \pm \vec{x}$.
    Consider the span of two vectors $\vec{w}$ and $\vec{x}$, say $S$. Any vector in $S$ is of the form
    $\alpha \vec{w} + \beta \vec{x}$. For this vector to be a $0/1$ vector, $\alpha, \beta \in \{0, 1\}$ and 
    $\alpha + \beta \in \{0, 1, 2\}$.
    Considering all the possibilities of $\alpha$ and $\beta$, there can be at most $3$ non-zero binary vectors in $S$.
    This completes the proof. 

\subsection{Proof of \cref{prop: counting-rank-t-matrices}} \label{sec: counting-rank-t-matrices}
Fix a matrix $\mathsf{M} \in \mathsf{M}_t$ of rank $t$.
Permute the columns so that the first $t$ columns
$M_1,\dots,M_t \in \{0,1\}^n$ are linearly independent over $\mathbb{Z}_s$.
Since these columns are independent, there exists a set of rows
$R \subseteq [n]$ with $|R| = t$ such that the $t\times t$ submatrix $A := \mathsf{M}[R,\{1,\dots,t\}]$ is invertible over $\mathbb{Z}_s$.

For each column index $i>t$, since $\rank(\mathsf{M})=t$, there exists a unique
coefficient vector $c_i \in \mathbb{Z}_s^t$ such that $
\begin{bmatrix}
M_1 & \cdots & M_t
\end{bmatrix} c_i
\;\equiv_s\;
M_i$. Restricting this identity to the rows in $R$ yields $
A\, c_i \;\equiv_s\; a_i$,
where $a_i := \mathsf{M}[R,i] \in \{0,1\}^t$.

Since $A \in \{0,1\}^{t\times t}$ is invertible, Hadamard’s inequality implies $|\det(A)| \le t^{t/2}$. Define the set
\[
S := \{ x \bmod s : x \in \mathbb{Z},\ |x| \le t^{t/2} \},
\]
so that $|S| = \BigO(t^{t/2})$ and $\det(A) \in S \setminus \{0\}$.
Hence $\det(A)^{-1} \in S^{-1}$, where $|S^{-1}| = |S|$. By Cramer’s rule, for each coordinate $j \in [t]$,
\[
(c_i)_j \;\equiv_s\; \det(A_{i,j}) \cdot \det(A)^{-1},
\]
where $A_{i,j}$ is the matrix obtained from $A$ by replacing its $j$-th column
with $a_i$. Since $A_{i,j} \in \{0,1\}^{t\times t}$, we also have
$\det(A_{i,j}) \in S$. Define the set $S' := S \cdot S^{-1} \subseteq \mathbb{Z}_s$.
Then $|S'| \le |S|^2 = \BigO(t^t)$, and each coordinate $(c_i)_j$ lies in $S'$.
Consequently, we get $c_i \in (S')^t$ and  $|(S')^t| = \BigO(t^{t^2})$.

To construct a rank-$t$ matrix, we have the following procedure:
\begin{itemize}
\item choose the $t$ independent columns in $\mathcal{O}(2^{tn})$ ways,
\item choose each of the remaining $m-t$ columns via a coefficient vector
      $c_i \in (S')^t$, contributing $\mathcal{O}((t^{t^2})^{m-t})$ choices,
\item and finally permute the columns, contributing a factor of $m!$.
\end{itemize}
Multiplying these bounds gives the claimed estimate.

\subsection{Proof of \cref{lemma: prob-counting-bad-vector-set}}
\label{sec: prob-counting-bad-vector-set}
Recall that $\mathcal{D}_{ij}$ (and similarly $\mathcal{D}_{lk}$) is a binary random variable equal to $1$ if and only if events
$\mathsf{E}_{\mathsf{LS},p}^{(i)}$,
$\mathsf{E}_{\mathsf{RS},p}^{(i)}$,
$\mathsf{E}_{\mathsf{LS},p}^{(j)}$,
$\mathsf{E}_{\mathsf{RS},p}^{(j)}$,
$\mathsf{E}_{\mathsf{TS},q}^{(i)}$, and
$\mathsf{E}_{\mathsf{TS},q}^{(j)}$
occur simultaneously (and analogously for indices $l,k$). Each $\mathsf{E}$ event corresponds to a linear equation.
In particular, the events
\[
\mathsf{E}_{\mathsf{LS},p}^{(i)},\,
\mathsf{E}_{\mathsf{RS},p}^{(i)},\,
\mathsf{E}_{\mathsf{LS},p}^{(j)},\,
\mathsf{E}_{\mathsf{RS},p}^{(j)},\,
\mathsf{E}_{\mathsf{LS},p}^{(l)},\,
\mathsf{E}_{\mathsf{RS},p}^{(l)},\,
\mathsf{E}_{\mathsf{LS},p}^{(k)},\,
\mathsf{E}_{\mathsf{RS},p}^{(k)}
\]
occur simultaneously if $\vec{a}^{(p)}$ satisfies
\begin{align*}
    &\vec{a}^{(p)}
    \begin{bmatrix}
        \mat{U}_{ijlk}^{(L)} & \mathbf{0} \\
        \mathbf{0} & \mat{U}_{ijlk}^{(R)}
    \end{bmatrix} \\
    &\equiv_p
    \begin{bmatrix}
        s_0 & s_0 & s_0 & s_0 & s_1 & s_1 & s_1 & s_1
    \end{bmatrix}.
\end{align*}

Writing $\vec{a}^{(p)} = \big(\vec{a}^{(p)}{}^{(L)}, \vec{a}^{(p)}{}^{(R)}\big)$,
this is equivalent to
\[
\vec{a}^{(p)}{}^{(L)} \, \mat{U}_{ijlk}^{(L)}
\equiv_p
\begin{bmatrix}
    s_0 & s_0 & s_0 & s_0
\end{bmatrix},
\]
and
\[
\vec{a}^{(p)}{}^{(R)} \, \mat{U}_{ijlk}^{(R)}
\equiv_p
\begin{bmatrix}
    s_1 & s_1 & s_1 & s_1
\end{bmatrix}.
\]

Similarly, the events
$\mathsf{E}_{\mathsf{TS},q}^{(i)}$,
$\mathsf{E}_{\mathsf{TS},q}^{(j)}$,
$\mathsf{E}_{\mathsf{TS},q}^{(l)}$,
$\mathsf{E}_{\mathsf{TS},q}^{(k)}$
occur simultaneously if $\vec{a}^{(q)}$ satisfies
\[
\vec{a}^{(q)} \, \mat{U}_{ijlk}
\equiv_q
\begin{bmatrix}
    s_2 & s_2 & s_2 & s_2
\end{bmatrix}.
\]

Since the above systems of linear equations are over
independent variables, we obtain
\[
\Pr_{\vec{a}^{(pq)}}\!\left(
\mathcal{D}_{ij} = 1 \wedge \mathcal{D}_{lk} = 1
\right)
\;\leq\;
\frac{1}{p^{t_1}\cdot p^{t_2}\cdot q^{t_3}}.
\]

\section{Additional Proofs for \cref{sec: average-case-integer}}

\subsection{Proof of \cref{lemma: pq_events_analysis}}
\label{sec: pq_events_analysis}
 For Property~(1), consider the following relation
    \begin{align*}
        S_{\mathsf{TS}, pq, i} = S_{\mathsf{C}, pq, i} \cdot pq + 
        S_{\mathsf{TS}, i}^{(pq)}.
    \end{align*}
    Using the law of total probability, we have the following
    \begin{align*}
        \Pr_{\vec{a}^{(pq)}}\left(S_{\mathsf{C}, pq, i} = c\right) = 
        \sum_{j = 0}^{pq - 1} 
        \Pr_{\vec{a}^{(pq)}}\left(S_{\mathsf{TS}, pq, i} = c\cdot pq + j\right) 
    \end{align*}
    Since $c = w_i/2 + \mathcal{O}(\sqrt{w_i})$, we have the following
    \begin{align*}
        s_j = c\cdot pq + j = (w_ipq)/2 + \mathcal{O}(pq\sqrt{w_i})
    \end{align*}
    for all $0 \leq j < pq$.
    
    Since $s_j$ is in the central window, using \cref{cor: LCLT_TS_pq}
    we get the desired result as follows
    \begin{equation*}
        \begin{aligned}
            &\Pr_{\vec{a}^{(pq)}}\left(S_{\mathsf{C}, pq, i} = c\right) \\
            &= \sum_{j = 0}^{pq - 1} 
            \Pr_{\vec{a}^{(pq)}}\left(S_{\mathsf{TS}, pq, i} = c\cdot pq + j = s_j\right),\\
            &= \sum_{j = 0}^{pq - 1} 
            \Theta\left(\frac{1}{\sigma_{pq} \sqrt{w_i}}\right),
            \\
            &= pq \cdot \Theta\left(\frac{1}{\sigma_{pq} \sqrt{w_i}}\right),\\
            &= \Theta\left(\frac{1}{\sqrt{w_i}}\right).
        \end{aligned}
    \end{equation*}
    This completes the proof of Property~(1).
    For Property~(2), we can observe that the following two events are the same
    \begin{align*}
        \left(S_{\mathsf{TS}, i}^{(pq)} = l, S_{\mathsf{C}, pq, i} = c\right),
    \end{align*}
    and 
    \[
    \left(S_{\mathsf{TS}, pq, i} = s = c \cdot pq + l\right).
    \]
    Hence, we have
    \begin{equation*}
        \begin{aligned}
            &\Pr_{\vec{a}^{(pq)}}\left(S_{\mathsf{TS}, pq, i} = c\cdot pq + l\right) \\
            &= 
            \Pr_{\vec{a}^{(pq)}}\left(S_{\mathsf{TS}, i}^{(pq)} = l, S_{\mathsf{C}, pq, i} = c\right),\\
            &= \Pr_{\vec{a}^{(pq)}}\left(S_{\mathsf{TS}, i}^{(pq)} = l 
            \; \middle| \; S_{\mathsf{C}, pq, i} = c\right) \cdot \Pr_{\vec{a}^{(pq)}}\left(S_{\mathsf{C}, pq, i} = c\right).
        \end{aligned}
    \end{equation*}
    Using 
    \cref{cor: LCLT_TS_pq} 
    and part $[1]$ of this lemma in the above equation, 
    we get the desired result
    \begin{equation*}
        \begin{aligned}
            &\Pr_{\vec{a}^{(pq)}}\left(S_{\mathsf{TS}, i}^{(pq)} = l 
            \; \middle| \; S_{\mathsf{C}, pq, i} = c\right) \\ 
            &= 
            \frac{\Pr_{\vec{a}^{(pq)}}\left(S_{\mathsf{TS}, pq, i} = c\cdot pq + l\right)}
            {\Pr_{\vec{a}^{(pq)}}\left(S_{\mathsf{C}, pq, i} = c\right)},\\
            &= \frac{\Theta\left(\frac{1}{pq\sqrt{w_i}}\right)}
            {\Theta\left(\frac{1}{\sqrt{w_i}}\right)},\\
            &= \Theta\left(\frac{1}{pq}\right). 
        \end{aligned}
    \end{equation*}

    It remains to prove Property (3). First, we will split the total sum into the left and right parts
    \begin{align*}
        S_{\mathsf{TS}, pq, i} &= S_{\mathsf{LS}, pq, i} + S_{\mathsf{RS}, pq, i}, \\
        &= \vec{a}^{(pq)}{}^{(L)}\cdot \vec{u}_i^{(L)} + 
        \vec{a}^{(pq)}{}^{(R)}\cdot \vec{u}_i^{(R)},
    \end{align*}
    and let $w_{l, i}$ and $w_{r, i}$ be the Hamming weights of the vectors $\vec{u}_i^{(L)}$
    and $\vec{u}_i^{(R)}$ respectively. Then, we have $w_i = w_{l, i} + w_{r, i}$.
    Now, we divide the left sum $S_{\mathsf{LS}, pq, i}$ by $p$
    \begin{align*}
        S_{\mathsf{LS}, pq, i} = Q_{\mathsf{LS}, i}^{(p)} \cdot p + S_{\mathsf{LS}. i}^{(p)}.
    \end{align*}
    Then, we have the following
    \begin{equation*}
        \begin{aligned}
            \Pr_{\vec{a}^{(pq)}{}^{(L)}}\left(S_{\mathsf{LS}, i}^{(p)} = k\right)
            = \sum_{y \equiv_p k}  \Pr_{\vec{a}^{(pq)}{}^{(L)}}
            \left(S_{\mathsf{LS}, pq, i} = y\right).
        \end{aligned}
    \end{equation*}
    We aim to analyze the following event
    \begin{align*}
        \left(S_{\mathsf{LS}, i}^{(p)} = k, 
            S_{\mathsf{TS}, i}^{(pq)} = l, 
            S_{\mathsf{C}, pq, i} = c\right)
    \end{align*}
    To do so, we split this event into the left and right sums as follows
    \begin{equation*}
        \begin{aligned}
            &\Pr_{\vec{a}^{(pq)}}\left(
            S_{\mathsf{LS}, i}^{(p)} = k, 
            S_{\mathsf{TS}, i}^{(pq)} = l, 
            S_{\mathsf{C}, pq, i} = c\right) \\
            &= 
            \Pr_{\vec{a}^{(pq)}}\left(
            S_{\mathsf{LS}, pq, i} \equiv k \mod p, 
            S_{\mathsf{TS}, pq, i} = s = c \cdot pq + l\right),\\
            &= \sum_{y \equiv_p k}  
            \Pr_{\vec{a}^{(pq)}}\left(
            S_{\mathsf{LS}, pq, i}  = y, 
            S_{\mathsf{RS}, pq, i} = s - y\right).
        \end{aligned}
    \end{equation*}
    Now, let us define the following function
    \begin{align*}
        f_{pq, i}(y, s) := \Pr_{\vec{a}^{(pq)}}\left(
            S_{\mathsf{LS}, pq, i}  = y, 
            S_{\mathsf{RS}, pq, i} = s - y\right).
    \end{align*}
    Since $s$ is fixed in our analysis, we will omit it and denote the above function by 
    $f_{pq, i}(y)$.
    Firstly, we can observe the following
    \begin{equation*}
        \begin{aligned}
            \sum_{y \equiv_p k}  f_{pq, i}(y) &= \sum_{y \equiv_p k}
            \Pr_{\vec{a}^{(pq)}}\left(S_{\mathsf{LS}, pq, i}  = y, 
            S_{\mathsf{RS}, pq, i} = s - y\right)\\
            &= \Pr_{\vec{a}^{(pq)}}\left(
                S_{\mathsf{LS}, i}^{(p)} = k, 
                S_{\mathsf{TS}, i}^{(pq)} = l, 
                S_{\mathsf{C}, pq, i} = c\right),
        \end{aligned}
    \end{equation*}
    and 
    \begin{equation*}
        \begin{aligned}
            \sum_{y}  f_{pq, i}(y) &= \sum_{y}
            \Pr_{\vec{a}^{(pq)}}\left(S_{\mathsf{LS}, pq, i}  = y, 
            S_{\mathsf{RS}, pq, i} = s - y\right)\\
            &= \Pr_{\vec{a}^{(pq)}}\left(
                S_{\mathsf{TS}, pq, i} = s = c \cdot pq + l\right).
        \end{aligned}
    \end{equation*}
    Now, we aim to compare the above two sums. To do so, we will prove a Gaussian smoothness result for the function $f_{pq, i}(y)$. 
    
    Let $W_{pq}^{(L)} := (w_{l, i}pq)/2 + 
    \mathcal{O}(pq\sqrt{w_{l, i}})$ is the central Gaussian window of the 
    random variable $S_{\mathsf{LS}, pq, i}$. Since $s = c\cdot pq + l$, where 
    $c = w_i/2 + \mathcal{O}(\sqrt{w_i})$ and $0 \leq l < pq$, thecrefore 
    $s = (w_ipq)/2 + \mathcal{O}(pq\sqrt{w_i})$ and we have the following
    $$s - W_{pq}^{(L)} = (w_{r, i}pq)/2 + \mathcal{O}(pq\sqrt{w_{r, i}}).$$
    Then, using \cref{cor: Theta_LCLT} for the left and right sums, for any $y = W_{pq}^{(L)}$ and 
    $y' = s - W_{pq}^{(L)}$ we have
    \begin{align*}
        \Pr_{\vec{a}^{(pq)}} \left(S_{\mathsf{LS}, pq, i} = y\right) &= 
        \Theta\left(\frac{1}{\sigma_{pq}\sqrt{w_{l, i}}}\right),\\
        \Pr_{\vec{a}^{(pq)}} \left(S_{\mathsf{RS}, pq, i} = y' = s - y\right) &= 
        \Theta\left(\frac{1}{\sigma_{pq} \sqrt{w_{r, i}}}\right).
    \end{align*}

    \begin{claim} \label{claim: f_y_smoothness}
        For any $y = W_{pq}^{(L)}$ and $|\Delta| \leq pq$, we have
        \begin{align}
            f_{pq, i}(y + \Delta) = \Theta(f_{pq, i}(y)).
        \end{align}
    \end{claim}
    \begin{proof}
        Since the standard deviations of both random variables
        $S_{\mathsf{LS}, pq, i}$ and $S_{\mathsf{RS}, pq, i}$ are
        $\Theta(pq\sqrt{n})$ and 
       $\Delta \leq pq$, a shift of magnitude $\Delta$ keeps us within the central region of the distribution. Consequently, the density function
       $f_{pq, i}(y)$ varies smoothly under a displacement of size
       $\Delta$ and its value remains stable up to constant factors when evaluated at $f_{pq, i}(y + \Delta)$.
    \end{proof}

    Now, for any $k \in \mathbb{Z}_{p}$, let us define a set $A_k$ that contains elements 
    congruent to $k$ modulo $p$ in the window $W_{pq}^{(L)}$ as follows
    \begin{align*}
        A_k := \{y \; | \; y = W_{pq}^{(L)},  y \equiv k \mod p\}.
    \end{align*}
    Since almost the entire mass of a Gaussian is in the central window,
    we have the following
    \begin{align}\label{eq: f_y_Theta_bound_central_window}
        \sum_{y} f_{pq, i}(y) = 
        \Theta\left(\sum_{y = W_{pq}^{(L)}} f_{pq, i}(y)\right).
    \end{align}

    Now, using the above equation for any $k \in [0, p)$, we have the
    following
    \begin{align}\label{eq: f_y_Theta_bound_A_k}
        \sum_{y \equiv k \mod p} f_{pq, i}(y) = \Theta\left(\sum_{y \in A_k} f_{pq, i}(y)\right).
    \end{align}

    Now, we combine the above two observations and get the following
    averaging argument
    \begin{claim}\label{claim: f_y_average}
        For any $k \in [0, pq)$, we have
        \begin{align*}
            \sum_{y = W_{pq}^{(L)}}  f_{pq, i}(y) = p \cdot \Theta\left(\sum_{y \in A_k}  f_{pq, i}(y)\right).
        \end{align*}
    \end{claim}
    \begin{proof}
        Consider $k' \in [0, pq)$ and the corresponding set $A_{k'}$. Then, for any element 
        in $A_{k'}$, say $y' = k' + mp$, there exists a unique element in $A_k$, $y = k + mp$, such that 
        $|y' - y| = |k' - k| < p$. Let $\Delta_{k'} := k' - k$. Using this observation, we get the following
        \begin{equation*}
            \begin{aligned}
                \sum_{y' \in A_{k'}}  f_{pq, i}(y') &= \sum_{y' \in A_{k'}} f_{pq, i}(k' + kp) \\
                &= \sum_{y \in A_k}  f_{pq, i}(y + \Delta_{k'})\\
                &= \Theta\left(\sum_{y \in A_k} f_{pq, i}(y)\right).
            \end{aligned}
        \end{equation*}
        This implies the $\sum_{y' \in A_{k'}} f_{pq, i}(y')$ corresponding to any residue class $k'$ is in the neighborhood of $\sum_{y \in A_{k}} f_{pq, i}(y)$. So, taking the average will give us the desired result.
        \begin{equation*}
            \begin{aligned}
                \sum_{y} f_{pq, i}(y) &= \sum_{k' = 0}^{p-1} \sum_{y \in A_{k'}} f_{pq, i}(y)
                = p \cdot \Theta\left(\sum_{y \in A_k} f_{pq, i}(y)\right).
            \end{aligned}
        \end{equation*}
    \end{proof}

    Now, we have all the tools to compute the desired result. Consider
    the following
    \begin{equation*}
        \begin{aligned}
            &\Pr_{\vec{a}^{(pq)}}\left(
            S_{\mathsf{LS}, i}^{(p)} = k, 
            S_{\mathsf{TS}, i}^{(pq)} = l, 
            S_{\mathsf{C}, pq, i} = c\right) \\
            &= \sum_{y \equiv_p k}  f_{pq, i}(y),\\
            &= \Theta\left(\sum_{y \in A_k}  f_{pq, i}(y)\right)  \qquad 
            \qquad (\text{Using \ref{eq: f_y_Theta_bound_A_k}}),\\
            &= \frac{1}{p} \cdot \Theta\left(\sum_{y = W_{pq}^{(L)}} f_{pq, i}(y)\right)
            \qquad (\text{Using Claim~\ref{claim: f_y_average}}),\\
            &= \frac{1}{p} \cdot \Theta\left(
            \sum_{y} f_{pq, i}(y)\right) 
            \qquad \quad (\text{Using \ref{eq: f_y_Theta_bound_central_window}}),\\
            &= \frac{1}{p} \cdot \Theta\left(\Pr_{\vec{a}^{(pq)}}\left(
            S_{\mathsf{TS}, pq, i}  = s = c \cdot pq + l\right) \right)\\
            &= \frac{1}{p} \cdot \Theta\left(\frac{1}{pq\sqrt{w_i}}\right)
            \qquad \qquad (\text{Using \cref{cor: LCLT_TS_pq}}).
        \end{aligned}
    \end{equation*}
This completes the proof of \cref{lemma: pq_events_analysis}.
\subsection{Discussion on $M = pqr + s$}\label{sec: discussion_M}
Firstly, we divide $M$ by the product of integers $p$ and $q$. Let $r$ and $s$ be unique integers such that
\begin{align*}
    M = r\cdot pq + s,
\end{align*}
and assume that $s \neq 0$.
Let $\vec{a}$ be a vector uniformly distributed over the set $[M]^n = \{1, 2, \ldots, M\}^n$.
Let us split this random vector, $\vec{a}$, into the following random variables
\begin{align*}
    a_i = a_i^{(pq)} + a_i^{(r)}\cdot pq,
\end{align*}
where $a_i^{(pq)} \in [1, pq)$ and $a_i^{(r)} = [0, r]$ $\forall$ $1 \leq i \leq n$.

Now, we look at the distribution of the random variable vectors 
$\vec{a}^{(pq)}$ and $\vec{a}^{(r)}$ over the random choice of $\vec{a}$.

\paragraph{Distribution of $a_i^{(pq)}$.} 
\begin{align}\label{eq: a_i_pq_dist}
    \Pr_{a_i}\left(a_i^{(pq)} = x\right) = 
    \begin{cases}
        \frac{r+1}{M} & 1 \leq x \leq s\\
        \frac{r}{M} & x > s.\\
    \end{cases}
\end{align}

\paragraph{Distribution of $a_i^{(r)}$.} 
\begin{align}\label{eq: a_i_r_dist}
    \Pr_{a_i}\left(a_i^{(r)} = y\right) = 
    \begin{cases}
        \frac{pq}{M} & 0 \leq y < r\\
        \frac{s}{M} & y = r\\
    \end{cases}
\end{align}

\paragraph{Joint Distribution of $a_i^{(pq)}$ and $a_i^{(r)}$.} 
\begin{align*}
    &\Pr_{a_i}\left(a_i^{(pq)} = x, a_i^{(r)} = y\right) \\ &= 
    \begin{cases}
        \frac{1}{M} & 1 \leq x \leq pq, 0 \leq y < r\\
        \frac{1}{M} & 1 \leq x \leq  s, y = r\\
        0 & y = r, x > s.
    \end{cases}
\end{align*}

\paragraph{Dependence between $a_i^{(pq)}$ and $a_i^{(r)}$.} 
Let us first make the following observations after looking at the above distributions
\begin{enumerate}
    \item If $s = 0$, it is easy to observe that the two random variables $a_i^{(pq)}$ and $a_i^{(r)}$ are independent
    and uniformly distributed over the sets $[0, pq)$ and $[0, r)$ respectively.
    \item We also observe that the main dependence occurs when $y = r$, since in this case, residues can only take $s$ values rather than $pq$ values, where $0 \leq s < pq$.
\end{enumerate}
Keeping these observations in mind, we split the probability mass of the distribution of the 
random variable $a_i$ into two cases. First, the inner rectangle, i.e. 
$1 \leq a_i \leq pqr$, and second the truncated slice, i.e. $pqr < a_i \leq M$.
Let $\mathsf{E_i^{(U)}}$ and $\mathsf{E_i^{(T)}}$ be two disjoint exhaustive events that denote these
two cases, respectively. Then, we have
    \begin{align*}
        \Pr_{a_i}\left(\mathsf{E_i^{(T)}}\right) &= \frac{s}{M},\\
        \Pr_{a_i}\left(\mathsf{E_i^{(U)}}\right) &= 1 - \frac{s}{M}.
    \end{align*}
Since $(s)/M < pq/M = \mathcal{O}\left(1/2^{n/2}\right)$, we can ignore this event and focus only on the event $\mathsf{E_i^{(U)}}$.
Let $\mathsf{E^{(T)}} := \lor_{i = 1}^{n} \mathsf{E_i^{(T)}}$, then using union bound, we have
    \begin{align*}
        \Pr_{a_i}\left(\mathsf{E^{(T)}}\right) \leq \frac{ns}{M} = \mathcal{O}\left(\frac{n}{2^n}\right).
    \end{align*}
Hence, we have
\begin{align*}
    \Pr_{a_i}\left(\mathsf{E^{(U)}}\right) \geq 1 - \mathcal{O}\left(\frac{n}{2^n}\right).
\end{align*}

\begin{proposition}\label{prop: a_pq_a_r_vector_dependency}
    Let $a_i \xleftarrow{\$} [M]$ $\forall$ $1 \leq i \leq n$, where $M = r.pq + s$. Then, we have the following distributions
    for all $1 \leq i \leq n$
    \begin{enumerate}
        \item For all $0 \leq x < pq$, we have
        \begin{align*}
            \Pr_{a_i}\left(a_i^{(pq)} = x \; \middle | \; \mathsf{E_i^{(U)}}\right) = 
            \begin{cases}
                \frac{r-1}{M} & x = 0\\
                \frac{r}{M} & 0 < x < pq
            \end{cases}.
        \end{align*}
        \item For all $0 \leq y < r$, we have
        \begin{align*} 
            \Pr_{a_i}\left(a_i^{(r)} = y \; \middle | \; \mathsf{E_i^{(U)}}\right) =
            \begin{cases}
                \frac{pq-1}{M} & y = 0\\
                \frac{pq}{M} & 0 < y < r
            \end{cases}.
        \end{align*}
        \item 
        \small{
        \begin{align*}
            &\Pr_{\vec{a}}\left(\vec{a}^{(pq)} = \vec{x}, \vec{a}^{(r)} = \vec{y}\; \middle | \; \mathsf{E^{(U)}}\right) \\
            &= 
            \Pr_{\vec{a}}\left(\vec{a}^{(pq)} = \vec{x}\; \middle| \; \mathsf{E^{(U)}}\right) 
            \Pr_{\vec{a}}\left(\vec{a}^{(r)} = \vec{y}\; \middle| \; \mathsf{E^{(U)}}\right)\left(1 + \mathcal{O}\left(\frac{npq}{M}\right)\right),
        \end{align*}}
        where $\mathsf{E^{(U)}} := \wedge_{j = 1}^{n} \mathsf{E_j^{(U)}}$.
    \end{enumerate}
\end{proposition}
\begin{proof}
The proof of parts $[1]$ and $[2]$ comes directly from equations $(\ref{eq: a_i_pq_dist})$ and $(\ref{eq: a_i_r_dist})$. Now, 
we look at the proof of part $[3]$. Let us define
    \begin{align*}
        \mathsf{P_i}(x, y) &:= \Pr_{a_i}\left(a_i^{(pq)} = x, a_i^{(r)} = y \; \middle| \; \mathsf{E_i^{(U)}}\right)\\
        \mathsf{Q_i}(x, y) &:= \Pr_{a_i}\left(a_i^{(pq)} = x\; \middle| \; \mathsf{E_i^{(U)}}\right) 
        \Pr_{a_i}\left(a_i^{(r)} = y\; \middle| \; \mathsf{E_i^{(U)}}\right)
    \end{align*}

We aim to analyze the ratio of $\mathsf{P_i}(x, y)$ and $\mathsf{Q_i}(x, y)$, denoted by $\mathsf{R_i}(x, y)$. 
This ratio will help us analyze the dependence between the two variables.
\begin{claim} For all $0 \leq x < pq$ and $0 \leq y < r$ such that $(x, y) \neq 0$, we have the following
    \begin{align*}
        \Omega\left(\frac{1}{M}\right) = \mathsf{R_i}(x, y) - 1 = \mathcal{O}\left(\frac{pq}{M}\right).
    \end{align*}
\end{claim}
\begin{claimproof}
    \begin{enumerate}
        \item  $x = 0$, $0 < y < r$. In this case, we have 
        \begin{align*}
            \mathsf{P_i}(x, y) = \frac{1}{M},  \qquad \mathsf{Q_i}(x, y) = \frac{r-1}{M} \cdot \frac{pq}{M}
        \end{align*}
        Hence, 
            \begin{align*}
                \mathsf{R_i}(x, y) &= \frac{\frac{1}{M}}{\frac{r-1}{M} \cdot \frac{pq}{M}}\\
                &= 1 + \frac{pq - s}{pq(r-1)}.
            \end{align*}
        Since $\Omega(1/M) = \frac{pq - s}{pq(r-1)} =  \mathcal{O}\left(pq/M\right)$, we get the desired result. 
        \item $x > 0$, $0 \leq y < r$. In this case, we have 
        \begin{align*}
            \mathsf{P_i}(x, y) = \frac{1}{M},  \qquad \mathsf{Q_i}(x, y) = \frac{r}{M} \cdot \frac{pq}{M}
        \end{align*}
        Hence,
        \begin{align*}
            \mathsf{R_i}(x, y) = 1 + \frac{s}{M}.
        \end{align*}
        Since $\frac{1}{M} \leq 1 + \frac{s}{M} < 1 + \frac{pq}{M}$, we get the desired result. 
    \end{enumerate}
    This completes the proof of the claim.
\end{claimproof}

Now, consider the following 
    {\small
    \begin{align*}
        &\Pr_{\vec{a}}\left(\vec{a}^{(pq)} = \vec{x}, \vec{a}^{(r)} = \vec{y}\; \middle | \; \mathsf{E^{(U)}}\right) \\
        &= \prod_{i = 1}^{n} \Pr_{a_i}\left(a_i^{(pq)} = x_i, a_i^{(r)} = y_i\; \middle | \; \mathsf{E_i^{(U)}}\right)\\
        &= \Pr_{\vec{a}}\left(\vec{a}^{(pq)} = \vec{x}\; \middle| \; \mathsf{E^{(U)}}\right) 
        \Pr_{\vec{a}}\left(\vec{a}^{(r)} = \vec{y}\; \middle| \; \mathsf{E^{(U)}}\right)\left(1 + \mathcal{O}\left(\frac{npq}{M}\right)\right).
    \end{align*}}
This completes the proof.
\end{proof}

\subsection{Proof of \cref{prop: big_weight_vec_count}}\label{sec: big_weight_vec_count}
Consider a set $R_{\alpha, c} := \{(a, b) \in [m]^2\; | \; a, b \geq \alpha m, a+b \leq m, |a - b| \leq c\sqrt{m}\}$. Let $S$ denote 
the above-mentioned sum in the problem. We can rewrite it as follows
\begin{equation*}
    \begin{aligned} 
        S &= \sum_{t_1 = \alpha m}^{m} \sum_{t_2 = \alpha m}^{m - t_1} \binom{m}{t_1} \binom{m - t_1}{t_2} 2^{m - (t_1 + t_2)}\\
        &= \sum_{(t_1, t_2) \in R_{\alpha, c}}
        \frac{m!}{t_1 ! t_2 ! t_3 !}1^{t_1} \cdot 1^{t_2} \cdot 2^{t_3}\\
        &= 4^m \cdot \sum_{(t_1, t_2) \in R_{\alpha, c}}
        \frac{m!}{t_1 ! t_2 ! t_3 !}
        \left(\frac{1}{4}\right)^{t_1} \cdot 
        \left(\frac{1}{4}\right)^{t_2} \cdot 
        \left(\frac{1}{2}\right)^{t_3}.
    \end{aligned}
\end{equation*}
Let $X_1, X_2$, and $X_3$ be three random variables that follows multinomial distribution
with $m$ repetitions and probabilities $1/4$, $1/4$, and $1/2$ respectively. 
Then, using the p.m.f. of the multinomial distribution in the above equation,
we have the following
\begin{align*}
    S = 4^m \cdot \Pr\left((X_1, X_2) \in R_{\alpha, c}\right).
\end{align*}
Since, $\mathbb{E}\left(X_1\right) = \mathbb{E}\left(X_2\right)
= m/4$, $\mathsf{Var}(X_1) = \mathsf{Var}(X_2) = 3m/16$ and $\alpha < 1/4$, 
using Chebyshev inequality, we get
\begin{align*}
    \Pr\left(X_1 < \alpha m, X_2 < \alpha m\right) = \mathcal{O}\left(\frac{1}{m}\right).
\end{align*}
Now, let $D = X_1 - X_2$, then we have $\mathbb{E}\left(D\right) = 0$ and $\mathsf{Var}(D) = 3m/16$.
Then, again by using Chebyshev inequality, we have 
\begin{align*}
    \Pr(|D| > c\sqrt{m}) \leq \frac{1}{2c^2}.
\end{align*}
Then, using the union bound, we get the desired result. 

\bibliographystyle{alpha}
\bibliography{reference}
\end{document}